\pdfoutput=1

\documentclass[cmp,final]{svjour}

\usepackage[utf8]{inputenc}
\usepackage{graphicx}
\usepackage{amsmath}
\usepackage{amssymb}

\providecommand*{\bra}[1]{\mathinner{\langle #1\rvert}}
\providecommand*{\ket}[1]{\mathinner{\lvert #1\rangle}}

\providecommand*{\abs}[1]{\mathinner{\lvert #1\rvert}}
\providecommand*{\norm}[1]{\lVert #1\rVert}
\providecommand*{\gen}[1]{\langle #1\rangle}
\providecommand*{\braket}[2]{\langle #1\vert #2\rangle}
\providecommand*{\eval}[2]{\langle #1,#2\rangle}
\providecommand*{\bigeval}[2]{\bigl\langle #1,#2\bigr\rangle}
\providecommand*{\relphantom}[1]{\mathrel{\phantom{#1}}}
\DeclareMathOperator{\tr}{tr}
\DeclareMathOperator{\id}{id}
\DeclareMathOperator{\ttr}{ttr}

\DeclareMathOperator{\ad}{ad}

\DeclareMathOperator{\Cocom}{Cocom}

\begin{document}

\title{A hierarchy of topological tensor network states}

\author{Oliver Buerschaper\inst{1} \and
        Juan Martín Mombelli\inst{2} \and
        Matthias Christandl\inst{3,4}\and
        Miguel Aguado\inst{1}}

\institute{Max-Planck-Institut für Quantenoptik,
           Hans-Kopfermann-Straße 1, 85748 Garching, Germany \and
           Facultad de Matemática, Astronomía y Física,
           Universidad Nacional de Córdoba, Medina Allende s/n,
           Ciudad Universitaria, 5000 Córdoba, Argentina \and
           Institute for Theoretical Physics, ETH Zurich, 8093 Zurich,
           Switzerland \and
           Fakultät für Physik, Ludwig-Maximilians-Universität München,
           Theresienstraße 37, 80333 Munich, Germany}

\date{Received: \today / Accepted: date}

\communicated{Unknown}

\maketitle

\begin{abstract}
    We present a hierarchy of quantum many-body states among which many
    examples of topological order can be identified by construction.
    We define these states in terms of a general, basis-independent framework of tensor networks based on
    the algebraic setting of finite-dimensional Hopf $C^*$-algebras.
    At the top of the hierarchy we identify
    ground states of new topological lattice models extending Kitaev's quantum double models~\cite{Kitaev:2003}.
    For these states we exhibit the mechanism
    responsible for their non-zero topological entanglement
    entropy by constructing a renormalization group flow.
    Furthermore it is shown that
    those states of the hierarchy
    associated with Kitaev's original quantum double models are related
    to each other by the condensation of topological charges.
    We conjecture that charge condensation is
    the physical mechanism underlying the hierarchy in general.
\end{abstract}

\section{Introduction}
\label{sec:intro}

\subsection{Background}

Quantum states of many-body systems, for instance lattice models of quantum
spins, furnish a rich physical arena where simplified models of
condensed matter phenomena such as high-$T_c$ superconductivity
can be studied~\cite{Anderson:1987p2991}. Such quantum systems are, however, difficult to solve numerically, since
the dimension of the corresponding Hilbert space grows exponentially in
the number of components.

Two promising developments tackling this difficulty have occurred in
recent times. Firstly, the advent of experimental settings
where lattice models can be implemented whose parameters can be
controlled very accurately~\cite{Bloch:2008p2992}. Secondly -- and relevant for the
present article -- the understanding of the entanglement structure of
models with local Hamiltonians. It is believed that ground states of
such models generically obey an area law for the entanglement
entropy~\cite{Vidal:2003p2577,Eisert:2010p2578}. This has justified the
development of a number of approaches where the physical states of the
model can be described by so-called tensor networks whose set of
parameters grows only polynomially in the number of components and which
are therefore suitable for
numerical simulations. Among these are the tensor product states ({\sc TPS})~\cite{Nishino:1995}, also known as projected
entangled-pair states ({\sc PEPS})~\cite{Verstraete:2004p114},
which generalize the structure of matrix product
states ({\sc MPS})~\cite{Fannes:1992p1590} to higher dimensions, and
the multi-scale entanglement renormalisation ansatz ({\sc MERA}) based on
renormalisation group
ideas~\cite{Vidal:2007p2581,Vidal:2008p2582}.

In 2003 Kitaev proposed topological quantum
computation~\cite{Kitaev:2003}: this is an approach to quantum computation which is based on quantum many-body systems exhibiting topological order, i.e.
systems that are effectively described by a topological quantum field
theory~\cite{Witten:1988p2583}. Topological phases appear in condensed matter
physics, for instance in the fractional quantum Hall effect~\cite{Wen:1995p1894}; they
can also be studied in the context of quantum field theories~\cite{Asorey:1993p2994}. In Kitaev's proposal information is encoded
in nontrivial loops of the underlying surface~\cite{Dennis:2002} and is
processed by the creation, braiding and annihilation of topological
charges~\cite{Freedman:2002,Freedman:2002b}. The simplest models
proposed for this purpose are Kitaev's original quantum double models,
which are based on discrete gauge theories~\cite{Kitaev:2003}, and Levin
and Wen's string-net models~\cite{Levin:2004p117}. The latter are
believed to describe all parity and time-reversal invariant topological phases on the lattice and are constructed
from the idea that topological phases correspond to fixed points of
renormalisation group procedures. In Kitaev's models, the space of
superselection sectors is the set of irreducible representations of the
quasitriangular Hopf algebra~$\mathrm{D}(\mathbb{C}G)$, Drinfeld's
quantum double~\cite{Drinfeld:1986} of the group algebra~$\mathbb{C}G$ describing the gauge
symmetry. We will therefore call these models
$\mathrm{D}(\mathbb{C}G)$-models. The simplest such model is the well-known
toric code which is based on $G=\mathbb{Z}_2$.

Recently, tensor networks for topological lattice models have been introduced
and it is this development which forms the starting point for the current work.
In~\cite{Verstraete:2006p113} a {\sc PEPS} ansatz is devised
for a distinguished ground state of the toric code. Ground states of general quantum
double models based on a group algebra can be described as
{\sc MERA} states~\cite{Aguado:2008p518}, as can ground states of
string-net models~\cite{Konig:2009p1839}. Other tensor networks for
string-nets are given in~\cite{Buerschaper:2009p1923,Gu:2009p1885}.

\subsection{Aims}

In the following we will outline the two main goals of this article: the extension of Kitaev's quantum double models and the derivation of tensor network representations for these generalized models.

The first aim of this article is to extend Kitaev's
quantum double models~\cite{Kitaev:2003} from the
case of the group algebra of a finite group~$G$ to a finite-dimensional Hopf algebra
with certain properties, as anticipated by Kitaev in~\cite{Kitaev:2003}.
This sheds more light on the Hopf algebra point
of view Kitaev outlined originally, in particular, it clarifies how the
structure maps of the Hopf algebra enter in the definition of the model
and how the quantum double of the Hopf algebra arises in the description
of the model's excitations. It turns out that an
involutory Hopf algebra (i.e. a Hopf algebra whose antipode map squares to the
identity) with an additional $*$-algebra structure is sufficient in
order to define the generalized quantum double model properly. As we will show both requirements
are satisfied if one chooses a finite-dimensional Hopf $C^*$-algebra (also known as
a finite-dimensional Kac algebra).
In order to illustrate the construction of our generalized model we give
an example of a finite-dimensional Hopf $C^*$-algebra which is nontrivial in the
sense that it is neither a group algebra nor the dual of a group
algebra. This extends well-known instances of the family like the toric
code for $G=\mathbb{Z}_2$ or the model based on $G=S_3$ which is
universal for topological quantum computation~\cite{Mochon:2004p2296}.

The second aim of the present article is to derive tensor network
representations for the ground states of these models.
As such representations allow to map out the phase
diagram of the model and determine phase transitions between topologically
and conventionally ordered phases it is crucial to understand the
properties of tensor networks describing states of the model deep within
the different phases.
In order to derive these tensor network representations for quantum
double model ground states we introduce a novel diagrammatic technique.
It turns out that the tensor networks can be formulated
basis-independently and that furthermore a single distinguished element
of the Hopf $C^*$-algebra, namely its unique Haar integral, plays the
key role in the construction. Other than that only the structure maps of
the Hopf $C^*$-algebra are employed in the definition of the tensor
network states.

In addition to the stated aims above we are able to extend these ground states to a hierarchy of tensor
network states for each finite-dimensional Hopf $C^*$-algebra. This hierarchy of
states is characterized by different values of the topological
entanglement entropy~$\gamma$ as defined in~\cite{Kitaev:2006,Levin:2006p288}, and hence these states represent different
instances of topological order, in other words different unitary modular
tensor categories ({\sc UMTC}s). The way this hierarchy depends on the
Hopf subalgebras of the original Hopf $C^*$-algebra points towards
condensation of topological charges~\cite{Bais:2002p1647,Bais:2003p1648}. Furthermore we show how different isomorphism
classes of finite-dimensional Hopf $C^*$-algebras exhibit different mechanisms for
the non-zero $\gamma$. In particular, we explain how the boundary
configurations of a region differ between models based on a group
algebra, the dual of a group algebra and non-trivial finite-dimensional Hopf
$C^*$-algebras.

\subsection{$\mathrm{D}(\mathbb{C}G)$-models and their tensor network
    representations}

We will now review the basics of Kitaev's quantum double model and its tensor network
representations in a language appropriate for our later discussion.

The quantum double model presented in~\cite{Kitaev:2003} is a
quantum spin model defined on any finite oriented graph $\Gamma=(V,E,F)$ with
vertices~$V$, edges~$E$ and faces~$F$ which can be embedded in a
two-dimensional surface. Throughout the article we will denote the set of all
edges connected to a vertex $s\in V$ by $E(s)$ and similarly $E(p)$ will stand for the set of edges that form the boundary of the face $p\in F$. In order to construct
the Hilbert space of the quantum double model one assigns to each oriented edge
the local Hilbert space~$\mathbb{C}G$ whose canonical basis is
$\{\ket{g}\mid g\in G\}$. Reversing the orientation of an edge
corresponds to the map $\ket{g}\mapsto\ket{g^{-1}}$. The Hamiltonian is
assembled from two sets of operators that act on vertices or faces of $\Gamma$
respectively. At each vertex $s\in V$ one may orient the adjacent edges
such that they point inwards. Then the vertex operator~$A(s)$, which acts on
$E(s)$, is given by
\begin{equation}
    \label{eq:vertex_group}
    A(s)=\frac{1}{\abs{G}}
         \sum_{g\in G}
         L(g)\otimes
         \dots\otimes
         L(g)
\end{equation}
where $L(g)\colon\mathbb{C}G\to\mathbb{C}G$, $\ket{h}\mapsto\ket{gh}$
denotes the left regular representation of $G$. Hence the action of
$A(s)$ is a simultaneous left multiplication at each incident edge
averaged over the group~$G$. In fact, it projects onto the trivial
representation. Using the orientation reversal isomorphism one may
orient the edges around a face $p\in F$ in counterclockwise fashion.
One then defines the face operator~$B(p)$ by its action on basis
states $\ket{g_1,\dots,g_r}$:
\begin{equation}
    \label{eq:face_group}
    B(p)\
    \vcenter{\hbox{\includegraphics{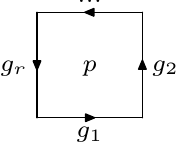}}}
    =\delta_{g_r\cdots g_1,e}\
     \vcenter{\hbox{\includegraphics{lattices2-14}}}\ .
\end{equation}
Again, this is a
projector which projects onto configurations with a trivial product
of group elements around the face. Finally, the Hamiltonian reads:
\begin{equation}
    \mathcal{H}
    =-\sum_{s\in V}
     A(s)
     -\sum_{p\in F}
     B(p).
\end{equation}
This is a sum of mutually commuting projectors and therefore exactly
solvable.

\bigskip

\noindent
In order to find a tensor network representation for
particular states occurring in the quantum double model one can start
from its inherent (group) symmetries. We can assign a
representation~$\rho_p$ of $G$ to each face $p\in F$ and associate the
tensor
\begin{equation}
    \label{eq:group_tensors}
    \vcenter{\hbox{\includegraphics{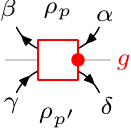}}}
    =A(\rho_p,\rho_{p'})_{\alpha\beta\gamma\delta}^{g}
    :=\rho_p(g)_{\alpha\beta}\,
      \rho_{p'}(g^{-1})_{\gamma\delta}
\end{equation}
with each oriented edge. Here $\rho_p(g)_{\alpha\beta}$ denotes a matrix
element of the representation~$\rho_p$ and the red dot represents the
\emph{physical} index~$g$ with its orientation inherited from the underlying
graph edge. The Greek letters attached to black arrows are called \emph{virtual}
indices. The tensor for a reversed graph edge naturally reads
\begin{equation}
    \vcenter{\hbox{\includegraphics{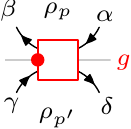}}}
    :=\vcenter{\hbox{\includegraphics{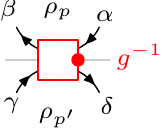}}}\ .
\end{equation}
Reversing virtual arrows is defined by
\begin{equation}
    \vcenter{\hbox{\includegraphics{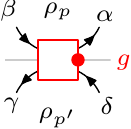}}}
    :=\vcenter{\hbox{\includegraphics{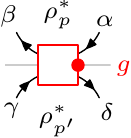}}}\ .
\end{equation}
where the dual representation~$\rho^*$ of a representation~$\rho$ is
given by the matrix equation $\rho^*(g):=\rho^T(g^{-1})$. Note that this
is equivalent to a reflection of the tensor about its vertical axis:
\begin{equation}
    \label{eq:tensor_reflected}
    \vcenter{\hbox{\includegraphics{tensors-25}}}
    =\vcenter{\hbox{\includegraphics{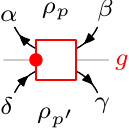}}}\ .
\end{equation}
The tensor network obtained by contracting all virtual indices
around each face represents what we will
call a \emph{group tensor network state} in the following.

For example, on a square lattice the tensor representing a vertex with a
particular edge orientation is given by
\begin{equation}
    \label{eq:tn_vertex}
    \vcenter{\hbox{\includegraphics{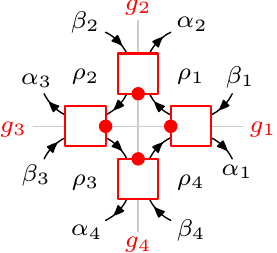}}}
    =\prod_{j=1}^4
     \rho_j(g_j^{-1}g_{j+1})_{\beta_j\alpha_{j+1}}
\end{equation}
where the index~$j$ labelling graph edges is assumed to be cyclic. Quite
similarly, the tensor corresponding to a particular face reads
\begin{equation}
    \label{eq:tn_face}
    \vcenter{\hbox{\includegraphics{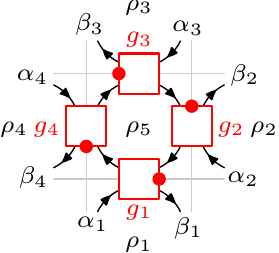}}}
    =\chi_{\rho_5}(g_4\cdots g_1)
     \prod_{j=1}^4
     \rho_j(g_j^{-1})_{\alpha_j\beta_j}
\end{equation}
where~$\chi_\rho$ is the character of the representation~$\rho$.

In fact, both this vertex and face tensor are \emph{partially}
contracted tensor networks and represent linear maps from the uncontracted virtual
indices to the physical ones. So a piece of a tensor network with open
indices can be regarded as a {\sc PEPS} projection
map~\cite{Verstraete:2004p265,Murg:2007p2715}. If $V_p$ is the
vector space associated with the representation~$\rho_p$ then this linear
map $P:V_p\otimes V_p\otimes V_{p'}\otimes V_{p'}\to\mathbb{C}G$ is given by
\begin{equation}
    \label{eq:projector}
    P=\!\!\sum_{g,\alpha,\dots,\delta\in G}\!\!
      A(\rho_p,\rho_{p'})_{\alpha\beta\gamma\delta}^g
      \ket{g}
      \bra{\alpha,\beta,\gamma,\delta}
\end{equation}
for a single edge tensor~$A(\rho_p,\rho_{p'})$.

On the contrary, a \emph{fully} contracted tensor network is a complex
number which equals the inner product between a particular basis state
and the group tensor network state. If viewed as a linear map from basis
configurations to amplitudes we will call it a \emph{tensor trace}.

Clearly, a group tensor network state only depends on the isomorphism
class of the representations~$\rho_p$ assigned to the faces. In general,
there is vastly more gauge freedom for a tensor network to represent the
same physical state, however, this does not necessarily respect the
$G$-action.

\bigskip

\noindent If we choose all representations~$\rho_p$ to coincide with the left
regular one the resulting group tensor network state will be a ground
state of the $\mathrm{D}(\mathbb{C}G)$-model. Since its Hamiltonian is
frustration free (i.e. a sum of mutually commuting terms) this can easily be shown by considering independently
the action of $A(s)$ and $B(p)$ on the partially contracted tensor
networks in~\eqref{eq:tn_vertex} and~\eqref{eq:tn_face}.

As an example, a ground state of the toric code is given by the
group tensor network for the group $\mathbb{Z}_2=\{e,a\}$ with all faces
carrying the regular representation. In that case the {\sc PEPS}
projection map reads $P=\ket{e}\bra{\phi^+}\bra{\phi^+}
+\ket{a}\bra{\psi^+}\bra{\psi^+}$ with the Bell states
\begin{align}
    \ket{\phi^+} & =\!\!\sum_{\alpha,\beta\in\mathbb{Z}_2}\!\!
                    L(e)_{\alpha\beta}
                    \ket{\alpha,\beta}
                   =\ket{e,e}+
                    \ket{a,a}, \\
    \ket{\psi^+} & =\!\!\sum_{\alpha,\beta\in\mathbb{Z}_2}\!\!
                    L(a)_{\alpha\beta}
                    \ket{\alpha,\beta}
                   =\ket{e,a}+
                    \ket{a,e}.
\end{align}
We would like to remark that this projection map coincides (up to a
trivial isomorphism) with the one given in~\cite{Verstraete:2006p113}.

Later in the article we will turn to the physical significance of group
tensor network states arising from representations of $G$ other than the
regular one.

\bigskip

\noindent
In order to build some more intuition for the main part of
this article we briefly review the structure of a local tensor in a
group tensor network state. It is easily seen that the tensor in~\eqref{eq:group_tensors}
factorizes with respect to partitioning the virtual indices into the
sets~$\{\alpha,\beta\}$ and $\{\gamma, \delta\}$. This implies that the
{\sc PEPS} projection map~\eqref{eq:projector} can be rewritten as:
\begin{equation}
    P
    =\sum_{g\in G}
     \ket{g}
     \bra{\phi_p(g)}
     \bra{\phi_{p'}(g^{-1})}
\end{equation}
where a state associated to the virtual indices in face~$p$ is given by:
\begin{equation}
    \ket{\phi_p(g)}=\!\sum_{\alpha,\beta\in G}\!
                    \overline{\rho_p(g)_{\alpha\beta}}
                    \ket{\alpha,\beta}.
\end{equation}
This means that one can carry out the contraction of the tensor network
in two steps. First one contracts each virtual loop separately and
then one needs to glue these pieces together. In a group tensor network
this glueing process is simply given by
\begin{equation}
    \label{eq:glueing}
    \ket{g}\otimes
    \ket{g^{-1}}
    \mapsto
    \ket{g}.
\end{equation}
It turns out that this step becomes rather nontrivial once we generalize
the quantum double model to finite-dimensional Hopf $C^*$-algebras. In fact, we
anticipate that this glueing step is one of the crucial ingredients for
deriving the hierarchy of tensor network states in this article.

\subsection{Outline}

The article is structured as follows. In Section~\ref{sec:drinfeld}, we
will provide a guide to the language of Hopf algebras, introducing the
necessary notation needed in order to extend the quantum double models based on
finite groups to Hopf algebras. Characteristic
examples of Hopf algrebras relevant for this work are supplied, too. In
Section~\ref{sec:model} we then generalize Kitaev's quantum double
construction from group algebras to Hopf algebras, thereby introducing
the physical model which stands in the centre of this work. In
Section~\ref{sec:networks}, we will solve this model by providing a
tensor network representation of its ground state. This representation
only involves the canonical structures associated to the underlying
Hopf algebra and leads us to propose a novel hierarchy of tensor network
states based on Hopf subalgebras. In Section~\ref{sec:entropy} we
calculate both the entanglement entropy and the topological entanglement entropy for
distinguished states in the hierarchy. In section~\ref{sec:discussion} we provide
concluding remarks and an outlook on future work.

\section{The language of Hopf algebras}
\label{sec:drinfeld}

\subsection{Hopf algebras}

If one sets out to find an algebraic structure for symmetries of
linear spaces with tensor product structure, such as encountered in
many-body quantum physics, one is naturally led to the concept of a
Hopf algebra. This is because the axioms of a Hopf algebra directly
allow for duals and tensor products of representations. As general
reference for this section we recommend~\cite{Kassel:1995}.

First and foremost a Hopf algebra~$H$ is a vector space (over some
field~$k$ which we will take to be $\mathbb{C}$) equipped with some additional structure. On the one hand
there is an associative multiplication~$\mu$ of vectors which is linear in each
argument, hence one may define it as the linear map $\mu:H\otimes H\to
H$ with
\begin{equation}
    \label{eq:multiplication}
    \mu(x\otimes y)
    =xy
\end{equation}
such that
\begin{equation}
    (xy)z
    =x(yz).
\end{equation}
This multiplication is accompanied by a unit~$\eta$ which can be defined
as the linear injection of scalars into $H$, hence $\eta:k\to H$ where
\begin{equation}
    \label{eq:unit}
    \eta(\lambda)
    =\lambda 1_H.
\end{equation}
The element $1_H\in H$ is a left and right unit for the
multiplication~$\mu$. At this stage $H$ has the structure of an algebra.
The multiplication encodes the composition of symmetry transformations
in a representation space.

On the other hand there is a dual notion to multiplication. This dual
linear map is the comultiplication $\Delta:H\to H\otimes H$ and
is usually written in Sweedler notation%
\footnote{If there are more than three tensor factors in iterated coproducts
such as $\Delta^{(3)}(x)=\sum_{(x)}x^{(1)}\otimes x^{(2)}\otimes x^{(3)}\otimes
x^{(4)}$ we will consistently use
superscripts instead of primes.}
as
\begin{equation}
    \label{eq:comultiplication}
    \Delta(x)
    =\sum_{(x)}
     x'\otimes x''
    =\sum_{(x)}
     x^{(1)}\otimes
     x^{(2)}.
\end{equation}
This serves as a shorthand for sums of the form $\sum_ix_i^{(1)}\otimes x_i^{(2)}$.
It is required to be coassociative:
\begin{equation}
    (\Delta\otimes\id)\circ
    \Delta
    =(\id\otimes\Delta)\circ
     \Delta.
\end{equation}
Elements which satisfy $\Delta(x)=x\otimes x$ are called grouplike and collected
in the set~$\mathcal{G}(H)$.
This comultiplication has a companion called the counit $\epsilon:H
\rightarrow k$ which is required to satisfy the axiom:
\begin{equation}
    \sum_{(x)}
    \epsilon(x')\,x''
    =\sum_{(x)}
     x'\epsilon(x'')
    =x.
\end{equation}
In other words, it neutralizes the comultiplication. With this structure
alone, the vector space~$H$ is called a coalgebra. The comultiplication
encodes how symmetry transformations act on tensor products of representation
spaces. The counit provides the appropriate notion of a trivial representation,
or invariance under symmetry transformations.

If both the algebra and coalgebra structure are compatible then $H$ is
called a bialgebra. The compatibility axioms read:
\begin{align}
    \Delta(xy)    & =\Delta(x)\,
                     \Delta(y), \\
    \Delta(1_H)   & =1_H\otimes 1_H,
                     \label{eq:coproduct_unit} \\
    \epsilon(xy)  & =\epsilon(x)\,
                     \epsilon(y), \\
    \epsilon(1_H) & =1_k.
\end{align}

Finally, if there is a linear map $S:H\to H$ satisfying the axiom
\begin{equation}
    \sum_{(x)}
    x'S(x'')
    =\epsilon(x)\,
     1_H
    =\sum_{(x)}
     S(x')\,x''
\end{equation}
then $H=(H;\mu,\eta;\Delta,\epsilon;S)$ is called a Hopf algebra. The
map~$S$ is called the antipode and has the following properties:
\begin{align}
    S(xy)
    & =S(y)\,
       S(x), \\
    S(1_H)
    & =1_H, \\
    \sum_{(S(x))}
    S(x)'
    \otimes
    S(x)''
    & =\sum_{(x)}
       S(x'')
       \otimes
       S(x'), \\
    \epsilon\bigl(S(x)\bigr)
    & =\epsilon(x).
\end{align}
In any finite-dimensional Hopf algebra the order of the antipode~$S$
is finite~\cite{Schneider:1995p2245}, hence $S$ is an invertible map.
We will always make this assumption in the sequel. The antipode
is used to define dual (or conjugate) representations.

We may also define opposite multiplication~$\mu^\mathrm{op}$ and
comultiplication~$\Delta^\mathrm{cop}$ by
\begin{align}
    \mu^\mathrm{op}(x\otimes y) & =yx, \\
    \Delta^\mathrm{cop}(x)      & =\sum_{(x)}
                                   x''\otimes x'
\end{align}
relative to the multiplication~\eqref{eq:multiplication} and
comultiplication~\eqref{eq:comultiplication}. Then the sets
\begin{align}
    Z(H)      & =\{x\in H
                   \mid
                   \forall y\in H\colon
                   \mu^\mathrm{op}(x\otimes y)
                   =\mu(x\otimes y)
                   \} \\
    \Cocom(H) & =\{x\in H
                   \mid
                   \Delta^\mathrm{cop}(x)
                   =\Delta(x)
                   \}
\end{align}
are called the centre of $H$ and the cocommutative elements of
$H$ respectively. If $Z(H)=H$ then $H$ is a commutative Hopf algebra,
in the case $\Cocom(H)=H$ it is a cocommutative Hopf algebra.
Furthermore it turns out that
$H^\mathrm{op}=(H;\mu^\mathrm{op},\eta;\Delta,\epsilon;S^{-1})$
is again a Hopf algebra, called the opposite Hopf algebra of $H$.

One can now show that for a given Hopf algebra $H=(H;\mu,\eta;\Delta,
\epsilon;S)$ with underlying vector space~$H$ the dual vector space~$H^*$
has again the structure of a Hopf algebra.
More precisely,
\begin{equation}
    H^*
    =(H^*;
      \Delta^T,
      \epsilon^T;
      \mu^T,
      \eta^T;
      S^T)
\end{equation}
with the
structure maps%
\footnote{For a linear map
$F:U\to V$ the transpose map $F^T:V^*\to U^*$ is defined as usual by
$\eval{F^T(\alpha)}{x}:=\eval{\alpha}{F(x)}$
for all $\alpha\in V^*$ and $x\in U$.}
as indicated is called the dual Hopf algebra of $H$. By the same token,
the opposite Hopf algebra
$H^\mathrm{op}$
has the natural dual
\begin{equation}
    X
    =(H^\mathrm{op})^*
    =\bigl(H^*;
           \Delta^T,
           \epsilon^T;
           (\mu^\mathrm{op})^T,
           \eta^T;
           (S^{-1})^T\bigr).
\end{equation}

Finally, there is the notion of a (two-sided) integral~$\Lambda$ in a
Hopf algebra~$H$. Such an element is defined via
\begin{equation}
    \label{eq:integral}
    x\Lambda
    =\Lambda x
    =\epsilon(x)
     \Lambda
\end{equation}
for all $x\in H$. In view of the role of the counit as a trivial representation, integrals are thus
invariant elements under multiplication. The dual definition of an integral $\Gamma\in H^*$ can
be phrased as
\begin{equation}
    \label{eq:dual_integral}
    \sum_{(x)}
    x'
    \Gamma(x'')
    =\sum_{(x)}
     \Gamma(x')\,
     x''
    =\Gamma(x)\,
     1_H
\end{equation}
for all $x\in H$. This turns out to be equivalent to~\eqref{eq:integral} once applied to $H^*$.

Particularly important is the notion of a Haar integral. This is the normalized
version of an integral which is guaranteed to exist uniquely for finite-dimensional
Hopf $C^*$-algebras. Those are the Hopf algebras we will be occupied with mostly in the discussion of the physical model. We will supply a precise definition of
Haar integrals in Section~\ref{sec:Hilbert}.

\subsection{Quantum doubles as bicrossed products}

In a general Hopf algebra~$H$ nothing can be said about how $\Delta$ and
$\Delta^\mathrm{cop}$ are related to each other. In other words, $H$ can
be non-cocommutative in the most unpleasant way. Therefore we seek
Hopf algebras which are \emph{almost} cocommutative%
\footnote{In physical applications this is related to the exchange statistics of particles in two dimensions.}
in a certain sense: they are equipped with a so-called quasitriangular
structure which controls the extent to which the comultiplication
fails to be cocommutative. In~\cite{Drinfeld:1986} Drinfeld gave his celebrated
quantum double construction which produces such a quasitriangular Hopf algebra from a given Hopf algebra.

At the heart of his construction lies the idea of a bicrossed product. For example, in the case of the semidirect product of two groups
one has an action of one group on the other and this action is used to
define the multiplication in the product group. Similarly, one may define
a bicrossed product of groups where there is an additional backaction of the
second group on the first one. It is this concept that can be generalized to Hopf algebras and will yield the quantum double.

So given a Hopf algebra~$H$ with invertible antipode~$S$
(in particular, any finite-dimensional Hopf algebra) one can
construct a new Hopf algebra $\mathrm{D}(H)=X\bowtie H$, the quantum double of $H$, as the bicrossed
product of $X=(H^\mathrm{op})^*$ and $H$. This means that as a vector
space $\mathrm{D}(H)$ simply equals $H^*\otimes H$. In order to define
its multiplication we need an action of $H$ on $X$ and another one vice
versa. This bicrossed structure is given by the actions
\begin{align}
    \rhd:
    H\times X
    & \to X,
    & a\rhd f
    & =\sum_{(a)}
       f\bigl(S^{-1}(a'')\,?a'\bigr),
       \label{eq:left_action}\\
    \lhd:
    H\times X
    & \to H,
    & a\lhd f
    & =\sum_{(a)}
       f\bigl(S^{-1}(a''')\,a'\bigr)\,
       a''
       \label{eq:right_action},
\end{align}
where the question mark denotes the argument of the function~$a\rhd f$, i.e.
in this particular instance
\begin{equation}
    (a\rhd f)(x)
    =\sum_{(a)}
     f\bigl(S^{-1}(a'')\,
            xa'\bigr)
\end{equation}
holds for any $x\in H$. We will use this notation frequently later on.

One can then show that
\begin{equation}
    (f\otimes a)
    (g\otimes b)
    =\sum_{(a)}
     f\,
     g\bigl(S^{-1}(a''')\,?a')
     \otimes
     a''b,
\end{equation}
together with the canonical comultiplication on the tensor product of
Hopf algebras, defines a valid Hopf algebra structure on the vector
space $\mathrm{D}(H)$.
Since both the embeddings $i_X:X\to \mathrm{D}(H)$, $f\mapsto f
\otimes1_H$ and $i_H:H\to \mathrm{D}(H)$, $a\mapsto 1_X\otimes a$ are
algebra morphisms this multiplication formula is actually already
determined by the so-called straightening formula%
\footnote{This should be contrasted with the canonical (uncrossed) multiplication
    in the Hopf algebra $X\otimes H$. In this case we simply have
    $af=fa$ for any $f\in X$ and $a\in H$ since multiplication
    is defined componentwise.}
\begin{equation}
    \label{eq:straighten}
    af
    :=(1_X\otimes a)
      (f\otimes 1_H)
    =\sum_{(a)}
     f\bigl(S^{-1}(a''')\,?a'\bigr)\,
     a''.
\end{equation}

It should be noted that both underlying actions~\eqref{eq:left_action}
and~\eqref{eq:right_action} are derived from the adjoint representation
which for $a,x\in H$ is defined by
\begin{equation}
    \ad(a)(x)
    =\sum_{(a)}
     a'x\,
     S(a'').
\end{equation}

\subsection{Hopf $*$-algebras}

If $k=\mathbb{C}$ then a Hopf algebra~$H$ may sometimes be equipped with
a so-called
$*$-structure that will become most important for physical applications.
Namely, it will allow us to define Hilbert spaces and unitarity. For
a general reference on Hopf $*$-algebras see~\cite{Klimyk:1997}.

Firstly, a conjugate-linear map $*:H\to H$ which satisfies
\begin{align}
    (x^*)^* & =x \\
    (xy)^* & =y^*x^*
\end{align}
is called an involution and turns $H$ into a $*$-algebra.
It follows naturally that
\begin{equation}
    1_H^*
    =1_H.
\end{equation}

If an involution $*$ is compatible with comultiplication
\begin{equation}
    \Delta(x^*)
    =\Delta(x)^*
\end{equation}
then $H$ is called a $*$-coalgebra. Here the involution of $H
\otimes H$ is defined by $(x\otimes y)^*=x^*\otimes y^*$. In a $*$-coalgebra
one always has
\begin{equation}
    \epsilon(x^*)
    =\overline{\epsilon(x)}
\end{equation}
where the bar denotes complex conjugation.

A bialgebra~$H$ with an involution for which it is both a $*$-algebra
and a $*$-coalgebra is called a $*$-bialgebra.

Interestingly, if a Hopf algebra~$H$ also has the structure of a
$*$-bialgebra then the interplay between antipode and involution is
already determined:
\begin{equation}
    S\bigl(S(x^*)^*\bigr)
    =x.
\end{equation}
Consequently, $H$ is called a Hopf $*$-algebra. Note that its antipode
is always invertible (even if $H$ is not finite-dimensional).

Finally, one can show that the dual of a Hopf $*$-algebra~$H$ is
again a Hopf $*$-algebra with the involution given by
\begin{equation}
    \label{eq:dual_star}
    f^*(x)
    =\overline{f\bigl(S(x)^*\bigr)}.
\end{equation}

\subsection{Trivial Hopf algebras and their quantum doubles}

Here we briefly discuss some cases of how Hopf algebras can arise from groups.
In particular we focus on how a group itself can be directly understood in the language
of Hopf algebras. This will be the key for relating Kitaev's original construction
to our generalization.

\subsubsection{Trivial Hopf algebras}
\label{sec:trivial_algebras}

A finite-dimensional Hopf algebra is said to be trivial if it is a group algebra or
the dual of a group algebra for some finite group~$G$. Let $\{g\mid g\in
G\}$ be a basis of the group algebra~$\mathbb{C}G$ with multiplication
\begin{equation}
    \mu(g\otimes h)
    =gh
\end{equation}
and unit $\eta(1_\mathbb{C})=e$. Its comultiplication and counit are
given by
\begin{equation}
    \Delta(g)
    =g\otimes g,\qquad
    \epsilon(g)
    =1_\mathbb{C}
\end{equation}
and extended by linearity. Similarly, the antipode map is defined by
\begin{equation}
    S(g)
    =g^{-1}.
\end{equation}
It is easy to see that every group algebra~$\mathbb{C}G$ is
\emph{co}commutative, but in general \emph{non}-commutative.
Furthermore, one may endow $\mathbb{C}G$ with the involution map
\begin{equation}
    g^*
    =g^{-1}
\end{equation}
which is extended by conjugate linearity. This turns $\mathbb{C}G$ into
a Hopf $C^*$-algebra. Its Haar integral reads
\begin{equation}
    h
    =\frac{1}{\abs{G}}
     \sum_{g\in G}
     g.
\end{equation}

The dual~$(\mathbb{C}G)^*$ of a group algebra coincides with the
space~$\mathbb{C}^G$ of linear functions from the group to the field of
complex numbers. For the dual basis~$\{\delta_g\mid g\in G\}$
multiplication and unit are defined by
\begin{equation}
    \delta_g
    \delta_h
    =\delta_{g,h}
     \delta_g,\qquad
    \eta(1_\mathbb{C})
    =\sum_{g\in G}
     \delta_g
    =1
\end{equation}
where $\delta_{g,h}$ denotes the Kronecker delta. The comultiplication,
counit and antipode maps are given by
\begin{equation}
    \Delta(\delta_g)
    =\!\sum_{uv=g}
     \delta_u\otimes
     \delta_v,\qquad
    \epsilon(\delta_g)
    =\delta_g(e)
    =\delta_{g,e},\qquad
    S(\delta_g)
    =\delta_{g^{-1}}
\end{equation}
for each $g\in G$. One may check that every function
algebra~$\mathbb{C}^G$ is commutative, but in general
\emph{non}-\emph{co}commutative. When contrasted with the corresponding
group algebra, this is not too surprising: since $\mathbb{C}G$ and
$\mathbb{C}^G$ are dual to each other, one obtains the multiplication of
one algebra from the comultiplication of the other one and vice versa.
Again, $\mathbb{C}^G$ is a Hopf $C^*$-algebra via
\begin{equation}
    \delta_g^*
    =\delta_g
\end{equation}
and conjugate-linear extension%
\footnote{If $G$ is an Abelian group then $\delta_g^\star
    =\delta_{g^{-1}}$ defines another Hopf $*$-structure for the
    standard convention. In fact, for $G$ non-Abelian the choice of $*$
    or $\star$ corresponds to a choice of either the standard convention
    or the convention $(a\otimes b)^\star=b^\star\otimes a^\star$ for
    the action on tensor products.}
. Finally,
\begin{equation}
    \phi
    =\delta_e
\end{equation}
is the Haar integral of $\mathbb{C}^G$.

\subsubsection{Quantum doubles of trivial Hopf algebras}

As a vector space the quantum double of a group algebra is nothing but
$\mathrm{D}(\mathbb{C}G)=\mathbb{C}^G\otimes\mathbb{C}G$ with basis
$\{\delta_x\otimes g\mid (x,g)\in G\times G\}$. The fact that the
quantum double is a crossed product is reflected in the multiplication
map
\begin{equation}
    (\delta_x\otimes g)
    (\delta_y\otimes h)
    =\delta_{x,gyg^{-1}}
     \delta_x\otimes
     gh
\end{equation}
where the algebra acts on itself by conjugation. Actually, this is fully
determined by the straightening formula
\begin{equation}
    g
    \delta_x
    :=(1\otimes g)
     (\delta_x\otimes e)
    =\delta_{gxg^{-1}}\otimes
     g.
\end{equation}
(and the embeddings $\delta_x\mapsto\delta_x\otimes e$ and $g\mapsto
1\otimes g$ being algebra morphisms) Here $g\delta_x$ is a frequently
used shorthand notation and should be clear from context whenever it
appears. The unit of $\mathrm{D}(\mathbb{C}G)$ is $1\otimes e$.
Comultiplication and counit are given by
\begin{align}
    \Delta(\delta_x\otimes g)
    & =\!\sum_{uv=x}
       (\delta_v\otimes g)\otimes
       (\delta_u\otimes g), \\
    \epsilon(\delta_x\otimes g)
    & =\epsilon(g)\,
       \delta_x(e)
      =\delta_{x,e}
\end{align}
on the basis elements. Note that the comultiplication in
$\mathrm{D}(\mathbb{C}G)$ is derived from the comultiplication in
$\mathbb{C}G$ and $(\mathbb{C}^G)^\mathrm{cop}
=((\mathbb{C}G)^\mathrm{op})^*$ rather than from $\mathbb{C}^G$ itself!
This is because as an algebra the quantum double is the bicrossed
product $\mathrm{D}(\mathbb{C}G)
=\bigl((\mathbb{C}G)^\mathrm{op}\bigr)^*\bowtie\mathbb{C}G$.

The antipode map is
\begin{equation}
    S(\delta_x\otimes g)
    =\delta_{g^{-1}x^{-1}g}\otimes
     g^{-1}.
\end{equation}

\subsubsection{Hopf subalgebras of trivial Hopf algebras}
\label{sec:subalgebras_group}

The Hopf subalgebras of a Hopf algebra will be important for the
construction of our hierarchy of Hopf tensor network states.

If $H=\mathbb{C}G$, all the Hopf subalgebras~$A$ and $B$ can be
characterized easily~\cite{Bais:2002p1647,Bais:2003p1648}. One finds that the
Hopf subalgebras of $\mathbb{C}G$ are precisely the group
algebras~$\mathbb{C}K$ for $K\subset G$ a subgroup. Furthermore the Hopf
subalgebras of $\mathbb{C}^G$ are exactly given by $\mathbb{C}^{G/N}$
with $N\lhd G$ a normal subgroup. It is also not difficult to find the
Haar integrals of these Hopf subalgebras. In the first case one has
\begin{equation}
    h_K
    =\frac{1}{\abs{K}}
     \sum_{g\in K}
     g.
\end{equation}
For the second case note that $\mathbb{C}^{G/N}$ can be identified with
the functions in $\mathbb{C}^G$ which are constant on cosets of $N$. Now
consider the character~$\chi_N$ of $G$ that factors through $N$ to the
regular character of $G/N$. It is easy to see that when multiplied with
the above functions~$\chi_N$ acts as a two-sided integral. Hence taking
normalization into account, one may convince oneself that
$\abs{G/N}^{-1}\chi_N$ is precisely the Haar integral of
$\mathbb{C}^{G/N}$. By abuse of notation we may write it as:
\begin{equation}
    \phi_N
    =\sum_{n\in N}
     \delta_n.
\end{equation}

\subsection{An example of a nontrivial finite-dimensional Hopf $C^*$-algebra}

We are primarily interested in finite-dimensional nontrivial Hopf $C^*$-algebras,
i.e. finite-dimensional nontrivial Hopf algebras over $\mathbb{C}$ with a compatible
involution. In order to find an example of such a Hopf algebra it is
enough to investigate semisimple Hopf algebras in low dimensions and
check whether these support a compatible $*$-structure. It turns out
that any semisimple Hopf algebra~$H$ of dimension~$p$ or $pq$ with prime
numbers $p<q$ is trivial~\cite{Etingof:1998p1915}. The same applies to
dimensions~$p^2$~\cite{Masuoka:1996p1917}. Hence one may expect
nontrivial semisimple Hopf algebras only for $\dim H\geq8$.

\subsubsection{The Hopf algebra~$H_8$}

In fact, we need to look no further than $\dim H=8$ for an example. At
this dimension there exists a unique nontrivial semisimple Hopf algebra
(up to isomorphism) which is commonly denoted by $H_8$ and is due to Kac
and Paljutkin~\cite{Kac:1966p2444}. Interestingly, it can also be
endowed with a compatible involution which turns $H_8$ into a Hopf
$C^*$-algebra.

Following~\cite{Masuoka:1995p2441} this algebra can be presented by
generators~$x$, $y$, $z$ and the relations
\begin{gather}
    x^2
    =y^2
    =1,\qquad
    z^2
    =\frac{1}{2}
     (1+x+y-xy), \\
    xy
    =yx,\qquad
    zx
    =yz,\qquad
    zy
    =xz.
\end{gather}
It is now easy to see that $\mathcal{B}=\{1,x,y,xy,z,zx,zy,zxy\}$ is a basis of
$H_8$. We will denote the corresponding dual basis by $\{\delta_i\mid
i\in\mathcal{B}\}$.
The coalgebra structure is determined by
\begin{align}
    \Delta(x)   & =x\otimes
                   x,\qquad
                  \Delta(y)
                  =y\otimes
                   y, \\
    \Delta(z)   & =\frac{1}{2}
                   (1\otimes
                    1+
                    y\otimes
                    1+
                    1\otimes
                    x-
                    y\otimes
                    x)
                   (z\otimes z), \\
    \epsilon(x) & =\epsilon(y)
                  =\epsilon(z)
                  =1,
\end{align}
and from this it is evident that $\mathcal{G}(H_8)=\{1,x,y,xy\}
\simeq\mathbb{Z}_2\times\mathbb{Z}_2$.
The antipode reads
\begin{equation}
    S(x)
    =x,\qquad
    S(y)
    =y,\qquad
    S(z)
    =z.
\end{equation}
Finally, the involution is given by
\begin{equation}
    x^*
    =x,\qquad
    y^*
    =y,\qquad
    z^*
    =z^{-1}
    =\frac{1}{2}
     (z+zx+zy-zxy)
    =z^3.
\end{equation}
This shows in particular that antipode and involution act differently on
the generators of $H_8$.

The Haar integrals $h\in H_8$ and $\phi\in H_8^*$ are given by:
\begin{align}
    h    & =\frac{1}{8}
            (1+x+y+xy+z+zx+zy+zxy), \\
    \phi & =\delta_1.
\end{align}
Remarkably they look like the Haar integrals of a group algebra and its
dual! The dual Haar integral~$\phi$ and the complex involution
automatically yield a Hermitian inner product on $H_8$ as discussed in
Section~\ref{sec:Hilbert}. Explicitly, we have for all $a,b\in H_8$
\begin{equation}
    (a,b)
    =\delta_1(a^*b).
\end{equation}
It follows that $\mathcal{B}$ is an orthonormal basis of
$H_8$.

\subsubsection{Hopf subalgebras of $H_8$ and its dual}
\label{sec:subalgebras_H8}

The Hopf subalgebras of $H=H_8$ read $\{1\}$, $\{1,x\}
\simeq\mathbb{CZ}_2$, $\{1,y\}\simeq\mathbb{CZ}_2$, $\{1,xy\}
\simeq\mathbb{CZ}_2$, $\{1,x,y,xy\}
\simeq\mathbb{C}(\mathbb{Z}_2\times\mathbb{Z}_2)$ and $H_8$. This is
because Hopf subalgebras of a Hopf algebra~$H$ correspond bijectively to
sets of irreducible representations of $H^*$ which close under fusion and
dualisation~\cite{Bais:2002p1647,Bais:2003p1648}. Now $H_8^*$ has four
one-dimensional irreducible representations~$Q_1$, $Q_x$, $Q_y$ and $Q_{xy}$
derived from $\mathcal{G}(H_8)$ and the two-dimensional representation given by
\begin{equation}
    Q_2
    =\frac{1}{2}
    \begin{pmatrix}
        z+zx   & z-zx \\
        zy-zxy & zy+zxy
    \end{pmatrix}.
\end{equation}
From the fusion rules
\begin{align}
    Q_g\times Q_{g'} & =Q_{gg'} \\
    Q_g\times Q_2    & =Q_2  \\
    Q_2\times Q_2    & =Q_1+Q_x+Q_y+Q_{xy}
\end{align}
for all $g,g'\in \mathcal{G}(H_8)$ one concludes that precisely the
sets~$\{Q_1\}$, $\{Q_1,Q_x\}$, $\{Q_1,Q_y\}$, $\{Q_1,Q_{xy}\}$, $\{Q_1,
Q_x,Q_y,Q_{xy}\}$ and $\{Q_1,Q_x,Q_y,Q_{xy},Q_2\}$ close under fusion
(and dualization) and therefore are the ones that correspond bijectively
to Hopf subalgebras of $H_8$.

The Haar integrals of the Hopf subalgebras read
\begin{align}
    h_{\{1\}}     & =1, \\
    h_{\gen{x}}   & =\frac{1}{2}\,
                     (1+x), \\
    h_{\gen{y}}   & =\frac{1}{2}\,
                     (1+y), \\
    h_{\gen{xy}}  & =\frac{1}{2}\,
                     (1+xy), \\
    h_{\gen{x,y}} & =\frac{1}{4}\,
                     (1+x+y+xy), \\
    h             & =\frac{1}{8}\,
                     (1+x+y+xy+z+zx+zy+zxy).
\end{align}

Since $H_8^*\simeq H_8$ as Hopf algebras the Hopf subalgebras $B\subset
H_8^*$ are given as the images of the Hopf subalgebras of $H_8$ under
the isomorphism. A particular such Hopf isomorphism can be obtained as
follows. From the dual basis~$\{\delta_a\}$ of the basis $\mathcal{B}\subset H_8$
define another basis $\{f_a\}\subset H_8^*$ via
\begin{align}
    f_1     & :=\delta_1+
                \delta_x+
                \delta_y+
                \delta_{xy}+
                \delta_z+
                \delta_{zx}+
                \delta_{zy}+
                \delta_{zxy}, \\
    f_x     & :=\delta_1+
                \delta_x+
                \delta_y+
                \delta_{xy}-
                \delta_z-
                \delta_{zx}-
                \delta_{zy}-
                \delta_{zxy}, \\
    f_y     & :=\delta_1-
                \delta_x-
                \delta_y+
                \delta_{xy}+
                i\delta_z-
                i\delta_{zx}-
                i\delta_{zy}+
                i\delta_{zxy}, \\
    f_{xy}  & :=\delta_1-
                \delta_x-
                \delta_y+
                \delta_{xy}-
                i\delta_z+
                i\delta_{zx}+
                i\delta_{zy}-
                i\delta_{zxy}, \\
    f_z     & :=\delta_1+
                \delta_x-
                \delta_y-
                \delta_{xy}+
                \delta_z+
                \delta_{zx}-
                \delta_{zy}-
                \delta_{zxy}, \\
    f_{zx}  & :=\delta_1-
                \delta_x+
                \delta_y-
                \delta_{xy}+
                \delta_z-
                \delta_{zx}+
                \delta_{zy}-
                \delta_{zxy}, \\
    f_{zy}  & :=\delta_1+
                \delta_x-
                \delta_y-
                \delta_{xy}-
                \delta_z-
                \delta_{zx}+
                \delta_{zy}+
                \delta_{zxy}, \\
    f_{zxy} & :=\delta_1-
                \delta_x+
                \delta_y-
                \delta_{xy}-
                \delta_z+
                \delta_{zx}-
                \delta_{zy}+
                \delta_{zxy}.
\end{align}
Then the map $H_8\to H_8^*$ given by $a\mapsto f_a$ for all $a
\in\mathcal{B}$ and linear extension is a Hopf isomorphism.

Consequently, the Haar integrals~$\phi_B$ read:
\begin{align}
    \phi_{\{f_1\}}
    & =\delta_1+
       \delta_x+
       \delta_y+
       \delta_{xy}+
       \delta_z+
       \delta_{zx}+
       \delta_{zy}+
       \delta_{zxy}
      =\epsilon, \\
    \phi_{\gen{f_x}}
    & =\delta_1+
       \delta_x+
       \delta_y+
       \delta_{xy},
       \displaybreak[0]\\
    \phi_{\gen{f_y}}
    & =\delta_1+
       \delta_{xy}+
       \frac{1+i}{2}\,
       (\delta_z+
        \delta_{zxy})+
       \frac{1-i}{2}\,
       (\delta_{zx}+
        \delta_{zy}),
       \displaybreak[0]\\
    \phi_{\gen{f_{xy}}}
    & =\delta_1+
       \delta_{xy}+
       \frac{1-i}{2}\,
       (\delta_z+
        \delta_{zxy})+
       \frac{1+i}{2}\,
       (\delta_{zx}+
        \delta_{zy}),
       \displaybreak[0]\\
    \phi_{\gen{f_x,f_y}}
    & =\delta_1+
       \delta_{xy}, \\
    \phi
    & =\delta_1.
\end{align}

\section{Constructing quantum double models from Hopf algebras}
\label{sec:model}

The goal of this section is to construct a two-dimensional quantum spin model whose
microscopic degrees of freedom are given by a finite-dimensional Hopf algebra~$H$, such that
its emerging degrees of freedom are characterized by the quantum
double~$\mathrm{D}(H)$. This illuminates Kitaev's
insight in~\cite{Kitaev:2003}.
The model obtained in this fashion dynamically implements the quantum
double construction. This means that its Hilbert
space aquires additional structure, namely that the graph underlying the
spin model can be interpreted as a $\mathrm{D}(H)$-module, or representation
of $\mathrm{D}(H)$. State vectors
describing (elementary) quasiparticle excitations above the ground state
then live in the irreducible representations of $\mathrm{D}(H)$ the Hilbert space
decomposes into. In fact, for every irreducible representation of the quantum
double there will be a corresponding type of quasiparticle
excitation. These quasiparticle excitations can then naturally be
braided via the quasitriangular structure of $\mathrm{D}(H)$ and, in general,
exhibit exchange statistics beyond bosons or fermions. In order to render
their topological nature manifest we would also like these quasiparticle
excitations to be agnostic to the details of the microscopic background.
In particular, they should only feel the topological properties of the
underlying surface rather than the precise shape of the embedded graph.
In other words, the quantum spin model we construct should be insensitive
to the particular discretization chosen. It goes without saying that any
discretization will typically have a fine granularity in order to make the
condensed matter character of the model evident.

Since there is a natural action of the quantum double~$\mathrm{D}(H)$ on
the Hopf algebra~$H$ itself one can nevertheless regard a minimal graph consisting of
precisely one edge as a representation of $\mathrm{D}(H)$ by identifying the edge
with $H$. However, it is clear that this system cannot contain all
irreducible representations of the quantum double. While this restriction
can already be overcome by considering two graph edges associated with
$H\otimes H$ one still needs to go to a macroscopic regime from either of these small graphs. It is precisely the process of spatially extending
the action of $\mathrm{D}(H)$
from one edge to many edges which will eventually yield the Hamiltonian
and thus define the quantum double model completely.

In the following we will first construct a representation of $\mathrm{D}(H)$ from
a small graph with just two edges. Next we will endow any larger graph
obtained from a surface cellulation with a local $\mathrm{D}(H)$-module structure. Finally
we will introduce the Hilbert space for our generalized quantum double model
and assemble the Hamiltonian from particular operators in the local
representations of $\mathrm{D}(H)$.

\subsection{Graph representations of quantum doubles}

Given a finite-dimensional Hopf algebra~$H$ and an oriented graph~$\Gamma=(V,E,F)$
we define the
microscopic Hilbert space as $\mathcal{N}_\Gamma=\bigotimes_{e\in E}H$. Our
goal is to find local operators~$A_a$ and $B_f$ acting on $\mathcal{N}_\Gamma$
which represent both parts~$H$ and $X=(H^\mathrm{op})^*$ of the quantum
double~$\mathrm{D}(H)$ separately via the embeddings~$i_X$ and $i_H$.
Additionally we want these two operators to commute in such a way that
they implement the straightening formula~\eqref{eq:straighten} of $\mathrm{D}(H)$. Put
differently, we require the operators~$A_a$ and $B_f$ to
\emph{establish} the bicrossed multiplication of the quantum double by
interacting nontrivially on the intersection of their supports. Loosely
speaking, this amounts to assembling the adjoint representation from
smaller building blocks. Note that on the minimal graph comprising a single edge
the adjoint representation can be turned into an operator “as is”.

\begin{figure}
    \centering
    \includegraphics{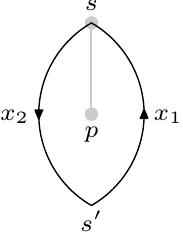}
    \caption{Small graph representing $H\otimes H$. It affords a
        $\mathrm{D}(H)$-module structure at the site~$(s,p)$
        via the operators $A_a(s,p)$ and $B_f(s,p)$.}
    \label{fig:single_loop}
\end{figure}

We begin with the small graph shown in Figure~\ref{fig:single_loop} and
its associated Hilbert space $\mathcal{N}_\Gamma=H\otimes H$. In the spirit
of~\cite{Kitaev:2003} we define the following module structures
on $H$ for all $a,x\in H$ and $f\in X$:
\begin{align}
    L_+^a(x) & :=ax,
                 \label{eq:action1}\\
    L_-^a(x) & :=x\,S(a),
                 \label{eq:action2}\\
    T_+^f(x) & :=\sum_{(x)}
                 \eval{f}
                      {x''}\,
                 x',
                 \label{eq:action3}\\
    T_-^f(x) & :=\sum_{(x)}
                 \eval{f}
                      {S^{-1}(x')}\,
                 x''.
                 \label{eq:action4}
\end{align}
More precisely, the operators $L_\pm$ define actions of $H$ on itself while
$T_\pm$ define actions of $X$ on $H$. Furthermore these actions
are intimately related: one may start with the
left multiplication~$L_+$ and then canonically derive all other
actions. As a consequence one has for example
\begin{align}
    L_-^a & =S\circ
             L_+^a\circ
             S^{-1},
             \label{eq:L+-}\\
    T_-^f & =S\circ
             T_+^f\circ
             S^{-1}.
             \label{eq:T+-}
\end{align}
             This means that
    if one fixes an arbitrary pattern of edge orientations one may
    relate any other pattern to the original one using these relations.
    Unfortunately this does not treat all possible orientations on
    equal footing. We will resolve this
    issue in Section~\ref{sec:Hilbert}. Also note that
\begin{align}
    [L_+^a,L_-^b] & =0
                     \label{eq:commutation_L}, \\
    [T_+^f,T_-^g] & =0
                     \label{eq:commutation_T}
\end{align}
for arbitrary $a,b\in H$ and $f,g\in X$.

Next we attribute operators~$A_a$ and $B_f$ to the sites~$(s,p)$ of the
graph~$\Gamma$. Following~\cite{Kitaev:2003}, a site is defined as a pair consisting of a vertex~$s$ and one of its
adjacent faces~$p$. For the particular site~$(s,p)$ shown in
Figure~\ref{fig:single_loop} we define the operators in terms of the above
actions:
\begin{align}
    A_a(s,p) & :=\sum_{(a)}
                 L_+^{a'}\otimes
                 L_-^{a''}, \\
    B_f(s,p) & :=\sum_{(f)}
                 T_-^{f'}\otimes
                 T_-^{f''}.
\end{align}
It is worth emphasizing that this definition only holds for the edge
orientation chosen in the figure. However, using the conventions set
forth by Figure~\ref{fig:kitaev_edge} we can extend this definition to
arbitrarily oriented edges. Also note an
important detail: our seemingly similar notation refers to two rather
different comultiplications, namely, the one in $H$ and the other one in
$X$.

\begin{figure}
    \centering
    \includegraphics{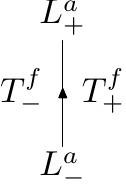}
    \caption{A graph edge representing the Hopf algebra~$H$ via the
        $H$-module structures~$L_\pm$. At the same time it represents
        the Hopf algebra~$X=(H^\mathrm{op})^*$ via the $X$-module
        structures~$T_\pm$. These are related to
        each other by means of the antipode in~\eqref{eq:L+-} and
        \eqref{eq:T+-}.}
    \label{fig:kitaev_edge}
\end{figure}

We are now prepared to analyze the commutation properties of these
operators which we take both at the \emph{same} site~$(s,p)$ for the
moment. It is clear that the support of both operators coincides on this
small graph, hence we can study the interaction of $A_a$ and $B_f$
without any additional complications. The following lemma shows that
indeed the two operators define a representation of $\mathrm{D}(H)$ on
$H\otimes H$.

\begin{lemma}
    For the operators~$A_a(s,p)$ and $B_f(s,p)$ the straightening
    formula~\eqref{eq:straighten} holds in the form:
    \begin{equation}
        A_aB_f
        =\sum_{(a)}
         B_{f(S^{-1}(a''')?a')}
         A_{a''}.
    \end{equation}
\end{lemma}

\begin{proof}
    The statement is proven by evaluating both sides of the
    straightening formula on arbitrary elements $x_1,x_2\in H$. To this
    end consider first:
    \begin{align*}
        A_aB_f
        (x_1\otimes x_2)
        & =A_a
           \sum_{(f)}
           T_-^{f'}(x_1)\otimes
           T_-^{f''}(x_2) \\
        & =A_a
           \sum_{(x_i)}
           \sum_{(f)}
           \eval{f'}
                {S^{-1}(x_1')}
           \eval{f''}
                {S^{-1}(x_2')}\,
           x_1''\otimes
           x_2'' \\
        & =\sum_{(x_i)}\,
           \Bigl\langle
           \sum_{(f)}
           f'\otimes f'',
           S^{-1}(x_1')\otimes
           S^{-1}(x_2')
           \Bigr\rangle
           \sum_{(a)}
           L_+^{a'}(x_1'')\otimes
           L_-^{a''}(x_2'') \\
        & =\!\sum_{(a)(x_i)}\!
           \eval{f}
                {S^{-1}(x_2')\,
                 S^{-1}(x_1')}\,
           a'x_1''\otimes
           x_2''\,
           S(a'') \\
        & =\!\sum_{(a)(x_i)}\!
           \eval{f}
                {S^{-1}(x_1'x_2')}\,
           a'x_1''\otimes
           x_2''\,
           S(a'').
    \end{align*}
    Note that in the third line we have used the opposite
    multiplication~$\mu^\mathrm{op}$ which is the appropriate dual to
    swap the comultiplication
    \begin{equation*}
        \sum_{(f)}
        f'\otimes f''
        =(\mu^\mathrm{op})^T(f)
    \end{equation*}
    in $X$ from $f$ onto its argument in $H$.

    In a second step compare this with:
    \begin{align*}
        \sum_{(a)}
        B_{f(S^{-1}(a''')?a')}
        A_{a''}
        (x_1\otimes x_2)
        & =\sum_{(a)}
           B_{f(S^{-1}(a^{(4)})?a^{(1)})}
           \bigl(a^{(2)}x_1\otimes
                 x_2\,S(a^{(3)})\bigr) \\
        & =\sum_{(a)}
           B_{\tilde{f}}
           \bigl(a^{(2)}x_1\otimes
                 x_2\,S(a^{(3)})\bigr)
    \end{align*}
    where we used the temporary abbreviation $\tilde{f}
    =f\bigl(S^{-1}(a^{(4)})\,? a^{(1)}\bigr)$. Now we have
    for the right-hand side
    \begin{align*}
        & \sum_{(a)(\tilde{f})}
          T_-^{\tilde{f}'}
          (a^{(2)}x_1)\otimes
          T_-^{\tilde{f}''}
          \bigl(x_2\,S(a^{(3)})\bigr)
          \displaybreak[0]\\
        & =\!\sum_{(a)(\tilde{f})}
           \sum_{(a^{(2)}x_1)}\!
           \bigeval{\tilde{f}'}
                   {S^{-1}\bigl((a^{(2)}x_1)'\bigr)}\,
           (a^{(2)}x_1)''\otimes{} \\
        & \relphantom{=}
           \hphantom{\!\sum_{(a)(\tilde{f})}}
           \sum_{(x_2S(a^{(3)}))}\!
           \bigeval{\tilde{f}''}
                   {S^{-1}\bigl[\bigl(x_2\,S(a^{(3)})\bigr)'\bigr]}\,
           \bigl(x_2\,S(a^{(3)})\bigr)''
           \displaybreak[0]\\
        & =\sum_{(a)}
           \sum_{(a^{(2)}x_1)}
           \sum_{(x_2S(a^{(3)}))}
           \bigeval{\tilde{f}}
                   {S^{-1}\bigl[\bigl(x_2\,S(a^{(3)})\bigr)'\bigr]\,
                    S^{-1}\bigl((a^{(2)}x_1)'\bigr)} \\
        & \relphantom{=}
           \hphantom{\sum_{(a)}
                     \sum_{(a^{(2)}x_1)}
                     \sum_{(x_2S(a^{(3)}))}}
           (a^{(2)}x_1)''\otimes
           \bigl(x_2\,S(a^{(3)})\bigr)''
           \displaybreak[0]\\
        & =\sum_{(a)}
           \sum_{(a^{(2)}x_1)}
           \sum_{(x_2S(a^{(3)}))}
           \bigeval{\tilde{f}}
                   {S^{-1}\bigl[(a^{(2)}x_1)'
           \bigl(x_2\,S(a^{(3)})\bigr)'\bigr]}\,
           (a^{(2)}x_1)''\otimes
           \bigl(x_2\,S(a^{(3)})\bigr)'' \\
        \intertext{and upon restoring $f$ this becomes}
        & \sum_{(a)}
          \sum_{(a^{(2)}x_1)}
          \sum_{(x_2S(a^{(3)}))}
          \bigeval{f}
                  {S^{-1}(a^{(4)})\,
                   S^{-1}\bigl[(a^{(2)}x_1)'
                               \bigl(x_2\,S(a^{(3)})\bigr)'\bigr]\,
                   a^{(1)}} \\
        & \relphantom{=}
           \hphantom{\sum_{(a)}
                     \sum_{(a^{(2)}x_1)}
                     \sum_{(x_2S(a^{(3)}))}}
           (a^{(2)}x_1)''\otimes
           \bigl(x_2\,S(a^{(3)})\bigr)''
           \displaybreak[0]\\
        & =\sum_{(a)}
           \sum_{(a^{(2)}x_1)}
           \sum_{(x_2S(a^{(3)}))}
           \bigeval{f}
                   {S^{-1}\bigl[(a^{(2)}x_1)'
                                \bigl(x_2\,S(a^{(3)})\bigr)'
                                a^{(4)}\bigr]\,
                    a^{(1)}}
           \displaybreak[0]\\
        & \relphantom{=}
           \hphantom{\sum_{(a)}
                     \sum_{(a^{(2)}x_1)}
                     \sum_{(x_2S(a^{(3)}))}}
           (a^{(2)}x_1)''\otimes
           \bigl(x_2\,S(a^{(3)})\bigr)''
           \displaybreak[0]\\
        & =\!\sum_{(a)(x_i)}
           \sum_{(a^{(2)})}
           \sum_{(S(a^{(3)}))}
           \bigeval{f}
                   {S^{-1}\bigl((a^{(2)})'
                                x_1'x_2'\,
                                S(a^{(3)})'
                                a^{(4)}\bigr)\,
                    a^{(1)}} \\
        & \relphantom{=}
           \hphantom{\!\sum_{(a)(x_i)}
                     \sum_{(a^{(2)})}
                     \sum_{(S(a^{(3)}))}}
           (a^{(2)})''
           x_1''\otimes
           x_2''\,
           S(a^{(3)})''
           \displaybreak[0]\\
        & =\!\sum_{(a)(x_i)}
           \sum_{(S(a^{(4)}))}\!
           \bigeval{f}
                   {S^{-1}\bigl(a^{(2)}
                                x_1'
                                x_2'\,
                                S(a^{(4)})'
                                a^{(5)}\bigr)\,
                    a^{(1)}}\,
           a^{(3)}
           x_1''\otimes
           x_2''\,
           S(a^{(4)})''
           \displaybreak[0]\\
        & =\!\sum_{(a)(x_i)}
           \sum_{(a^{(4)})}
           \bigeval{f}
                   {S^{-1}\bigl[a^{(2)}
                                x_1'
                                x_2'\,
                                S\bigl((a^{(4)})''\bigr)\,
                                a^{(5)}\bigr]\,
                    a^{(1)}}\,
           a^{(3)}
           x_1''\otimes
           x_2''\,
           S\bigl((a^{(4)})'\bigr)
           \displaybreak[0]\\
        & =\!\sum_{(a)(x_i)}\!
           f\bigl[S^{-1}\bigl(a^{(2)}
                              x_1'
                              x_2'\,
                              S(a^{(5)})\,
                              a^{(6)}\bigr)\,
                  a^{(1)}\bigr]\,
           a^{(3)}
           x_1''\otimes
           x_2''\,
           S(a^{(4)})
           \displaybreak[0]\\
        & =\!\sum_{(a)(x_i)}\!
           f\bigl[S^{-1}\bigl(a^{(2)}
                              x_1'
                              x_2'\,
                              \epsilon(a^{(5)})\bigr)\,
                  a^{(1)}\bigr]\,
           a^{(3)}
           x_1''\otimes
           x_2''\,
           S(a^{(4)})
           \displaybreak[0]\\
        & =\!\sum_{(a)(x_i)}\!
           f\bigl(S^{-1}(x_1'x_2')\,
                  S^{-1}(a^{(2)})\,
                  a^{(1)}\bigr)\,
           a^{(3)}
           x_1''\otimes
           x_2''\,
           S(a^{(4)})
           \displaybreak[0]\\
        & =\!\sum_{(a)(x_i)}\!
           f\bigl(S^{-1}(x_1'x_2')\,
                  \epsilon(a')\bigr)\,
           a''
           x_1''\otimes
           x_2''\,
           S(a''')
           \displaybreak[0]\\
        & =\!\sum_{(a)(x_i)}\!
           \eval{f}
                {S^{-1}(x_1'x_2')}\,
           a'
           x_1''\otimes
           x_2''\,
           S(a'').
    \end{align*}
    Note that in the third line we have employed the fact that $\Delta$
    is an algebra morphism. Since the above is true for all $x_1,x_2\in
    H$ we have just proven the claim.\qed
\end{proof}

\bigskip

\noindent
In order to obtain local $\mathrm{D}(H)$-representations at the sites of
arbitrary graphs we need to extend the above actions of $A_a$ and $B_f$
to larger vertices and faces. It will inevitably happen that
the two operators~$A_a$ and $B_f$ act on different sets of edges which do \emph{not} fully
coincide. We need to ensure that a)~both operators
continue to represent their respective parts~$H$ and $X$ individually
and b)~the commutation relation arising from common edges still implements the
bicrossed structure of the quantum double~$\mathrm{D}(H)$.

The extension can be done as follows for general vertices~$s$ and faces~$p$.
\begin{definition}
\label{def:general_module}
Let $(s,p)$ be a site of the graph~$\Gamma$, $a\in H$ and $f\in X$. Then
we define vertex and face operators via
\begin{align}
    A_a(s,p)\
    \vcenter{\hbox{\includegraphics{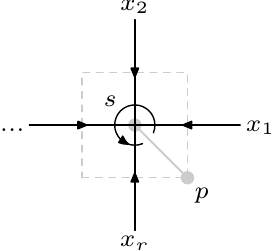}}}
    & :=\sum_{(a)}
       L_+^{a^{(1)}}(x_1)\otimes
       \dots\otimes
       L_+^{a^{(r)}}(x_r)
       \nonumber\\
    & =\sum_{(a)}\,
       \vcenter{\hbox{\includegraphics{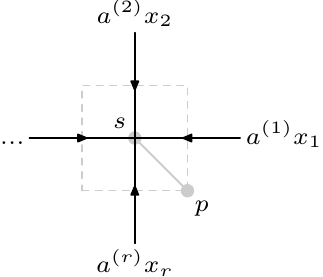}}}\ ,
       \label{eq:vertex_operator_simple}
       \displaybreak[0]\\
    B_f(s,p)\
    \vcenter{\hbox{\includegraphics{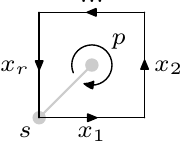}}}
    & :=\sum_{(f)}
       T_-^{f^{(r)}}(x_1)\otimes
       \dots\otimes
       T_-^{f^{(1)}}(x_r)
       \nonumber\\
    & =\sum_{(x_i)}
       f\bigl(S^{-1}(x_1')
              \cdots
              S^{-1}(x_r')\bigr)\
       \vcenter{\hbox{\includegraphics{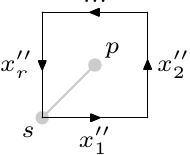}}}\ .
       \label{eq:face_operator_simple}
\end{align}
\end{definition}

Note that for $A_a(s,p)$ the vertex~$s$ denotes the centre of the loop
on the dual graph and $p$ denotes the starting point within the loop
for comultiplication.
Analogously, for $B_f(s,p)$ the starting point of the loop around the face~$p$
is marked by $s$. For different edge orientations the actions $L_+$ and $T_-$
in \eqref{eq:vertex_operator_simple} and \eqref{eq:face_operator_simple} may need to
be replaced per edge by $L_-$ and $T_+$ according to Figure~\ref{fig:kitaev_edge}.

The idea behind this is: first one fixes the action of $a\in H$ such
that the coproduct of $a$ is applied to the edges of the vertex~$s$ in
\emph{counterclock}wise order. This is indicated by the arrow winding
around $s$. When defining the action of $f\in X$ on a face it then
appears at first there might be two choices for the orientation of the
coproduct of $f$. It turns out that only one of them, namely
distributing the coproduct of $f$ in \emph{clock}wise orientation on the
graph, will yield a meaningful theory. The orientation of the coproduct
is again indicated by the arrow winding around $p$.

At this point the reader may wonder why one does not take the
function~$f$ directly from $H^*$ (the dual of $H$ \emph{without} flipped
comultiplication) and the following more symmetric definition of the
action~\eqref{eq:face_operator_simple}:
\begin{equation}
    \label{eq:face_operator_simple2}
    \begin{split}
        \widetilde{B}_f(s,p)\
        \vcenter{\hbox{\includegraphics{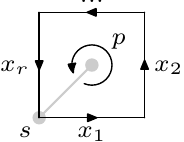}}}
        & :=\sum_{(f)}
           T_-^{f^{(1)}}(x_1)\otimes
           \dots\otimes
           T_-^{f^{(r)}}(x_r) \\
        & =\sum_{(x_i)}
           f\bigl(S^{-1}(x_1')
                  \cdots
                  S^{-1}(x_r')\bigr)\
           \vcenter{\hbox{\includegraphics{lattices2-2}}}\ .
    \end{split}
\end{equation}
Note that here the arrow winds around $p$ in \emph{counterclock}wise
orientation since $f\in H^*$ has a different coproduct now as compared
to~\eqref{eq:face_operator_simple}. Clearly,
both~\eqref{eq:face_operator_simple}
and~\eqref{eq:face_operator_simple2} define the same (algebra) action on the
boundary edges of the face.
However, as cumbersome as taking $f\in X$ in the definition of $B_f(s,p)$ may
seem it will make the next task tremendously easier: making
contact with the Drinfeld double $\mathrm{D}(H)=X\bowtie H$ which these
operators~$A_a$ and $B_f$ represent.

\begin{figure}
    \centering
    \includegraphics{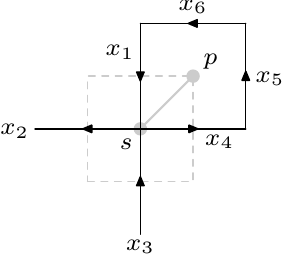}
    \caption{Graph used in the proof of Theorem~\ref{thm:local_modules}. The commutation relation and hence the
        structure of $\mathrm{D}(H)$ is determined on the intersection
        of the supports of $A_a(s,p)$ and $B_f(s,p)$.}
    \label{fig:kitaev_site}
\end{figure}

\begin{theorem}
    \label{thm:local_modules}
    Let $H$ a finite-dimensional Hopf algebra. Then each site~$(s,p)$
    of the graph~$\Gamma$ supports a $\mathrm{D}(H)$-module structure given by
    the operators~$A_a(s,p)$ and $B_f(s,p)$ from Definition~\ref{def:general_module}.
\end{theorem}
\begin{proof}
    In order to keep the notation
simple we only consider sites such as the one depicted in
Figure~\ref{fig:kitaev_site}. The generality of the argument will not be
affected.

    In order to prove the $\mathrm{D}(H)$-module structure at site~$(s,p)$ it is again
    enough to show that the straightening formula holds. Let $a\in H$,
    $f\in X$ and evaluate $A_aB_f$ on arbitrary edges $x_i\in H$:
    \begin{align*}
        A_a
        B_f
        (x_1\otimes
         \dots\otimes
         x_6)
        & =A_a
           \sum_{(x_i)}
           f\bigl(S^{-1}(x_4')\,
                  S^{-1}(x_5')\,
                  S^{-1}(x_6')\,
                  S^{-1}(x_1')\bigr) \\
        & \relphantom{=}
           \hphantom{A_a
                     \sum_{(x_i)}}
           x_1''\otimes
           x_2\otimes
           x_3\otimes
           x_4''\otimes
           x_5''\otimes
           x_6''
           \displaybreak[0]\\
        & =\sum_{(x_i)}\,
           f\bigl(S^{-1}(x_4')\,
                  S^{-1}(x_5')\,
                  S^{-1}(x_6')\,
                  S^{-1}(x_1')\bigr) \\
        & \relphantom{=}
           \sum_{(a)}
           a^{(1)}
           x_1''\otimes
           x_2\,
           S(a^{(2)})\otimes
           a^{(3)}
           x_3\otimes
           x_4''\,
           S(a^{(4)})\otimes
           x_5''\otimes
           x_6''
           \displaybreak[0]\\
        & =\!\!\sum_{(a)(x_i)}\!
           \eval{f}
                {S^{-1}(x_1'x_6'x_5'x_4')} \\
        & \relphantom{=}
           \hphantom{\!\!\sum_{(a)(x_i)}\!}
           a^{(1)}
           x_1''\otimes
           x_2\,
           S(a^{(2)})\otimes
           a^{(3)}
           x_3\otimes
           x_4''\,
           S(a^{(4)})\otimes
           x_5''\otimes
           x_6''. \\
    \end{align*}
    Compare this with the case:
    \begin{multline*}
        \sum_{(a)}
        B_{f(S^{-1}(a''')
           ?
           a')}
        A_{a''}
        (x_1\otimes
         \dots\otimes
         x_6) \\
        =\sum_{(a)}
         B_{f(S^{-1}(a^{(6)})
            ?
            a^{(1)})}\,
         a^{(2)}
         x_1\otimes
         x_2\,
         S(a^{(3)})\otimes
         a^{(4)}
         x_3\otimes
         x_4\,
         S(a^{(5)})\otimes
         x_5\otimes
         x_6.
    \end{multline*}
    Again, abbreviating $\tilde{f}=f\bigl(S^{-1}(a^{(6)})\,?a^{(1)}\bigr)$ we
    obtain for the right-hand side
    \begin{align*}
        & \sum_{(a)}
          \sum_{(x_4S(a^{(5)}))}
          \sum_{(x_5)(x_6)}
          \sum_{(a^{(2)}x_1)}
          \bigeval{\tilde{f}}
                  {S^{-1}\bigl[\bigl(x_4\,
                                     S(a^{(5)})\bigr)'\bigr]\,
                   S^{-1}(x_5')\,
                   S^{-1}(x_6')\,
                   S^{-1}\bigl((a^{(2)}x_1)'\bigr)} \\
        & (a^{(2)}x_1)''\otimes
          x_2\,S(a^{(3)})\otimes
          a^{(4)}x_3\otimes
          \bigl(x_4\,
                S(a^{(5)})\bigr)''\otimes
          x_5''\otimes
          x_6''
          \displaybreak[0]\\
        & =\sum_{(a)}
           \sum_{(x_i)}
           \sum_{(a^{(2)})}
           \sum_{(a^{(5)})}
           \bigeval{\tilde{f}}
                   {S^{-1}\bigl[x_4'\,
                                S\bigl((a^{(5)})''\bigr)\bigr]\,
                    S^{-1}(x_5')\,
                    S^{-1}(x_6')\,
                    S^{-1}\bigl((a^{(2)})'x_1'\bigr)} \\
        & \relphantom{=}
           (a^{(2)})''x_1''\otimes
           x_2\,S(a^{(3)})\otimes
           a^{(4)}x_3\otimes
           x_4''\,
           S\bigl((a^{(5)})'\bigr)\otimes
           x_5''\otimes
           x_6''
           \displaybreak[0]\\
        & =\!\sum_{(a)(x_i)}
           \sum_{(a^{(2)})}
           \sum_{(a^{(5)})}
           \bigeval{\tilde{f}}
                   {S^{-1}\bigl[(a^{(2)})'
                                x_1'
                                x_6'
                                x_5'
                                x_4'\,
                                S\bigl((a^{(5)})''\bigr)\bigr]} \\
        & \relphantom{=}
           \hphantom{\!\sum_{(a)(x_i)}
                     \sum_{(a^{(2)})}
                     \sum_{(a^{(5)})}}
           (a^{(2)})''
           x_1''\otimes
           x_2\,
           S(a^{(3)})\otimes
           a^{(4)}
           x_3\otimes
           x_4''\,
           S\bigl((a^{(5)})'\bigr)\otimes
           x_5''\otimes
           x_6'' \\
        \intertext{and now restoring $f$ this becomes}
        & \sum_{(a)(x_i)}\!
          \bigeval{f\bigl(S^{-1}(a^{(8)})\,
                          ?
                          a^{(1)}\bigr)}
                  {S^{-1}\bigl(a^{(2)}
                               x_1'
                               x_6'
                               x_5'
                               x_4'\,
                               S(a^{(7)})\bigr)} \\
        & \hphantom{\!\sum_{(a)(x_i)}\!}
          a^{(3)}
          x_1''\otimes
          x_2\,
          S(a^{(4)})\otimes
          a^{(5)}
          x_3\otimes
          x_4''\,
          S(a^{(6)})\otimes
          x_5''\otimes
          x_6'' \\
        & =\!\sum_{(a)(x_i)}\!
           \eval{f}
                {S^{-1}(a^{(8)})\,
                 a^{(7)}
                 S^{-1}(x_1'
                 x_6'
                 x_5'
                 x_4')\,
                 S^{-1}(a^{(2)})\,
                 a^{(1)}} \\
        & \relphantom{=}
           \hphantom{\!\sum_{(a)(x_i)}\!}
           a^{(3)}
           x_1''\otimes
           x_2\,
           S(a^{(4)})\otimes
           a^{(5)}
           x_3\otimes
           x_4''\,
           S(a^{(6)})\otimes
           x_5''\otimes
           x_6''
           \displaybreak[0]\\
        & =\!\sum_{(a)(x_i)}\!
           \eval{f}
                {S^{-1}(x_1'
                        x_6'
                        x_5'
                        x_4')\,
                 S^{-1}(a^{(2)})\,
                 a^{(1)}} \\
        & \relphantom{=}
           \hphantom{\!\sum_{(a)(x_i)}\!}
           a^{(3)}
           x_1''\otimes
           x_2\,
           S(a^{(4)})\otimes
           a^{(5)}
           x_3\otimes
           x_4''\,
           S(a^{(6)})\otimes
           x_5''\otimes
           x_6''
           \displaybreak[0]\\
        & =\!\sum_{(a)(x_i)}\!
           \eval{f}
                {S^{-1}(x_1'
                        x_6'
                        x_5'
                        x_4')} \\
        & \relphantom{=}
           \hphantom{\!\sum_{(a)(x_i)}\!}
           a^{(1)}
           x_1''\otimes
           x_2\,
           S(a^{(2)})\otimes
           a^{(3)}
           x_3\otimes
           x_4''\,
           S(a^{(4)})\otimes
           x_5''\otimes
           x_6''.
    \end{align*}
    This proves the statement.\qed
\end{proof}

\subsection{Hilbert space}
\label{sec:Hilbert}

In order to obtain a physical system from the above discussion we need to
define both a Hilbert space and a Hamiltonian for our topological lattice
model. In particular, we need to find Hopf $*$-algebras~$H$ that allow for an
inner product and $*$-representations. This means that adjoint operators
(which as usual are defined relative to the inner product) in a representation
of $H$ are compatible with the $*$-structure of $H$ itself.
It turns out we can reach both goals by requiring $H$ to be a finite-dimensional Hopf $C^*$-algebra.
Such a Hopf algebra comes endowed with a unique element called the Haar integral which will naturally define both the inner product and the Hamiltonian.

In order to define the Hilbert space we begin with the following proposition which is obtained from~\cite{Nill:1997p1989,Kodiyalam:2003p2638}.

\begin{proposition}
    \label{prop:Haar_integral}
    Let $H$ be a finite-dimensional Hopf $C^*$-algebra. Then $S^2=\id$ and there
    exists a unique two-sided integral $h\in H$ with the following
    properties:
    \begin{enumerate}
        \item $h^2=h$,
        \item $h^*=h$,
        \item $S(h)=h$,
        \item $h\in\Cocom(H)$.
    \end{enumerate}
    Furthermore, $H^*$ is a Hopf $C^*$-algebra again and its unique
    integral $\phi\in H^*$ satisfying 1--4 is a faithful positive functional,
    or trace, on $H$.
\end{proposition}

As a first consequence we can easily resolve the issue of edge orientation:
since the antipode is now involutive we define the reversal of an edge~$e$ simply
by
\begin{equation}
    \label{eq:orientation_reversal}
    x_e
    \mapsto
    S(x_e).
\end{equation}
This is obviously compatible with the actions~\eqref{eq:action1},
\eqref{eq:action2}, \eqref{eq:action3} and \eqref{eq:action4}. At the
same time we no longer
need to pick a distinguished pattern of edge orientations, rather all
patterns are equivalent.

Although the preceding proposition tells us that $\phi$ is a positive
trace on $H$, we need to know its precise relationship with the usual
trace $\tr_H(a)=\tr(L_+^a)$ found in the literature on Hopf algebras.
Setting $\abs{H}:=\dim H$ we have

\begin{lemma}
    Let $H$ a finite-dimensional Hopf $C^*$-algebra. Then
    \begin{equation}
        \tr_H
        =\abs{H}\cdot
         \phi
    \end{equation}
    holds where $\phi\in H^*$ is the Haar functional on $H$.
\end{lemma}
\begin{proof}
    By the above Proposition we know that $S^2=\id$. In this case we
    have
    \begin{equation}
        \tr_H
        =\epsilon(\Lambda)\,
         \phi
    \end{equation}
    for the integral $\Lambda\in H$ which is normalized such that
    $\phi(\Lambda) =1$~\cite{Schneider:1995p2245}. Since $\phi(h)=
    \abs{H}^{-1}$ for the Haar integral $h\in H$
    \cite{Kodiyalam:2003p2638}
    we actually have that $\Lambda=\abs{H}\,h$. Finally, $\epsilon(h)
    =1$ concludes the proof.\qed
\end{proof}

Since $\phi$ is a faithful positive trace on $H$ we can derive a
Hermitian inner product on $H$ from it by setting
\begin{equation}
    \label{eq:inner_product}
    (a,b)_H
    =\phi(a^*b).
\end{equation}

This inner product now turns both the module structures~$L_\pm$ and $T_\pm$ on $H$ into
$*$-representations. Indeed, relative to~\eqref{eq:inner_product} the
adjoint map of $L_\pm^a$ is given by $(L_\pm^a)^\dagger=L_\pm^{a^*}$ because of
\begin{equation}
    \bigl(x,L_+^a(y)\bigr)
    =\phi(x^*ay)
    =\eval{\phi}
          {(a^*x)^*b}
    =\bigl(L_+^{a^*}(x),
           y\bigr).
\end{equation}
An easy but tedious calculation shows that $(T_\pm^f)^\dagger=T_\pm^{f^*}$
holds, too. Remember that $f^*$ is given by~\eqref{eq:dual_star}.

This means that for the operators $A_a$ and $B_f$ which represent the
Drinfeld double~$\mathrm{D}(H)$ the adjoint operators are given by:
\begin{align}
    A_a^\dagger(s,p) & =A_{a^*}(s,p),
                        \label{eq:adjoint_vertex} \\
    B_f^\dagger(s,p) & =B_{f^*}(s,p).
                        \label{eq:adjoint_face}
\end{align}

\subsection{Hamiltonian}

It remains to specify a Hamiltonian for the model. In analogy
to~\cite{Kitaev:2003} we would like to get a frustration-free
Hamiltonian, i.e. a sum of commuting terms, and we would like to derive
it from the local operators~$A_a$ and $B_f$ defined previously. Hence we
need to identify a subset of these operators such that they mutually
commute with each other.

Before anything else it is natural to analyze the commutation relation
between $A_a$ and $B_f$ at the \emph{same} site~$(s,p)$ of the
graph~$\Gamma$. Suppose in the following that both $a\in H$ and $f\in X$
are such that $\eval{fg}{a}=\eval{gf}{a}$ for any $f,g\in X$ and $f(xy)
=f(yx)$ for any $x,y\in H$.%
\footnote{It is easy to see that this is equivalent to $a \in\Cocom(H)$
    and $f\in\Cocom(X)$. Clearly, $\Cocom(H)$ and $\Cocom(X)$ are
    defined relative to the respective coproducts in $H$ and $X$, which
    are fundamentally different.}
Then at the level of $\mathrm{D}(H)$ we have
\begin{equation*}
    \begin{split}
        af & =\sum_{(a)}
              f\bigl(S^{-1}(a''')\,
                     ?
                     a'\bigr)\,
              a'' \\
           & =\sum_{(a)}
              f\bigl(a'\,
                     S^{-1}(a''')\,
                     ?\bigr)\,
              a'' \\
           & =\sum_{(a)}
              f\bigl(a'''
                     S^{-1}(a'')\,
                     ?\bigr)\,
              a' \\
           & =\sum_{(a)}
              f\bigl(\epsilon(a'')\,
                     ?\bigr)\,
              a' \\
           & =fa,
    \end{split}
\end{equation*}
where we used the cocommutativity of $f$ (or $a$ respectively) in the second (third)
line and the skew-antipode~$S^{-1}$ in the fourth one. In other words,
the straightening formula~\eqref{eq:straighten} becomes trivial for such elements.
Since for a fixed site~$(s,p)$ the operators~$A_a(s,p)$ and $B_f(s,p)$
form a representation of $\mathrm{D}(H)$ this commutation relation
immediately carries over to:
\begin{equation}
    \label{eq:commutation}
    A_a(s,p)\,
    B_f(s,p)
    =B_f(s,p)\,
     A_a(s,p).
\end{equation}

Interestingly, with this choice of $a$ and $f$ the full specification
 of a site~$(s,p)$ becomes
unnecessary for the operators~$A_a$ and $B_f$. Indeed, since coproducts can now be permuted cyclically, the starting point~$p$ of the dual loop used to define the vertex operator~$A_a(s,p)$ loses any significance. All that matters for the action is the vertex~$s$ itself. In order to reflect this fact we set for all $a\in\Cocom(H)$:
\begin{equation}
    A_a(s)
    :=A_a(s,p).
\end{equation}
By the same token a face operator~$B_f(s,p)$ only cares about the face~$p$ once we choose $f\in\Cocom(X)$:
\begin{equation}
    B_f(p)
    :=B_f(s,p).
\end{equation}

Having restricted the possible candidates for the terms in the Hamiltonian by exploiting the $\mathrm{D}(H)$-module structure at a single fixed site we need to ensure subsequently that \emph{all} vertex operators~$A_a(s)$ commute among themselves, too. Observe first that
\begin{equation}
    [A_a(s),A_b(s')]
    =0
\end{equation}
for any $a,b\in H$ whenever two vertices~$s$ and $s'$ do not coincide. Indeed, if $A_a(s)$ acts on a common edge via $L_+$ then $A_b(s')$ acts on the same edge via $L_-$ and vice versa, hence both operators commute by~\eqref{eq:commutation_L}. If there are no common edges then there is nothing to show. On the other hand if $s=s'$ then the edge set $E(s)$ has an $H$-module structure (inherited from the $\mathrm{D}(H)$-module structure at any site that contains $s$) which implies
\begin{equation}
    [A_a(s),A_b(s)]
    =A_{[a,b]}(s).
\end{equation}
This suggests to further narrow down our set of candidate vertex operators $A_a(s)$ by additionally requiring $a\in\mathrm{Z}(H)$ where $\mathrm{Z}(H)$ denotes the centre of the Hopf algebra~$H$. By analogy, we are naturally led to assume $f\in\mathrm{Z}(X)\cap\Cocom(X)$ in the following if we want all face operators~$B_f(p)$ to commute.

Finally we would like to remark that $A_a(s)$ and $B_f(p)$ trivially commute if
the pair~$(s,p)$ is not a site. In fact, the sets of edges they act on are even disjoint in this case.

\bigskip
\noindent
Recalling the properties of the Haar integral from Proposition~\ref{prop:Haar_integral} we now state the main result of this section.
\begin{theorem}[Generalized quantum double model]
    \label{thm:model}
    Let $H$ a finite-dimensional Hopf $C^*$-algebra with Haar integral~$h$
    and Haar functional~$\phi$ and let $\Gamma$ a graph.
    Furthermore for each $s\in V$ and $p\in F$ define the projectors
    \begin{align}
        A(s) & :=A_h(s,p), \\
        B(p) & :=B_\phi(s,p).
    \end{align}
    
    Then
    \begin{equation}
        \mathcal{H}
        =-\sum_{s\in V}
         A(s)
         -\sum_{p\in F}
         B(p)
    \end{equation}
    is a local, frustration-free Hamiltonian defining the
    $\mathrm{D}(H)$-model.
\end{theorem}
\begin{proof}
    First observe that
    \begin{equation}
        A(s)^2
        =A_{h^2}(s,p)
        =A(s)
    \end{equation}
    by Proposition~\ref{prop:Haar_integral}. The same argument shows that $B(p)$
    is a projector, too.
    
    By the preceding discussion and Proposition~\ref{prop:Haar_integral}
    it is clear that all local terms~$A(s)$ and $B(p)$ commute with each other.
    Furthermore, \eqref{eq:adjoint_vertex} and \eqref{eq:adjoint_face} imply that they are Hermitian.\qed
\end{proof}

For the reader's convenience we give a short summary of how these operators act on the graph:
\begin{align}
    A(s)\
    \vcenter{\hbox{\includegraphics{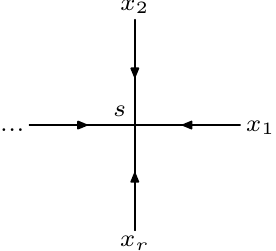}}}
    & =\sum_{(h)}\,
       \vcenter{\hbox{\includegraphics{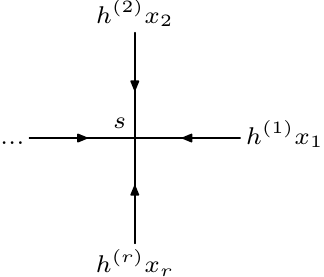}}}\ ,
       \label{eq:vertex_Hopf}\\
    B(p)\
    \vcenter{\hbox{\includegraphics{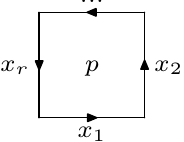}}}
    & =\sum_{(x_i)}
       \phi(x_r'\cdots x_1')\
       \vcenter{\hbox{\includegraphics{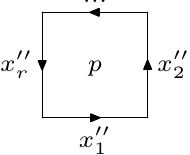}}}\ .
       \label{eq:face_Hopf}
\end{align}
Alternatively, for a different orientation of the graph edges the
face operator acts as follows:
\begin{equation}
    B(p)\
    \vcenter{\hbox{\includegraphics{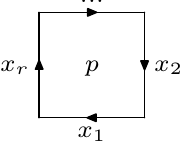}}}
    =\sum_{(x_i)}
     \phi(x_1''\cdots x_r'')\
     \vcenter{\hbox{\includegraphics{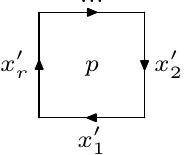}}}\ .
\end{equation}

It is no coincidence that already at this stage do the Haar integrals of $H$ and $H^*$ reveal themselves as the crucial
ingredients for the generalized quantum double models. Since we set out to construct
a quantum spin model whose (elementary) quasiparticle excitations are
characterized by irreducible representations of $\mathrm{D}(H)$ the ground
state sector with no quasiparticles present necessarily has the structure
of a trivial representation of $\mathrm{D}(H)$ locally. Precisely this is what the element
$\phi\otimes h\in\mathrm{D}(H)$ embodies. As the Hamiltonian intrinsically
encodes information about the ground state sector it should not surprise the
reader that the Haar integrals~$h$ and $\phi$ play such a prominent role.

\subsection{Comparison with $\mathrm{D}(\mathbb{C}G)$-models}

We would now like to briefly comment on the relation between our generalized
quantum double models and Kitaev's original construction in~\cite{Kitaev:2003}.

First note that in the case $H=\mathbb{C}G$ the terms given by~\eqref{eq:vertex_Hopf} and \eqref{eq:face_Hopf} in the above Hamiltonian
reduce to the operators of~\eqref{eq:vertex_group} and \eqref{eq:face_group}
which are the ones Kitaev employed for his quantum double models based on a
group~$G$. This becomes clear from Section~\ref{sec:trivial_algebras} where the relevant expressions for $h$ and $\phi$ are listed.

Secondly, the entire theory of ribbon operators readily carries over
from~\cite{Kitaev:2003} to our generalized quantum double models. This is because
ribbon operators
are constructed from certain elementary operators associated with two types of triangles
any given ribbon path decomposes into. These operators are nothing but the $L_\pm$
and $T_\pm$ that implement the $H$- and $X$-module structures. Furthermore patching together a ribbon operator from those elementary pieces only involves the structure
maps of the Hopf algebra~$H$ itself. Actually, we even used ribbon paths
in order to define the local $\mathrm{D}(H)$-module structures (and hence
the Hamiltonian) without saying so.

Finally, our generalized quantum double models inherit all the beautiful topological properties of the original since these follow exclusively
from the algebraic structure of ribbon operators. In particular, these features
include the degeneracy of the ground state sector as well as the exotic
statistics of the quasiparticle excitations whose anyonic nature is revealed
via braiding and fusion operations.

\section{A hierarchy of topological tensor network states}
\label{sec:networks}

In this section we develop a general diagrammatic language for tensor
network states built from finite-dimensional Hopf $C^*$-algebras. Underlying
surfaces both with and without boundaries are considered and we show
how to naturally describe subsystems. Using
this framework we solve the generalized quantum double model introduced in the preceding section by providing a tensor network representation for
one of its ground states. Any other energy eigenstate can be obtained from there by an appropriate ribbon operator. The tensor network
representation for that ground state only involves the canonical structures of the underlying
Hopf $C^*$-algebra~$H$ and its dual: multiplication, comultiplication, antipode and Haar integral. This insight leads us to propose a novel hierarchy of tensor network
states based on Hopf subalgebras. For $H=\mathbb{C}G$ we are able to
completely classify this hierarchy of states in terms of charge condensation.
Finally we describe the relation between our Hopf tensor network language
and the usual formulation of {\sc PEPS}.

Unless otherwise noted, from now on $H$ will be a finite-dimensional Hopf $C^*$-algebra
with Haar integral $h\in H$ and Haar functional $\phi\in H^*$.

\subsection{Tensor traces}

As we discussed in the introduction, the fully contracted tensor network
(which is a complex number) for a certain ground
state of the
$\mathrm{D}(\mathbb{C}G)$-model on the graph~$\Gamma$ can be interpreted
as a collection of virtual loops in the faces of $\Gamma$ that are
suitably glued together to form the physical degrees of freedom.

We can extend this idea to the case of any finite-dimensional Hopf $C^*$-algebra~$H$
now. In each face $p\in F$ place a virtual loop oriented in
counterclockwise direction and associate a function $f_p\in X$ to this
loop. For the moment we may restrict to $f_p\in\Cocom(X)$. With each
oriented edge $e\in E$ we associate an algebra element $x_e\in H$ which
splits into two parts as follows:
\begin{equation}
    \label{eq:edge_split}
    \bigl((S\otimes\id)\circ
          \Delta\bigr)
    (x_e)
    =\sum_{(x_e)}
     S(x_e')\otimes
     x_e''.
\end{equation}
Subsequently, we attribute $x_e''$ to the left adjacent face of $e$ and
$S(x_e')$ to the right one. A virtual loop in face~$p$ is then evaluated
by taking the clockwise product of all elements thus associated
with the loop from the surrounding edges of $p$. The result is then fed
to the function~$f_p$.

In order to simplify arguments we introduce a diagrammatic notation
which encodes calculations with these tensor networks. In diagrammatic
language the evaluation rule just described reads
\begin{equation}
    \label{eq:simple_rule1}
    \vcenter{\hbox{\includegraphics{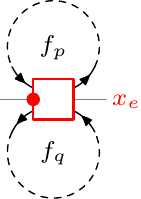}}}
    :=\sum_{(x_e)}
      f_p\bigl(S(x_e')\dots\bigr)\,
      f_q(x_e''\dots).
\end{equation}
Since $f_p$ is assumed cocommutative its argument can be permuted
cyclically and we may start both virtual loops around $p$ and $q$ at
the edge~$e$ without loss of generality. The clockwise order
of the product remains important though. Both the dashed lines and the
ellipses denote the remaining degrees of freedom of the faces~$p$ and
$q$ respectively. Note that the glueing procedure
generalizing~\eqref{eq:glueing} is implemented by
the coproduct. We will introduce a full description of this later
in~\eqref{eq:rule4} and~\eqref{eq:rule5}. Finally, a reversed edge is resolved via
\begin{equation}
    \label{eq:simple_rule2}
    \vcenter{\hbox{\includegraphics{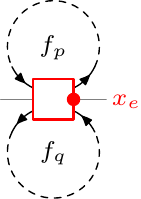}}}
    :=\vcenter{\hbox{\includegraphics{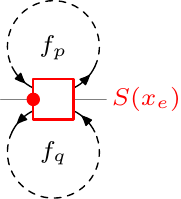}}}.
\end{equation}

In summary we have the following
\begin{definition}[Hopf tensor trace without boundary]
    \label{def:simple_Hopf_ttr}
    For each $p\in F$ let $f_p\in\Cocom(X)$. Then the
    \emph{Hopf tensor trace} associated with the graph~$\Gamma$ is the
    function $\ttr_\Gamma:H^{\otimes\abs{E}}\otimes X^{\otimes\abs{F}}
    \to\mathbb{C}$,
    \begin{equation}
        \bigotimes_{e\in E}
        x_e
        \bigotimes_{p\in F}
        f_p
        \mapsto
        \ttr_\Gamma(\{x_e\};
                    \{f_p\})
    \end{equation}
    which is defined via diagrams and the evaluation
    rules~\eqref{eq:simple_rule1} and \eqref{eq:simple_rule2}.
\end{definition}
From this definition it is clear that the tensor trace is linear in each
argument. It can be regarded as the wavefunction amplitude of a
quantum many-body state as we will see shortly.

\bigskip

\noindent
In the following we would like to investigate Hopf tensor networks
on graphs which are embedded in a surface~$M$ \emph{with} boundary as
well as Hopf tensor networks at interfaces between
regions~$R$ of a surface with or without boundary.

\begin{figure}
    \centering
    \begin{equation*}
        \vcenter{\hbox{\includegraphics{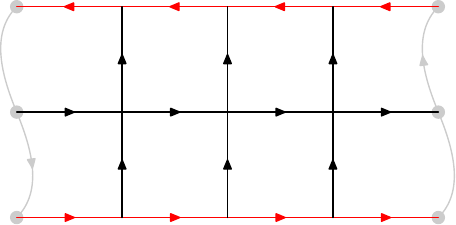}}}\
        \simeq\
        \vcenter{\hbox{\includegraphics{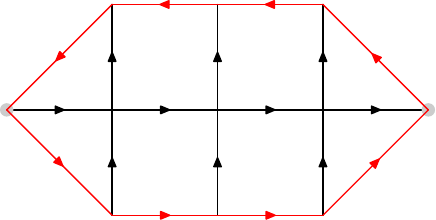}}}
    \end{equation*}
    \caption{Ordinary faces vs. boundary “faces”. Wherever the boundary
        of the surface~$M$ does not coincide with an edge of the graph~$\Gamma$ it is
        displayed in grey. Upon smooth deformation of this boundary the
        grey vertices on either side can actually be identified with each other.
        This modifies $\Gamma$ and thus creates new faces whose preimages
        on the left-hand side we call boundary “faces”.
        Boundary edges are drawn in red.}
    \label{fig:boundary_faces}
\end{figure}

In general the edge set $E=E_M$ of the graph~$\Gamma$ embedded in $M$
may be decomposed into the disjoint union of interior edges
$E_{M^\circ}$ and boundary edges $E_{\partial M}$ (see Figure~\ref{fig:boundary_faces}). On the other hand,
the set of faces $F=F_M=F_{M^\circ}$ only contains interior faces by
definition. However, one may collect incomplete faces (where two
vertices of degree~$1$ can be connected/merged to yield a genuine new
face) into the set $F_{\partial M}$ of boundary “faces” as illustrated
in Figure~\ref{fig:boundary_faces}. By construction
these are not faces of the original graph, hence $F\cap F_{\partial M}
=\emptyset$.

Furthermore
we may collect all boundary segments into the set $X_{\partial M}
=E_{\partial M}\cup F_{\partial M}$. Once we fix a distinguished segment
on the boundary this set is equipped with a natural ordering inherited
from the orientation of $\partial M$. Without loss of generality we will
take any boundary to be oriented in counterclockwise direction with
respect to its interior for the remainder of the article.

In order to take nontrivial boundaries into account we need to refine
our diagrammatic notation for tensor networks. Pick an arbitrary edge
$e\in E$. Since it belongs to the boundary of at least one face~$p$ in either
$F$ or $F_{\partial M}$ (i.e. $p$ is either an interior or a boundary face) we may define the following elementary diagram
for any elements $f_p\in X$ and $x_e\in H$ simply as their canonical pairing:
\begin{equation}
    \vcenter{\hbox{\includegraphics{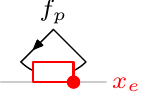}}}
    :=f_p(x_e).
    \label{eq:rule1}
\end{equation}
Different orientations of graph edges
and virtual loops are resolved via
\begin{align}
    \vcenter{\hbox{\includegraphics{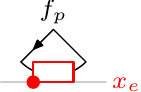}}}
    & :=\vcenter{\hbox{\includegraphics{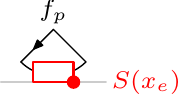}}}
        \label{eq:rule2}\\
    \vcenter{\hbox{\includegraphics{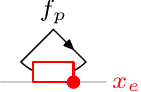}}}
    & :=\vcenter{\hbox{\includegraphics{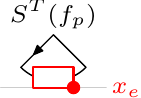}}}
        \label{eq:rule3}
\end{align}
where all these elementary diagrams are assumed to be invariant under arbitrary
rotations, for instance:
\begin{equation}
    \vcenter{\hbox{\includegraphics{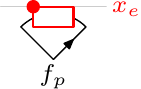}}}
    =\vcenter{\hbox{\includegraphics{lattices-80}}}
\end{equation}
Note also that $S^T=(S^{-1})^T$ is the correct antipode of $X
=(H^\mathrm{op})^*$. As a consequence, any of the above elementary diagrams
has the same value as its mirror
image under reflection about a vertical axis. In fact, \eqref{eq:rule2}
and \eqref{eq:rule3} are mirror images of each other in that sense.
Furthermore for some $f_p$ (such as the dual Haar integral~$\phi$) it may happen
that $S^T(f_p)=f_p$ and the loop orientation may become
unimportant.

If the face~$p$ has edges other than $e$ in its boundary we
may extend the above diagrams as follows. Pick another edge $e'\in E(p)$
which shares a common vertex with $e$. Then for any $x_{e'}\in H$ we
define a “virtual” glueing operation by
\begin{equation}
    \label{eq:rule4}
    \vcenter{\hbox{\includegraphics{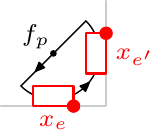}}}
    :=\sum_{(f_p)}\,
      \vcenter{\hbox{\includegraphics{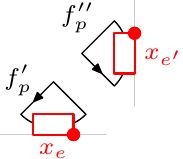}}}
    =f_p(x_{e'}
         x_e)
\end{equation}
where the arrows indicate the order in which the coproduct of $f_p$
is applied to the elementary diagrams. The black dot denotes the origin
for this comultiplication. While this origin is
very important in general it can be neglected iff $f_p$ is cocommutative.
In that case we will simply omit the corresponding dot from the diagram
as we did implicitly in Definition~\ref{def:simple_Hopf_ttr}. In any case
one needs to pay attention to the correct comultiplication in
$X$ which causes the product around the edges of $p$ to be taken in clockwise order.

Finally for any interior edge $e\in E_{M^\circ}$ with adjacent
interior or boundary faces $p,q\in F\cup F_{\partial M}$ we pick $x_e\in H$
and $f_p,f_q\in X$ arbitrarily and define a “physical” glueing
operation by
\begin{equation}
    \label{eq:rule5}
    \vcenter{\hbox{\includegraphics{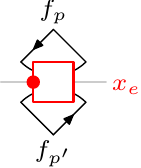}}}
    :=\sum_{(x_e)}\,
      \vcenter{\hbox{\includegraphics{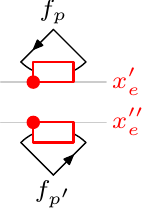}}}
\end{equation}
where the order of comultiplication is determined by the orientation
of the underlying graph edge. Consequently one has for instance
\begin{multline}
    \vcenter{\hbox{\includegraphics{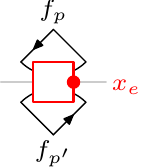}}}
    =\sum_{(x_e)}\,
     \vcenter{\hbox{\includegraphics{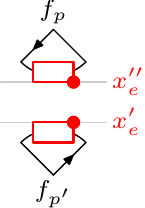}}}
    =\sum_{(x_e)}\,
     \vcenter{\hbox{\includegraphics{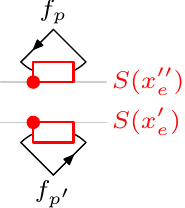}}} \\
    =\!\!\sum_{(S(x_e))}
     \vcenter{\hbox{\includegraphics{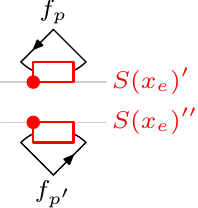}}}
    =\vcenter{\hbox{\includegraphics{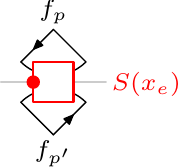}}}
\end{multline}
which is perfectly compatible with~\eqref{eq:orientation_reversal} as expected.
In contrast to elementary diagrams such a composite diagram may only be invariant under rotation
by $\pi$ if $S(x_e)=x_e$. Also note that in general a simultaneous reversal
of both virtual loops is \emph{not} given by reflecting a composite diagram about its
vertical axis:
\begin{equation}
    \vcenter{\hbox{\includegraphics{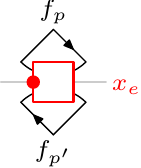}}}
    =\sum_{(x_e)}\,
     \vcenter{\hbox{\includegraphics{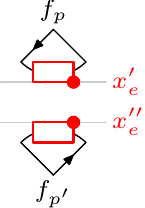}}}
    \neq\sum_{(x_e)}\,
     \vcenter{\hbox{\includegraphics{lattices-91}}}
    =\vcenter{\hbox{\includegraphics{lattices-90}}}\ .
\end{equation}
Rather, equality holds iff $x_e\in\Cocom(H)$. This should be compared
with~\eqref{eq:tensor_reflected}.

\bigskip

\noindent
Thus by starting from the elementary diagram~\eqref{eq:rule1} we have reexpressed the evaluation
rule~\eqref{eq:simple_rule1} for an interior edge entirely in terms of
“virtual”~\eqref{eq:rule4} and “physical”~\eqref{eq:rule5} glueing
operations. These are given by comultiplication in $X$ and $H$
respectively.
This means that an interior edge is formed by appropriately glueing two
boundary edges together via comultiplication. Alternatively, one may regard
this as glueing together two virtual interior loops. Virtual loops
themselves are assembled via glueing together smaller loop pieces.

\begin{figure}
    \centering
    \includegraphics{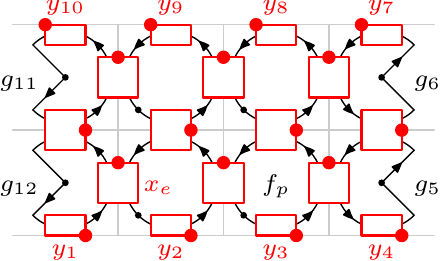}
    \caption{Diagram encoding the tensor trace $\ttr(\{x_e\};\{f_p\};
        \{y_e\};\{g_q\})$. While the interior degrees of freedom $\{x_e\}
        \subset H$ and $\{f_p\}\subset X$ are only shown partially
        the boundary degrees of freedom $\{y_e\}\subset H$ and $\{g_q\}
        \subset X$ are labelled in such a way that the ordering of the
        boundary is evident.}
    \label{fig:ttr}
\end{figure}

In summary we have the following general
\begin{definition}[Hopf tensor trace]
    \label{def:Hopf_ttr}
    Let $E_{M^\circ}$ and $F_M$ the sets of interior edges and faces
    respectively, and accordingly $E_{\partial M}$ and $F_{\partial M}$
    the sets of boundary edges and faces.
    
    The \emph{Hopf tensor trace} associated with the graph~$\Gamma$ is
    the function $\ttr_\Gamma:H^{\otimes\abs{E_{M^\circ}}}\otimes
    X^{\otimes\abs{F_M}}\otimes H^{\otimes\abs{E_{\partial M}}}\otimes
    X^{\otimes\abs{F_{\partial M}}}\to\mathbb{C}$,
    \begin{equation}
        \bigotimes_{e\in E_{M^\circ}}\!\!
        x_e
        \bigotimes_{p\in F_M}\!\!
        f_p
        \bigotimes_{e\in E_{\partial M}}\!\!
        y_e
        \bigotimes_{q\in F_{\partial M}}\!\!
        g_q
        \mapsto
        \ttr_{\Gamma}(\{x_e\};
                      \{f_p\};
                      \{y_e\};
                      \{g_q\})
    \end{equation}
    which is defined via diagrams and the evaluation
    rules~\eqref{eq:rule1}, \eqref{eq:rule2}, \eqref{eq:rule3},
    \eqref{eq:rule4} and \eqref{eq:rule5}.
\end{definition}

Note that this generalized Hopf tensor trace reduces to the Hopf tensor
trace for a surface without boundary (as in
Definition~\ref{def:simple_Hopf_ttr}) in the natural way by restricting
any virtual loop function~$f_p$ to be cocommutative and setting:
\begin{equation}
    \ttr_\Gamma(\{x_e\};
                \{f_p\})
    :=\ttr_\Gamma(\{x_e\};
                  \{f_p\};
                  \emptyset;
                  \emptyset).
\end{equation}

\subsection{Quantum states}
\label{sec:Hopf_TN_states}

So far we have merely defined a particular, fully contracted tensor network
(which is a complex number) with \emph{no}
reference to a quantum many-body state whatsoever. We now take the
next step and use the tensor trace above to generate actual
quantum states in a remarkably straightforward fashion:

\begin{definition}[Hopf tensor network state]
    Let $x_e,y_e\in H$ and $f_p,g_q\in X$ as in Definition~\ref{def:Hopf_ttr}.
    Let $\Gamma$ the graph embedded in the surface~$M$.
    \begin{enumerate}
        \item
        If $\partial M\neq\emptyset$ then
        \begin{multline}
            \ket{\psi_\Gamma(\{x_e\};
                             \{f_p\};
                             \{y_e\};
                             \{g_q\})} \\
            :=\!\!\sum_{\substack{(x_e)\\
                        e\in E_{M^\circ}}}
             \sum_{\substack{(y_e)\\
                   e\in E_{\partial M}}}\!\!
            \ttr_\Gamma(\{x_e''\};
                        \{f_p\};
                        \{y_e''\};
                        \{g_q\})
            \bigotimes_{e\in E_{M^\circ}}\!\!
            \ket{x_e'}
            \bigotimes_{e\in E_{\partial M}}\!\!
            \ket{y_e'}.
        \end{multline}
        \item
        If $\partial M=\emptyset$ then
        \begin{equation}
            \label{eq:Hopf_TN_state_compact}
            \ket{\psi_\Gamma(\{x_e\};
                             \{f_p\})}
            :=\!\sum_{\substack{(x_e)\\
                      e\in E_M}}\!
            \ttr_\Gamma(\{x_e''\};
                        \{f_p\})
            \bigotimes_{e\in E_M}\!
            \ket{x_e'}.
        \end{equation}
    \end{enumerate}
    In both cases we call the resulting state a \emph{Hopf tensor
    network state} on the graph~$\Gamma$.
\end{definition}

Given an arbitrary region~$R$ within the surface~$M$ it is
straightforward to partition such a quantum state into interior and
exterior parts accordingly. In the language of Hopf tensor network
states this simply amounts to cutting virtual loops. Without loss of
generality we will show this for a graph embedded in a surface without
boundaries.

Before proceeding we need to clarify though what a region~$R$ of a surface
discretized by the graph~$\Gamma$ is supposed to mean. We will only consider
regions whose boundary does not cross any edge of $\Gamma$. Equivalently,
such a boundary may only run through vertices and faces or \emph{along}
edges of $\Gamma$. Any part of the boundary which coincides with a graph edge
will be called \emph{smooth} and otherwise \emph{rough}. This coincides with
the way a graph may be embedded in a surface with boundary (see
Figure~\ref{fig:boundary_faces}). It is clear that a valid region~$R$ is
determined by a subgraph~$\Gamma_R$ up to deformations of rough boundaries,
hence we will define a region by its associated subgraph. This means that
a partition of M into $R$ and its complement~$\bar{R}$ naturally divides $E$
into the disjoint sets~$E_R$ and $E_{\bar{R}}$. At the same time it divides
$F$ into the disjoint sets~$F_R$, $F_{\bar{R}}$ and $F_{\partial R}$ or,
put differently, boundary faces~$F_{\partial R}$ are created which do \emph{not}
belong to either subgraph~$\Gamma_R$ or $\Gamma_{\bar{R}}$. In contrast,
there are \emph{no} boundary edges. Hence the (discretized) boundary of the region~$R$
coincides with $F_{\partial R}$ and has a length of $\abs{F_{\partial
R}}$. By abuse of notation we will call this length simply
$\abs{\partial R}$.

\begin{proposition}[Partitions]
    \label{prop:partitions}
    Let $\ket{\psi_\Gamma(\{x_e\};\{f_p\})}$ be a Hopf tensor network
    state on the graph~$\Gamma$ and $R$ an arbitrary region which is
    associated to the subgraph $\Gamma_R=(V_R,E_R,F_R)\subset\Gamma$.
    Let $g_q\in X$ and define
    \begin{equation}
        \ket{\psi_R(\{g_q\})}
        :=\ket{\psi_{\Gamma_R}(\{x_e\}_R;
                               \{f_p\}_R;
                               \emptyset;
                               \{g_q\})}
    \end{equation}
    where $\{x_e\}_R=\{x_e\mid e\in E_R\}$ and $\{f_p\}_R=\{f_p\mid p\in
    F_R\}$ are the natural restrictions to $\Gamma_R$. Then
    \begin{equation}
        \label{eq:decomposition}
        \ket{\psi_\Gamma(\{x_e\};
                         \{f_p\})}
        =\!\sum_{\substack{(f_q)\\
                 q\in F_{\partial R}}}\!
         \ket{\psi_R(\{f_q'\})}\otimes
         \ket{\psi_{\bar{R}}(\{f_q''\})}.
    \end{equation}
\end{proposition}
\begin{proof}
    Obvious from the complete diagrams and evaluation rule~\eqref{eq:rule4}
    which needs to be applied backwards.\qed
\end{proof}

\bigskip

\noindent
Building on this general framework of Hopf tensor network states we can
solve the generalized quantum double model now. Namely, we identify a particular
Hopf tensor network state as a ground state of the model. As such
this state is directly seen to be topologically ordered.
Note that we will only make use of the
structure maps of $H$ as well as the Haar integral and its dual to describe
the state. Again
we assume a surface without boundaries.

\begin{theorem}[Ground state of the generalized quantum double model]
    \label{thm:Hopf_gs}
    Let $h\in H$ and $\phi\in X$ the respective Haar integrals. The
    state
    \begin{equation}
        \ket{\psi_\Gamma}
        :=\ket{\psi_\Gamma(h,\dots,h;
          \phi,\dots,\phi)}
    \end{equation}
    is a ground state of the $\mathrm{D}(H)$-model.
\end{theorem}
\begin{proof}
    Since the Hamiltonian of the $\mathrm{D}(H)$-model is a sum of
    local, commuting terms by Theorem~\ref{thm:model} it is enough to show that each
    operator~$A(s)$ and $B(p)$ leaves the state~$\ket{\psi_\Gamma}$
    invariant individually. In order to do so we may partition
    $\ket{\psi_\Gamma}$ into an interior part corresponding to the
    support of such an operator and an exterior part. Both parts are
    glued via comultiplication in $X$. It will then suffice to prove
    that this interior part remains unchanged by
    either $A(s)$ or $B(p)$ respectively.

    Hence consider a face $p\in F$ with a boundary consisting of
    $r$~edges. The interior part of $\ket{\psi_\Gamma}$ is given by
    \begin{align}
        \ket{\psi_p(f_1,\dots,f_r)}
        & =\sum_{(h_i)}\,
           \vcenter{\hbox{\includegraphics{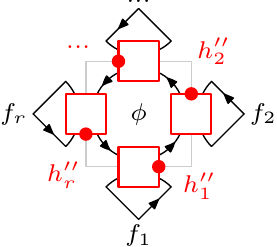}}}\quad
           \vcenter{\hbox{\includegraphics{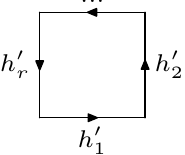}}}
           \label{eq:face1} \\
        & =\sum_{(h_i)}
           \phi(h_r'''\cdots h_1''')
           \prod_{j=1}^r
           f_j\bigl(S(h_j'')\bigr)
           \ket{h_1'}\otimes
           \dots\otimes
           \ket{h_r'}.
           \label{eq:face2}
    \end{align}
    with the Haar integrals $h_i=h\in H$ and arbitrary $f_i\in X$. It is
    invariant under the action of $B(p)$ as can be seen from
    \begin{align*}
        & B(p)
        \ket{\psi_p(f_1,
                    \dots,
                    f_r)} \\
        & =\sum_{(h_i)}
           \phi(h_r^{(4)}\cdots h_1^{(4)})\,
           \phi(h_r^{(1)}\cdots h_1^{(1)})
           \prod_{j=1}^r
           f_j\bigl(S(h_j^{(3)})\bigr)
           \ket{h_1^{(2)}}\otimes
           \dots\otimes
           \ket{h_r^{(2)}}
           \displaybreak[0]\\
        & =\sum_{(h_i)}
           \phi(h_r^{(3)}\cdots h_1^{(3)})\,
           \phi(h_r^{(4)}\cdots h_1^{(4)})
           \prod_{j=1}^r
           f_j\bigl(S(h_j^{(2)})\bigr)
           \ket{h_1^{(1)}}\otimes
           \dots\otimes
           \ket{h_r^{(1)}}
           \displaybreak[0]\\
        & =\sum_{(h_i)}
           \phi^2(h_r'''\cdots h_1''')
           \prod_{j=1}^r
           f_j\bigl(S(h_j'')\bigr)
           \ket{h_1'}\otimes
           \dots\otimes
           \ket{h_r'} \\
        & =\ket{\psi_p(f_1,\dots,f_r)}.
    \end{align*}

    Next consider a vertex $s\in V$ with $r$ attached edges. In this
    case the interior part of $\ket{\psi_\Gamma}$ reads
    \begin{align}
        \ket{\psi_s(f_1,\dots,f_r)}
        & =\sum_{(h_i)}\,
           \vcenter{\hbox{\includegraphics{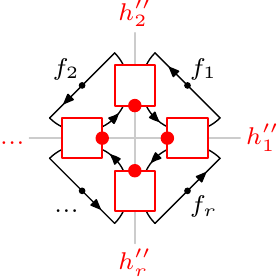}}}\quad
           \vcenter{\hbox{\includegraphics{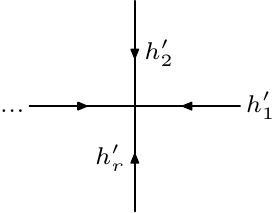}}}
           \label{eq:vertex1}\\
        & =\sum_{(h_i)}
           \prod_{j=1}^r
           f_j\bigl(S(h_j'')\,
                    h_{j+1}'''\bigr)
           \ket{h_1'}\otimes
           \dots\otimes
           \ket{h_r'}.
           \label{eq:vertex2}
    \end{align}
    Note that the orientation of the graph edges can always be reduced
    to the above setting using~\eqref{eq:orientation_reversal} for kets
    and~\eqref{eq:rule2} for diagrams. Now we obtain
    \begin{align*}
        & A(s)
          \ket{\psi_s(f_1,\dots,f_r)} \\
        & =\sum_{(h_i)}
           \prod_{j=2}^{r-1}
           f_j\bigl(S(h_j'')\,
                    h_{j+1}'''\bigr)
           \sum_{(h)}
           f_1\bigl(S(h_1'')\,
                    h_2'''\bigr)\,
           f_r\bigl(S(h_r'')\,
                    h_1'''\bigr) \\
        & \relphantom{=}
           \ket{h^{(1)}h_1'}\otimes
           \dots\otimes
           \ket{h^{(r)}h_r'}
           \displaybreak[0]\\
        & =\sum_{(h_i)}
           \prod_{j=3}^{r-1}
           f_j\bigl(S(h_j'')\,
                    h_{j+1}'''\bigr)
           \sum_{(h)}
           f_1\bigl(S(h_1'')\,
                    h^{(2)}
                    h_2'''\bigr)\,
           f_2\bigl(S(h_2'')\,
                    h_3'''\bigr)\,
           f_r\bigl(S(h^{(1)}h_r'')\,
                    h_1'''\bigr) \\
        & \relphantom{=}
           \ket{h_1'}\otimes
           \ket{h^{(3)}h_2'}\otimes
           \dots\otimes
           \ket{h^{(r+1)}h_r'}
           \displaybreak[0]\\
        & =\sum_{(h_i)}
           f_1\bigl(S(h_1'')\,
                    h_2'''\bigr)
           \prod_{j=3}^{r-1}
           f_j\bigl(S(h_j'')\,
                    h_{j+1}'''\bigr)
           \sum_{(h)}
           f_2\bigl(S(h_2'')\,
                    h^{(2)}
                    h_3'''\bigr)\,
           f_r\bigl(S(h^{(1)}h_r'')\,
                    h_1'''\bigr) \\
        & \relphantom{=}
           \ket{h_1'}\otimes
           \ket{h_2'}\otimes
           \ket{h^{(3)}h_3'}\otimes
           \dots\otimes
           \ket{h^{(r)}h_r'}
           \displaybreak[0]\\
        & =\sum_{(h_i)}
           \prod_{j=1}^{r-2}
           f_j\bigl(S(h_j'')\,
                    h_{j+1}'''\bigr)
           \sum_{(h)}
           f_{r-1}\bigl(S(h_{r-1}'')\,
                        h^{(2)}h_r'''\bigr)\,
           f_r\bigl(S(h_1^{(1)}h_r'')\,
                    h_1'''\bigr) \\
        & \relphantom{=}
           \ket{h_1'}\otimes
           \dots\otimes
           \ket{h_{r-1}'}\otimes
           \ket{h^{(3)}h_r'} \\
        & =\ket{\psi_s(f_1,\dots,f_r)}
    \end{align*}
    where we repeatedly used Lemma~\ref{lem:Hopf_singlet1}.
    This concludes the proof.\qed
\end{proof}

In fact, the proof of Theorem~\ref{thm:Hopf_gs} implies
that $\ket{\psi_\Gamma}$ is also invariant under each local action
of $\mathrm{D}(H)$ at any site~$(s,p)$, not just under the operators
constituting the Hamiltonian. More precisely, one has
\begin{equation}
    B_f(s,p)\,
    A_a(s,p)
    \ket{\psi_\Gamma}
    =\epsilon(a)\,
     f(1_H)
     \ket{\psi_\Gamma}
\end{equation}
for any $f\otimes a\in\mathrm{D}(H)$. This can be easily seen from
the local $\mathrm{D}(H)$-module structure, the fact that $A(s)=A_h(s,p)$
and $B(p)=B_\phi(s,p)$ leave $\ket{\psi_\Gamma}$ strictly invariant and
the properties of the Haar integrals. In other words, \emph{the quantum
state~$\ket{\psi_\Gamma}$ is nothing but a trivial representation of the
quantum double~$\mathrm{D}(H)$}. Comparing with the comment after
Theorem~\ref{thm:model} one realizes that $\ket{\psi_\Gamma}$ should be
viewed as the spatially distributed version of the integral $\phi\otimes h
\in\mathrm{D}(H)$. For these reasons one may call $\ket{\psi_\Gamma}$
the vacuum of the model and as such it has trivial topological charge
everywhere.

Let us emphasize again that one really needs just a single datum,
i.e. the finite-dimensional Hopf $C^*$-algebra~$H$, to produce this topological ground
state since the Haar integral~$h$ (and the Haar functional~$\phi$) is uniquely defined. In
particular, the construction is fully basis-independent.

Furthermore the construction is symmetric%
\footnote{Strictly speaking, symmmetry holds up to a flip in the
    comultiplication, which is precisely the difference between
    $H^*$ and $X$.}
in the
    algebras~$H$ and $X$ with their respective integrals~$h$ and $\phi$,
    hence it has a natural dual notion. In fact, this foreshadows
    electric-magnetic duality as shown
    in~\cite{Buerschaper:2010p3003}.

While we excluded surface boundaries explicitly for
Theorem~\ref{thm:Hopf_gs} the following example shows what one
can learn from the presence of boundaries imposed by the underlying
surface. 

\begin{example}
    Consider the graph~$\Gamma$ underlying the diagram shown in
    Figure~\ref{fig:ttr} and let $H=\mathbb{CZ}_2$. Then the state
    \begin{equation}
        \ket{\psi_0}
        :=\ket{\psi_\Gamma(\{h\};
          \{\phi\};
          \{h\};
          \{\phi\})}
    \end{equation}
    is a codeword of the surface code defined in~\cite{Bravyi:1998}.
    It encodes the logical state~$\ket{+}$. The other
    codeword~$\ket{\psi_1}$ can be obtained by acting on this Hopf
    tensor network state with the appropriate string operator connecting
    the two rough boundaries. This
    generalizes to any finite-dimensional Hopf $C^*$-algebra by using the
    appropriate ribbon operators.
\end{example}

\begin{remark}
    From the proof of Theorem~\ref{thm:Hopf_gs} it is clear that
    any Hopf tensor network state of the form
    \begin{equation}
        \ket{\psi_\Gamma(h,\dots,h;f_1,\dots,f_{\abs{F}})}
    \end{equation}
    is invariant under all vertex operators~$A(s)$. From the perspective
    of lattice gauge theory this means that deforming the ground
    state~$\ket{\psi_\Gamma}$ of the $\mathrm{D}(H)$-model by changing
    the functions $\{\phi\}\mapsto\{f_i\}$ only will never break the gauge
    symmetry. These deformations in the state might therefore well correspond
    to a local perturbation of the Hamiltonian and thus preserve
    topological order provided the strength of the perturbation
    is limited to a certain finite
    threshold~\cite{Bravyi:2010p2791,Bravyi:2010p2836}.
    Some initial work towards this direction has been conducted
    in~\cite{Chen:2010p2865} which is concerned with tensor network
    deformations of the toric code.
\end{remark}

\subsection{Hierarchy}

Apart from ground states of the $\mathrm{D}(H)$-models the
framework of Hopf tensor network states based on an arbitrary Hopf
$C^*$-algebra~$H$ comprises more intriguing examples of quantum many-body
states. These can be obtained from certain other choices of elements in
$H$ and $X$ which are motivated both by the algebraic structure itself
as well as by ideas of charge condensation~\cite{Bais:2002p1647,Bais:2003p1648}:

\begin{definition}[Hierarchy]
    \label{def:Hopf_hierarchy_states}
    Let $A\subset H$ and $B\subset X$ Hopf subalgebras with Haar
    integrals $h_A\in A$ and $\phi_B\in B$ respectively. Set
    \begin{equation}
        \ket{\psi_\Gamma^{A,B}}
        :=\ket{\psi_\Gamma(h_A,\dots,h_A;
          \phi_B,\dots,\phi_B)}.
    \end{equation}
\end{definition}

Obviously, with the choice $A=H$ and $B=X$ we recover the
state~$\ket{\psi_\Gamma}$ which is topologically ordered as a ground
state of the $\mathrm{D}(H)$-model. On the other hand, if $B=\{1_X\}$ is
the trivial Hopf subalgebra of $X$ then the resulting Hopf tensor
network state is a product state for any Hopf subalgebra $A\subset H$.
This can be seen from~\eqref{eq:coproduct_unit}. Indeed, suppose a face~$p$
has $r$~edges in its boundary and $f_p=1_X$. Then
iterating~\eqref{eq:coproduct_unit} yields
\begin{equation}
    \sum_{(f_p)}
    f_p^{(1)}\otimes
    \dots\otimes
    f_p^{(r)}
    =1_X\otimes
     \dots\otimes
     1_X
\end{equation}
and consequently one has
\begin{equation}
   \ttr_\Gamma(\{x_e\};1_X,\dots,1_X)
   =\!\prod_{e\in E_M}
    \vcenter{\hbox{\includegraphics{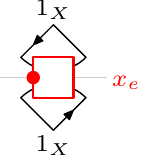}}}
\end{equation}
up to a possible application of the antipode~$S$ on each graph edge
depending on its orientation. Since $1_X=\epsilon^T(1_\mathbb{C})$ is the
analogue of~\eqref{eq:unit} for the Hopf $C^*$-algebra~$X$ we deduce
\begin{equation}
    \vcenter{\hbox{\includegraphics{lattices-97}}}
    =\sum_{(x_e)}
     \eval{\epsilon^T(1_\mathbb{C})}
          {S(x_e')}\,
     \eval{\epsilon^T(1_\mathbb{C})}
          {x_e''}
    =\sum_{(x_e)}
     \epsilon\bigl(S(x_e')\bigr)\,
     \epsilon(x_e'')
    =\epsilon(x_e)
\end{equation}
and therefore
\begin{equation}
    \ttr_\Gamma(\{x_e\};1_X,\dots,1_X)
    =\!\prod_{e\in E_M}\!
     \epsilon(x_e)
\end{equation}
which no longer depends on the orientation of $\Gamma$. Now by
virtue of~\eqref{eq:Hopf_TN_state_compact} the quantum state
\begin{equation}
    \ket{\psi_\Gamma(h_A,\dots,h_A;1_X,\dots,1_X)}
    =\!\sum_{(h_{A,e})}
     \bigotimes_{e\in E_M}\!
     \epsilon(h_{A,e}'')
     \ket{h_{A,e}'}
    =\!\bigotimes_{e\in E_M}\!
     \ket{h_A}
\end{equation}
is seen to factor into a simple product state.

In between these two extremes a hierarchy of quantum states unfolds
which are indexed by different choices of $A$ and $B$. However,
depending on the Hopf algebra~$H$ in question the interior of this
hierarchy may collapse partially. This means that different
pairs~$(A,B)$ of Hopf subalgebras may actually define identical quantum
states. Unfortunately we do not know how to characterize the surviving
equivalence classes of states in closed form without additional
assumptions on $H$.

However, if $H=\mathbb{C}G$ we have a clear picture of the above
hierarchy.
It turns out that the classes of Hopf tensor network states
emerging from the partial collapse are isomorphic to ground states of
certain quantum double models based on groups smaller than $G$. Indeed, the relevant Hopf subalgebras in
this case exactly read $A=\mathbb{C}K$ and $B=\mathbb{C}^{G/N}$ where
$K\subset G$ is a subgroup and $N\lhd G$ a normal subgroup (see
Section~\ref{sec:subalgebras_group}). For simplicity we abbreviate such
a pair of Hopf subalgebras by $(K,N)$. Then it is not difficult to see
that both $(K,N)$ and $(K,K\cap N)$ yield identical Hopf tensor network
states. Furthermore, if $\ket{k_1,\dots,k_{\abs{E}}}$ with $k_e\in K$
for each edge $e\in E$ is a basis state then one has
\begin{equation*}
    \braket{k_1,
            \dots,
            k_{\abs{E}}}
           {\psi_\Gamma^{K,K\cap N}}
    =\braket{k_1
             l_1,
             \dots,
             k_{\abs{E}}
             l_{\abs{E}}}
            {\psi_\Gamma^{K,K\cap N}}
\end{equation*}
for arbitrary elements $l_e\in K\cap N$, so all amplitudes are actually
constant on the cosets~$k_e(K\cap N)$. This means that one may apply the
canonical projection $\pi\colon K\to K/(K\cap N)$ at each edge and
regard $\phi_{K\cap N}$ as the Haar integral of
$\mathbb{C}^{K/(K\cap N)}$. Hence the resulting state coincides with the
ground state~$\ket{\psi_\Gamma}$ of the quantum double based on the
group algebra of the group $K/(K\cap N)\simeq KN/N$.

We conclude that the hierarchy of Hopf tensor network states arranges
quantum double models based on \emph{different} groups in one coherent
picture. These groups are precisely isomorphic to $KN/N$ so the
equivalence classes of states in the hierarchy are the trivial representations
of the quantum doubles~$\mathrm{D}\bigl(\mathbb{C}(KN/N)\bigr)$. These
comprise all possibilities for the residual symmetry algebra after the
full $\mathrm{D}(\mathbb{C}G)$-symmetry has been partially broken down
to a smaller symmetry algebra~\cite{Bais:2002p1647,Bais:2003p1648}. For that
reason the quantum states in our hierarchy realize the condensation of
topological charges within the \emph{same} underlying Hilbert space.

\subsection{{\sc PEPS}}

At this stage it is important to make contact with one of the usual
formulations of tensor network states. In the {\sc PEPS}
approach~\cite{Verstraete:2004p114} a quantum
state is represented by choosing some fixed basis and encoding the
wavefunction amplitude for each basis element in a fully contracted
tensor network. Contrastingly, our Hopf tensor network states are defined
without reference to any basis. It is rather the choice of particular,
often canonical, elements $\{x_e\}\subset H$ and $\{f_p\}\subset X$ which
determines the properties of the quantum state. If the need arises one
can still obtain an explicit wavefunction expansion in a straightforward manner.
As the next example shows, those distinguished elements naturally
generate all relevant sums and amplitudes.

\begin{figure}
    \centering
    \includegraphics{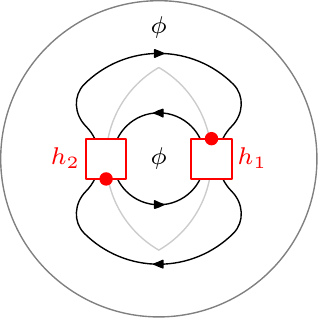}
    \caption{Hopf tensor network state $\ket{\psi_\Gamma}$ on a small
        graph~$\Gamma$ embedded in $S^2$. The outer circle is identified
        with the north pole of $S^2$.}
    \label{fig:small_graph}
\end{figure}

\begin{example}
    Let $H=\mathbb{C}^G$. We would like to explicitly construct the
    Hopf tensor network state~$\ket{\psi_\Gamma}$ for the graph~$\Gamma$
    shown in Figure~\ref{fig:small_graph} and obtain an expansion in
    terms of the canonical basis $\{\delta_g\mid g\in G\}$ of $\mathbb{C}^G$.
    We assume the underlying surface to be homeomorphic to $S^2$ as indicated
    in Figure~\ref{fig:small_graph}. Using the Haar integral $h_i=\delta_e$
    of the dual group algebra~$\mathbb{C}^G$ and its corresponding Haar functional
    $\phi\in H$ we get by
    Theorem~\ref{thm:Hopf_gs}
    \begin{equation*}
        \begin{split}
            \ket{\psi_\Gamma}
            & =\sum_{(h_i)}
               \phi(h_1'''h_2''')\,
               \phi\bigl(S(h_1'')\,
                         S(h_2'')\bigr)
               \ket{h_1'}\otimes
               \ket{h_2'} \\
            & =\sum_{(h_i)}
               \eval{\phi\otimes
                     \phi}
                    {\Delta(h_1''h_2'')}
               \ket{h_1'}\otimes
               \ket{h_2'} \\
            & =\sum_{(h_i)}
               \phi(h_1''h_2'')
               \ket{h_1'}\otimes
               \ket{h_2'} \\
            & =\!\sum_{u,v\in G}
               \phi(\delta_{u^{-1}}
                    \delta_{v^{-1}})
               \ket{\delta_u}\otimes
               \ket{\delta_v} \\
            & =\frac{1}{\abs{G}}
               \sum_{g\in G}
               \ket{\delta_g}\otimes
               \ket{\delta_g}.
        \end{split}
    \end{equation*}
    Interestingly, while the Haar integral~$\delta_e$ merely represents
    a single element of the basis, it is actually the act of
    comultiplication which produces the sum over the entire basis
    from $\delta_e$.\qed
\end{example}

As outlined in the
introduction {\sc PEPS} are defined in terms of local tensors with
\emph{open} virtual indices, for a particular example
see~\eqref{eq:group_tensors}. Contrastingly, our diagrammatic notation as
given by~\eqref{eq:rule4} uses the comultiplication of the Hopf algebra
to disconnect and separate local objects on the virtual level, hence it
appears there are \emph{no} virtual indices at all in our formalism.

This is not entirely true. If the functions~$\{f_p\}$ for evaluating
virtual loops belong to a particular class%
\footnote{This is the so-called character ring
    $R_\mathbb{C}(H)=\sum_{i=1}^n\mathbb{C}\chi_i$ where the~$\chi_i$
    are the irreducible characters of $H$.},
then indeed one may regard the evaluation
of a virtual loop as tracing over a product of certain matrices, once a
particular basis of $H$ has been chosen. Consequently, each such matrix
naturally constitutes part of a local tensor with open virtual indices
and evaluating the loop corresponds to a cyclic contraction of those
indices. Note that this applies to all Hopf tensor network states of our
hierarchy (but is not limited to these). For example, for the
state~$\ket{\psi_\Gamma}$ at the top of the hierarchy one has the local
projector
\begin{equation}
    P
    =\norm{h}^{-1}
     \sum_{(h)}
     \sum_{\alpha,\beta,\gamma,\delta\in\mathcal{B}}
     \bigl(L_+^{S(h'')}\bigr)_{\alpha\beta}\,
     \bigl(L_+^{h'''}\bigr)_{\gamma\delta}
     \ket{h'}
     \bra{\alpha,\beta,\gamma,\delta}
\end{equation}
where $(L_+^a)_{\alpha\beta}$ denotes a matrix element of the
action~\eqref{eq:action1} with respect to some basis~$\mathcal{B}$.%
\footnote{Furthermore one may express the (physical) ket in the same
    basis, too, and appreciate the similarities and differences as
    compared to the trivial case~\eqref{eq:projector}.}
Note that in this setting the local matrices encode both information
about the spin state on the graph edge as well as about the function
used to evaluate the virtual loop.

Furthermore, if a Hopf tensor trace admits a representation in terms
of local tensors with open indices properties like \emph{injectivity}
may be studied. More precisely, a {\sc PEPS} is called
\emph{injective}~\cite{PerezGarcia:2008p401} if there is a partition of the
underlying graph in disjoint regions~$R_i$ such that for every region~$R_i$
the linear map from open virtual indices at the boundary to physical indices
in the interior is injective. Whether or not a particular tensor network
representation of a quantum state is injective has important consequences for e.g.
the existence of a parent Hamiltonian that has the given {\sc PEPS}
as its unique ground state~\cite{PerezGarcia:2008p401}. One can show
that all Hopf tensor network states~$\ket{\psi_\Gamma^{A,B}}$ in the hierarchy
are \emph{not} injective for any finite-dimensional Hopf $C^*$-algebra~$H$.
However, (at least) the state~$\ket{\psi_\Gamma}$ at the top of the hierarchy
obeys a relaxed version of injectivity which one might call \emph{$H$-injectivity}.
This is a certain generalization of the $G$-injectivity condition
defined in~\cite{Schuch:2010p2806}.

\section{Calculating the topological entanglement entropy}
\label{sec:entropy}

Topological order is commonly associated with \emph{non-local}
order parameters. It is believed that among these resides the topological
entanglement entropy~$\gamma$~\cite{Kitaev:2006,Levin:2006p288} which is a
universal additive correction to the area law for the entanglement entropy of a bipartition of the system into a region~$R$ and its complement. Given a ground state
of a topologically ordered system, the entanglement entropy~$S_R$ for a
region~$R$ is argued to scale as
\begin{equation}
    S_R
    =\alpha\cdot
     \abs{\partial R}
     -\gamma
\end{equation}
in the limit of infinitely large regions. Here $\alpha$ is a parameter
which encodes non-universal behaviour on short length scales. In fact, the
topological entanglement entropy~$\gamma$ also contains partial information
about the types of quasiparticle excitations that may occur in the low-energy
sector of the system.

For these reasons we would like to show in this section how the non-local order
parameter~$\gamma$ can be understood naturally in the context of Hopf tensor
network states. In particular, we are going to show that different classes
of finite-dimensional Hopf $C^*$-algebras yield fundamentally different
mechanisms for the emergence of a non-vanishing topological entanglement
entropy.

To this end, we will exploit the generic decomposition of Hopf tensor
network states into interior and exterior parts as stated in
Proposition~\ref{prop:partitions}. We will then compute the block entropy
directly from certain properties of the reduced density operator. In order
to render its simple structure evident we will make use of isometries and
completely clear out the interior of both the region~$R$ and its
complement~$\bar{R}$. Effectively, we concentrate all topological information
contained in the ground state into the (inner) boundaries of the system.
Furthermore, the distillation process will be crafted such that all intermediate
quantum states can be kept track of conveniently via Hopf tensor traces.

It should be noted that our scheme of applying isometries implements entanglement renormalisation~\cite{Vidal:2007p2581,Vidal:2008p2582} both for the $\mathrm{D}(H)$-models as well as for states of the hierarchy. In fact, it directly extends the work in~\cite{Aguado:2008p518}
which is concerned with states at the top of the hierarchy for $H=\mathbb{C}G$.
At the same time our scheme provides a complementary view on entanglement renormalisation for
string-net models~\cite{Konig:2009p1839}, namely from the perspective of a local symmetry algebra.

\subsection{Isometries}
\label{sec:isometries}

We begin by developing the distillation process. For that we will
introduce a hierarchy of “little” isometries
that insert minimal faces and vertices into a given graph. The little
isometries at the top end of the hierarchy are taylored to the generalized quantum
double model, i.e. they preserve all excitations of that model.
Subsequently we will define a pair of local unitary maps that
allow for reconnecting edges. Taken together, both the little isometries
and the unitary maps will
yield a hierarchy of isometries that add or remove arbitrary edges.
Their effect on Hopf tensor network states will be analyzed in the next section.

As far as notation is concerned we will always denote the original graph
by $\Gamma_1$ and the modified one by $\Gamma_2$. Whenever it is
appropriate to talk about actual quantum double models, $\mathcal{H}_i$ will
denote the Hamiltonian of the $\mathrm{D}(H)$-model on the
graph~$\Gamma_i$.

\begin{proposition}[Little isometries]
    \label{prop:little_isometries}
    Let $A\subset H$ and $B\subset X$ Hopf subalgebras with the
    respective Haar integrals $h_A\in A$ and $\phi_B\in B$.
    Additionally, let
    \begin{align}
        \lambda_{A,B} & :=\sum_{(h_A)}
                          \phi_B(h_A'')\,
                          h_A', \\
        1_{A,B}       & :=\frac{\lambda_{A,B}}
                          {\norm{\lambda_{A,B}}}, \\
        \Lambda_A     & :=\frac{h_A}{\norm{h_A}}
                         =\frac{1}{\sqrt{\phi(h_A)}}\,
                          h_A.
    \end{align}

    Then the linear maps~$i_F^{A,B}$ and $i_V^A$ defined via
    \begin{align}
        \vcenter{\hbox{\includegraphics{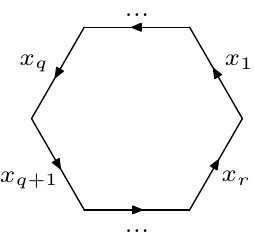}}}\,
        & \xrightarrow{i_F^{A,B}}
           \vcenter{\hbox{\includegraphics{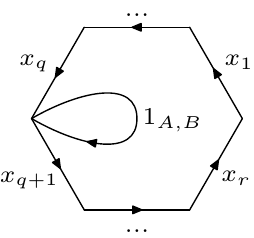}}} \\
        \intertext{and}
        \vcenter{\hbox{\includegraphics{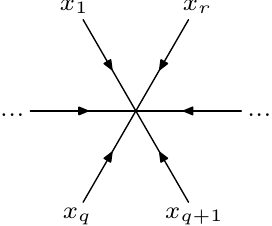}}}\,
        & \xrightarrow{i_V^A}\,
           \vcenter{\hbox{\includegraphics{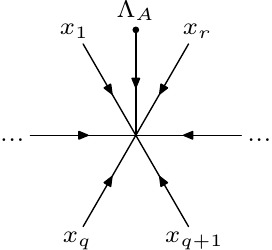}}}
    \end{align}
    are isometries.

    Furthermore, for $i_F:=i_F^{H,H^*}$ and $i_V:=i_V^H$ one has
    \begin{align}
        i_F
        \mathcal{H}_1
        & =\mathcal{H}_2
           i_F, \\
        i_V
        \mathcal{H}_1
        & =\mathcal{H}_2
           i_V.
    \end{align}
\end{proposition}
\begin{proof}
    For the first claim, it is enough to show the invariance of the
    inner product or equivalently that $\norm{1_{A,B}}=\norm{\Lambda_A}
    =1$. Indeed, from
    \begin{equation*}
        \norm{h_A}^2
        =(h_A,h_A)
        =\phi(h_A^*h_A)
        =\phi(h_A^2)
        =\phi(h_A).
    \end{equation*}
    this is easily seen to be true.

    As for the second claim we note that $i_F$ inserts a loop which is
    automatically stabilized by the corresponding face operator in
    $\mathcal{H}_2$. It is also easily seen that $1_{H,H^*}=1_H$ and hence it does
    not affect the magnetic flux through the original face as measured
    by either $\mathcal{H}_1$ or $\mathcal{H}_2$. Since $1_H$ is trivially invariant under
    the adjoint action we see that the vertex operator at the leftmost
    vertex is not affected by the additional face, either. By similar
    arguments one proves that $i_V$ intertwines the Hamiltonians,
    too.\qed
\end{proof}

For the sake of completeness we also give the maps~$(i_F^{A,B})^\dagger$
and $(i_V^A)^\dagger$:
\begin{align}
    \vcenter{\hbox{\includegraphics{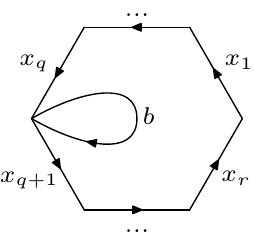}}}\,
    & \xrightarrow{(i_F^{A,B})^\dagger}
       (1_{A,B},b)
       \vcenter{\hbox{\includegraphics{lattices-65}}} \\
    \vcenter{\hbox{\includegraphics{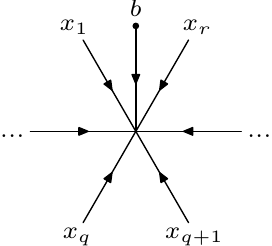}}}
    & \xrightarrow{(i_V^A)^\dagger}
       (\Lambda_A,b)\,\,
       \vcenter{\hbox{\includegraphics{lattices-67}}}
       \label{eq:blub}
\end{align}

Now we define the unitary maps that allow for reconnecting edges of the
underlying graph.
\begin{proposition}[Unitaries]
    \label{lem:unitaries}
    The linear maps~$U_F$ and $U_V$ defined via
    \begin{align}
        \vcenter{\hbox{\includegraphics{lattices-61}}}
        & \xrightarrow{U_F}
           \sum_{(x_i)}\,\,
           \vcenter{\hbox{\includegraphics{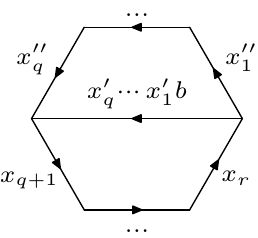}}}
           \label{eq:unitary_face}\\
        \intertext{and}
        \vcenter{\hbox{\includegraphics{lattices-57}}}\,
        & \xrightarrow{U_V}
           \sum_{(b)}\,
           \vcenter{\hbox{\includegraphics{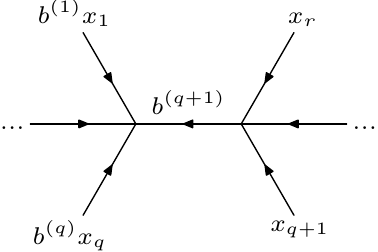}}}
           \label{eq:unitary_vertex} \\
        \intertext{are unitary and their inverses read}
        \vcenter{\hbox{\includegraphics{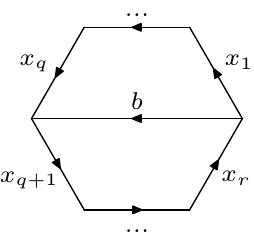}}}
        & \xrightarrow{U_F^\dagger}
           \sum_{(x_i)}\,\,
           \vcenter{\hbox{\includegraphics{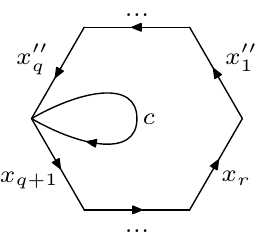}}} \\
        \intertext{where $c=S(x_q'\cdots x_1')\,b$ and}
        \vcenter{\hbox{\includegraphics{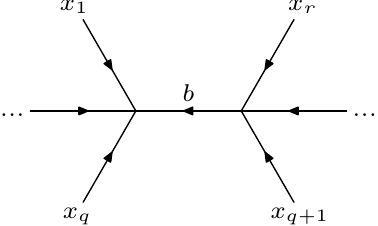}}}\,
        & \xrightarrow{U_V^\dagger}
           \sum_{(b)}
           \vcenter{\hbox{\includegraphics{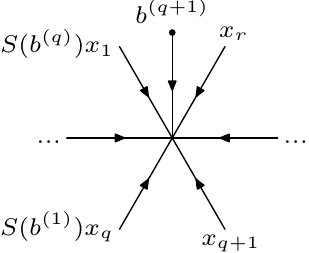}}}
    \end{align}
    respectively.

    Furthermore one has
    \begin{align}
        U_F
        \mathcal{H}_1
        & =\mathcal{H}_2
           U_F,
           \label{eq:unitary_face_commute}\\
        U_V
        \mathcal{H}_1
        & =\mathcal{H}_2
           U_V.
           \label{eq:unitary_vertex_commute}
    \end{align}
\end{proposition}
\begin{proof}
    It is easy to check that $U_F^\dagger$ ($U_V^\dagger$) as defined
    above is both a left and a right inverse of the map~$U_F$ ($U_V$).
    Indeed, leaving out the part $x_{q+1}\otimes\dots\otimes x_r$ (which
    is unaffected by $U_F$) one has:
    \begin{align*}
        U_F^\dagger
        U_F
        (x_1\otimes
         \dots\otimes
         x_q\otimes
         b)
        & =U_F^\dagger
           \biggl(\sum_{(x_i)}
                  x_1''\otimes
                  \dots\otimes
                  x_q''\otimes
                  x_q'\cdots
                  x_1'
                  b\biggr) \\
        & =\sum_{(x_i)}
        x_1'''\otimes
        \dots\otimes
        x_q'''\otimes
        S(x_q''\cdots
        x_1'')\,
        x_q'\cdots
        x_1'
        b
        \displaybreak[0]\\
        & =\sum_{(x_i)}
        x_1'''\otimes
        \dots\otimes
        x_q'''\otimes
        S(x_1'')\cdots
        S(x_q'')\,
        x_q'\cdots
        x_1'
        b
        \displaybreak[0]\\
        & =\sum_{(x_i)}
        x_1''\otimes
        \dots\otimes
        x_q''\otimes
        \epsilon(x_1')\cdots
        \epsilon(x_q')\,
        b \\
        & =x_1\otimes
        \dots\otimes
        x_q\otimes
        b.
    \end{align*}
    Similarly, one shows $U_FU_F^\dagger=\id$. Now
    \begin{align*}
        U_V^\dagger
        U_V(x_1\otimes
        \dots\otimes
        x_q\otimes
        b)
        & =U_V^\dagger
        \sum_{(b)}
        b^{(1)}
        x_1\otimes
        \dots\otimes
        b^{(q)}
        x_q\otimes
        b^{(q+1)} \\
        & =\sum_{(b)}
        \sum_{(b^{(q+1)})}
        S\bigl((b^{(q+1)})^{(q)}\bigr)\,
        b^{(1)}
        x_1\otimes
        \cdots \\
        & \relphantom{=}
        \hphantom{\sum_{(b)}
          \sum_{(b^{(q+1)})}}
        \otimes
        S\bigl((b^{(q+1)})^{(1)}\bigr)\,
        b^{(q)}
        x_q\otimes
        (b^{(q+1)})^{(q+1)}
        \displaybreak[0]\\
        & =\sum_{(b)}
        S(b^{(2q)})\,
        b^{(1)}
        x_1\otimes
        \dots\otimes
        S(b^{(q+1)})\,
        b^{(q)}
        x_q\otimes
        b^{(2q+1)}
        \displaybreak[0]\\
        & =\sum_{(b)}
        S(b^{(2q-2)})\,
        b^{(1)}
        x_1\otimes
        \dots\otimes
        S(b^{(q)})\,
        b^{(q-1)}
        x_{q-1} \\
        & \relphantom{=}
        \hphantom{\sum_{(b)}}
        \otimes
        x_q\otimes
        b^{(2q-1)}
        \displaybreak[0]\\
        & =\sum_{(b)}
        S(b'')\,
        b'
        x_1\otimes
        x_2\otimes
        \dots\otimes
        x_q\otimes
        b''' \\
        & =x_1\otimes
        \dots\otimes
        x_q\otimes
        b.
    \end{align*}
    and by the same token one proves $U_VU_V^\dagger=\id$.

    For the first claim it remains to show that for all $x_i,y_i,b,c\in
    H$ the inner product~\eqref{eq:inner_product} is invariant under
    $U_F$ ($U_V$):
    \begin{align*}
        & \bigl(U_F
        (x_1\otimes
        \dots\otimes
        x_q\otimes
        b),
        U_F
        (y_1\otimes
        \dots\otimes
        y_q\otimes
        c)\bigr) \\
        & =\biggl(\sum_{(x_i)}
                  x_1''\otimes
                  \dots\otimes
                  x_q''\otimes
                  x_q'\cdots
                  x_1'
                  b,
                  \sum_{(y_i)}
                  y_1''\otimes
                  \dots\otimes
                  y_q''\otimes
                  y_q'\cdots
                  y_1'
                  c\biggr)
        \displaybreak[0]\\
        & =\!\!\sum_{(x_i)(y_i)}\!
        (x_q'\cdots
        x_1'
        b,
        y_q'\cdots
        y_1'
        c)
        \prod_{j=1}^q
        (x_j'',y_j'')
        \displaybreak[0]\\
        & =\!\!\sum_{(x_i)(y_i)}\!
        \phi\bigl(b^*(x_1^*)'\cdots
        (x_q^*)'
        y_q'\cdots
        y_1'
        c\bigr)\,
        \phi\bigl((x_q^*)''
        y_q''\bigr)
        \prod_{j=1}^{q-1}
        \phi\bigl((x_j^*)''
        y_j''\bigr)
        \displaybreak[0]\\
        & =(x_q,y_q)\!
        \sum_{(x_i)(y_i)}\!
        \phi\bigl(b^*(x_1^*)'\cdots
        (x_{q-1}^*)'
        y_{q-1}'\cdots
        y_1'
        c\bigr)
        \prod_{j=1}^{q-1}
        \phi\bigl((x_j^*)''
        y_j''\bigr)
        \displaybreak[0]\\
        & =(b,c)
        \prod_{j=1}^q
        (x_j,y_j) \\
        & =(x_1\otimes
        \dots\otimes
        x_q\otimes
        b,
        y_1\otimes
        \dots\otimes
        y_q\otimes
        c).
    \end{align*}
    Note that we used property~\eqref{eq:dual_integral} of the Haar
    integral~$\phi$ in the fifth line. Furthermore we have
    \begin{align*}
        & \bigl(U_V
        (x_1\otimes
        \dots\otimes
        x_q\otimes
        b),
        U_V
        (y_1\otimes
        \dots\otimes
        y_q\otimes
        c)\bigr) \\
        & =\biggl(\sum_{(b)}
                  b^{(1)}
                  x_1\otimes
                  \dots\otimes
                  b^{(q)}
                  x_q\otimes
                  b^{(q+1)},
                  \sum_{(c)}
                  c^{(1)}
                  y_1\otimes
                  \dots\otimes
                  c^{(q)}y_q\otimes
                  c^{(q+1)}\Biggr)
           \displaybreak[0]\\
        & =\!\sum_{(b)(c)}
        (b^{(q+1)},
        c^{(q+1)})
        \prod_{j=1}^q
        (b^{(j)}x_j,
        c^{(j)}y_j)
        \displaybreak[0]\\
        & =\!\sum_{(b)(c)}
        \phi\bigl((b^{(q+1)})^*
        c^{(q+1)}\bigr)
        \prod_{j=1}^q
        \phi\bigl(x_j^*
        (b^{(j)})^*
        c^{(j)}
        y_j\bigr)
        \displaybreak[0]\\
        & =\sum_{(b^*c)}
        \phi\bigl((b^*c)^{(q+1)}\bigr)
        \prod_{j=1}^q
        \phi\bigl(x_j^*
        (b^*c)^{(j)}
        y_j\bigr)
        \displaybreak[0]\\
        & =(x_q,y_q)
        \sum_{(b^*c)}
        \phi\bigl((b^*c)^{(q)}\bigr)
        \prod_{j=1}^{q-1}
        \phi\bigl(x_j^*
        (b^*c)^{(j)}
        y_j\bigr)
        \displaybreak[0]\\
        & =(b,c)
        \prod_{j=1}^q
        (x_j,y_j) \\
        & =(x_1\otimes
        \dots\otimes
        x_q\otimes
        b,
        y_1\otimes
        \dots\otimes
        y_q\otimes
        c).
    \end{align*}

    For the second claim, one needs to show that $U_F$ ($U_V$) commutes
    appropriately with the terms~$A(s)$ and $B(p)$ in the
    Hamiltonians~$\mathcal{H}_i$. This can be done directly for those faces which
    are complete in~\eqref{eq:unitary_face}. Stating this more
    precisely we have:
    \begin{equation*}
        U_F
        B(p_i)
        =B(p_i')\,
        U_F\qquad
        i\in\{1,2\}
    \end{equation*}
    where faces are labelled as follows:
    \begin{equation*}
        \vcenter{\hbox{\includegraphics{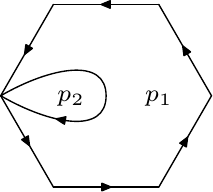}}}\ ,\qquad
        \vcenter{\hbox{\includegraphics{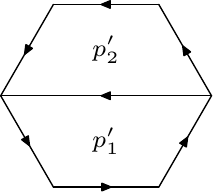}}}\ .
    \end{equation*}
    Indeed, consider
    \begin{align*}
        & U_F
        B(p_2)
        (x_1\otimes
        \dots\otimes
        x_q\otimes
        b) \\
        & =U_F
           \biggl(x_1\otimes
                  \dots\otimes
                  x_q\otimes
                  \sum_{(b)}
                  \phi(b')\,
                  b''\biggr)
           \displaybreak[0]\\
        & =\!\!\sum_{(x_i)(b)}\!
        x_1''\otimes
        \dots\otimes
        x_q''\otimes
        x_q'\cdots
        x_1'
        b'
        \phi(b'')
        \displaybreak[0]\\
        & =\!\!\sum_{(x_i)(b)}\!
        \phi\bigl(S(b'')\,
        S(x_1^{(2)})\cdots
        S(x_q^{(2)})\,
        x_q^{(3)}\cdots
        x_1^{(3)}\bigr)\,
        x_1^{(4)}\otimes
        \dots\otimes
        x_q^{(4)}\otimes
        x_q^{(1)}\cdots
        x_1^{(1)}
        b'
        \displaybreak[0]\\
        & =\!\!\sum_{(x_i)(b)}\!
        \phi\bigl(S(x_q^{(2)}\cdots
        x_1^{(2)}
        b'')\,
        x_q^{(3)}\cdots
        x_1^{(3)}\bigr)\,
        x_1^{(4)}\otimes
        \dots\otimes
        x_q^{(4)}\otimes
        x_q^{(1)}\cdots
        x_1^{(1)}
        b'
        \displaybreak[0]\\
        & =(\id\otimes
        \dots\otimes
        \id\otimes
        S)\!
        \sum_{(x_i)(b)}
        \sum_{(S(x_q'\cdots x_1'b))}\!
        \phi\bigl(S(x_q'\cdots
        x_1'
        b)'
        x_q''\cdots
        x_1'') \\
        & \relphantom{=}
        \hphantom{(\id\otimes
          \dots\otimes
          \id\otimes
          S)
          \!\sum_{(x_i)(b)}
          \sum_{(S(\dots))}}
        x_1'''\otimes
        \dots\otimes
        x_q'''\otimes
        S(x_q'\cdots
        x_1'
        b)''
        \displaybreak[0]\\
        & =B(p_2')
        (x_1''\otimes
        \dots\otimes
        x_q''\otimes
        x_q'\cdots
        x_1'
        b) \\
        & =B(p_2')\,
        U_F
        (x_1\otimes
        \dots\otimes
        x_q\otimes
        b).
    \end{align*}
    Note that in the fourth line we repeatedly used the definition of
    the antipode. Hence both $B(p_2)$ and $B(p_2')$ project $b$ onto the
    value~$\phi(b)\,1_H$. In the following we may therefore assume that
    $b$ is replaced by $\phi(b)\,1_H$. We now analyze the other complete
    face:
    \begin{align*}
        & U_F
        B(p_1)
        \bigl(x_1\otimes
        \dots\otimes
        x_q\otimes
        \phi(b)\,
        1_H\otimes
        x_{q+1}\otimes
        \dots\otimes
        x_r\bigr) \\
        & =\phi(b)\,
           U_F
           \biggl(\sum_{(x_i)}
           \phi(x_r'\cdots
                x_1')\,
           x_1''\otimes
           \dots\otimes
           x_q''\otimes
           1_H\otimes
           x_{q+1}''\otimes
           \dots\otimes
           x_r''\biggr)
        \displaybreak[0]\\
        & =\phi(b)
        \sum_{(x_i)}
        \phi(x_r'\cdots
        x_1')\,
        x_1'''\otimes
        \dots\otimes
        x_q'''\otimes
        x_q''\cdots
        x_1''\otimes
        x_{q+1}''\otimes
        \dots\otimes
        x_r''
        \displaybreak[0]\\
        & =\phi(b)
        \sum_{(x_i)}
        \sum_{(x_q'\cdots
          x_1')}
        \phi\bigl(x_r'\cdots
        x_{q+1}'
        (x_q'\cdots
        x_1')'\bigr)\,
        x_1''\otimes
        \dots\otimes
        x_q''\otimes
        (x_q'\cdots
        x_1')'' \\
        & \relphantom{=}
        \hphantom{\phi(b)
          \sum_{(x_i)}
          \sum_{(x_q'\cdots
            x_1')}}
        \otimes
        x_{q+1}''\otimes
        \dots\otimes
        x_r''
        \displaybreak[0]\\
        & =\phi(b)\,
           B(p_1')
           \biggl(\sum_{(x_i)}
                  x_1''\otimes
                  \dots\otimes
                  x_q''\otimes
                  x_q'\cdots
                  x_1'\otimes
                  x_{q+1}\otimes
                  \dots\otimes
                  x_r\biggr)
           \displaybreak[0]\\
        & =B(p_1)\,
        U_F(x_1\otimes
        \dots\otimes
        x_q\otimes
        \phi(b)\,
        1_H\otimes
        x_{q+1}\otimes
        \dots\otimes
        x_r).
    \end{align*}
    For those surrounding faces which are affected by $U_F$ one may
    easily prove
    \begin{equation*}
        [U_F,(T_+^f)_i]
        =0\qquad
        i\in\{1,\dots,q\}
    \end{equation*}
    for any $f\in X$. Here $(T_+^f)_i$ denotes the action on the i-th
    tensor factor.

    Those incomplete vertices in~\eqref{eq:unitary_face} which are
    affected by $U_F$ can be dealt with if the following is true for all
    $a\in H$, $i\in\{1,\dots,q-1\}$ and $b\propto1_H$:
    \begin{align*}
        U_F(L_-^a)_1
        & =\biggl(\sum_{(a)}
                  (L_-^{a'})_1\otimes
                  (L_-^{a''})_\lozenge\biggr)
           U_F \\
        \biggl[U_F,
               \sum_{(a)}
               (L_+^{a'})_i\otimes
               (L_-^{a''})_{i+1}\biggr]
        & =0 \\
        U_F
        \biggl(\sum_{(a)}
               (L_+^{a''})_q\otimes
               \ad(a')_\lozenge\biggr)
        & =\biggl(\sum_{(a)}
                  (L_+^{a''})_q\otimes
                  (L_+^{a'})_\lozenge\biggr)
           U_F.
    \end{align*}
    Here $\lozenge$ denotes the bubble edge that is reconnected by
    $U_F$. While it is easy to show the first two equations, the
    remaining one needs some more attention:
    \begin{align*}
        & U_F
          \biggl(\sum_{(a)}
                 (L_+^{a''})_q\otimes
                 \ad(a')_\lozenge\biggr)
          \bigl(x_1\otimes
                \dots\otimes
                x_q\otimes
                \phi(b)\,
                1_H\bigr) \\
        & =U_F
           \biggl(\sum_{(a)}
                  x_1\otimes
                  \dots\otimes
                  x_{q-1}\otimes
                  a''
                  x_q\otimes
                  \phi(b)
                  \ad(a')(1_H)\biggr)
           \displaybreak[0]\\
        & =\phi(b)\,
        U_F(x_1\otimes
        \dots\otimes
        x_{q-1}\otimes
        ax_q\otimes
        1_H)
        \displaybreak[0]\\
        & =\phi(b)\!
        \sum_{(x_i)(a)}\!
        x_1''\otimes
        \dots\otimes
        x_{q-1}''\otimes
        a''
        x_q''\otimes
        a'
        x_q'\dots
        x_1'
        \displaybreak[0]\\
        & =\biggl(\sum_{(a)}
                  (L_+^{a''})_q\otimes
                  (L_+^{a'})_\lozenge\biggr)
           U_F
           (x_1\otimes
            \dots\otimes
            x_q\otimes
            \phi(b)\,
            1_H)
    \end{align*}
    This completes the proof of \eqref{eq:unitary_face_commute}.

    The proof of \eqref{eq:unitary_vertex_commute} is analogous
    and left as an exercise for the reader.\qed
\end{proof}

\begin{corollary}[Isometries]
    The linear maps $I_F^{A,B}:=U_F\circ i_F^{A,B}$,
    \begin{equation}
        \vcenter{\hbox{\includegraphics{lattices-65}}}
        \xrightarrow{I_F^{A,B}}
        \sum_{(x_i)}\,\,
        \vcenter{\hbox{\includegraphics{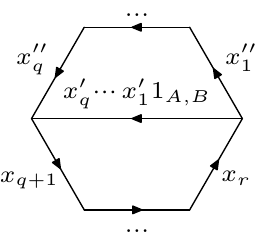}}}
    \end{equation}
    and $I_V^A:=U_V\circ i_V^A$,
    \begin{equation}
        \vcenter{\hbox{\includegraphics{lattices-67}}}\,
        \xrightarrow{I_V^A}
        \sum_{(\Lambda_A)}\,
        \vcenter{\hbox{\includegraphics{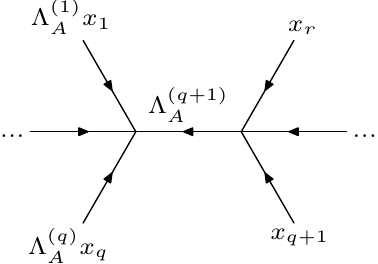}}}
    \end{equation}
    are isometries. Furthermore, for $I_F:=I_F^{H,H^*}$ and $I_V:=I_V^H$
    one has
    \begin{align}
        I_F
        \mathcal{H}_1
        & =\mathcal{H}_2
           I_F, \\
        I_V
        \mathcal{H}_1
        & =\mathcal{H}_2
           I_V.
    \end{align}
\end{corollary}

Put differently, the collection of these maps allows to move between the
$\mathrm{D}(H)$-models defined on two arbitrary graphs which are
embedded in the \emph{same} surface. Furthermore, all maps are local,
i.e. they cannot change nonlocal topological quantum numbers.

\subsection{Transforming hierarchy states}

In the preceding section we have learned that the set~$\{I_F,I_V\}$ of
isometries maps the ground state subspaces of generalized quantum double
models on related graphs onto each other. We would like to prove that this
not only holds for the entire subspace but also for individual ground states.
Namely, the isometries~$\{I_F,I_V\}$ precisely identify the Hopf tensor
network states~$\ket{\psi_{\Gamma_i}}$ as given by
Theorem~\ref{thm:Hopf_gs} for the different graphs~$\Gamma_i$.

Keeping our focus on states we will furthermore
show (for a special case) that the hierarchy of isometries~$\{I_F^{A,B},
I_V^A\}$ has the very same effect on the hierarchy of
states~$\ket{\psi_{\Gamma_i}^{A,B}}$ as given by
Definition~\ref{def:Hopf_hierarchy_states}. In other words, identifying
these states in an isometric fashion is as simple as adding or removing
edges from the underlying graph and adjusting the tensor trace
canonically.

Once we act with the isometries on our hierarchy of Hopf tensor network
states it will become clear that the quantum states defined in
Section~\ref{sec:Hopf_TN_states} are not properly normalized. We can
(partially) remedy this situation by choosing integrals of unit norm
instead of the Haar integrals~$h_A$ and $\phi_B$. In other words, we
need to include a factor of $\norm{h_A}^{-1}$ for each edge and a factor
of $\norm{\phi_B}^{-1}$ for each (complete) face of a Hopf tensor
network state. It is understood that from now on all Hopf tensor network
states will be normalized in this fashion without change of notation
unless otherwise stated.

We now fix some notation in the following definition so that we can
easily refer to pieces of Hopf tensor network states later on.

\begin{definition}
    Let $A\subset H$, $B\subset X$ Hopf subalgebras, $h_A\in A$,
    $\phi_B\in B$ their respective Haar integrals and let $f_1,\dots,
    f_r\in B$. Set
    \begin{align}
        \ket{\psi_p^{A,B}(\{f_i\})}
        & :=\frac{1}{\norm{h_A}^r\,
          \norm{\phi_B}}
        \sum_{(h_{A,i})}
        \vcenter{\hbox{\includegraphics{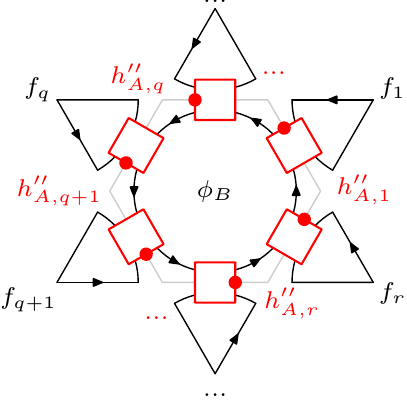}}}
        \nonumber\\
        & \relphantom{:=}
        \vcenter{\hbox{\includegraphics{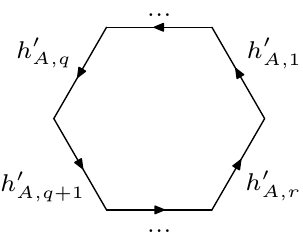}}},
        \displaybreak[0]\\
        \ket{\psi_{p\cup p'}^{A,B}(\{f_i\})}
        & :=\frac{1}{\norm{h_A}^{r+1}
          \norm{\phi_B}^2}
        \sum_{(h_{A,i})(h_A)}
        \vcenter{\hbox{\includegraphics{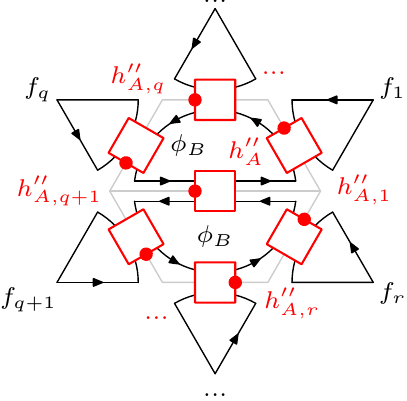}}}
        \nonumber\\
        & \relphantom{:=}
        \vcenter{\hbox{\includegraphics{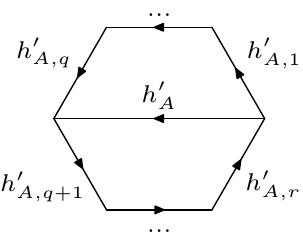}}},
        \displaybreak[0]\\
        \ket{\psi_s^{A,B}(\{f_i\})}
        & :=\frac{1}{\norm{h_A}^r}
        \sum_{(h_{A,i})}\,
        \vcenter{\hbox{\includegraphics{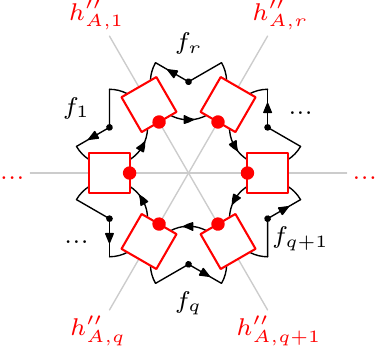}}}\quad
        \vcenter{\hbox{\includegraphics{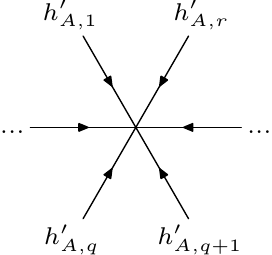}}}, \\
        \intertext{and}
        \ket{\psi_{s\cup s'}^{A,B}(\{f_i\})}
        & :=\frac{1}{\norm{h_A}^{r+1}}\!
        \sum_{(h_{A,i})(h_A)}
        \vcenter{\hbox{\includegraphics{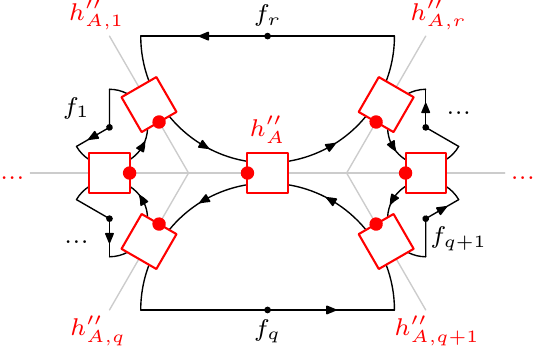}}}
        \nonumber\\
        & \relphantom{:=}
        \vcenter{\hbox{\includegraphics{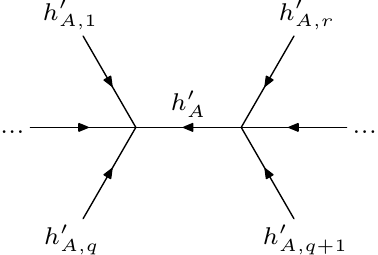}}}.
    \end{align}
\end{definition}

As before we simply write $\ket{\psi_\Gamma(f_1,\dots,f_r)}$ whenever
$A=H$ and $B=X$. Note that in this case one deals with a fragment of a
quantum double model ground state!

\begin{proposition}
    \label{prop:isometries_states}
    \begin{align}
        I_F
        \ket{\psi_p
          (f_1,\dots,f_r)}
        & =\ket{\psi_{p\cup p'}
          (f_1,\dots,f_r)}, \\
        I_V^A
        \ket{\psi_s^{A,B}
          (f_1,\dots,f_r)}
        & =\ket{\psi_{s\cup s'}^{A,B}
          (f_1,\dots,f_r)}.
    \end{align}
\end{proposition}
\begin{proof}
    First let us prove the relation involving the face isometry. It is
    enough to consider a face which is bounded by two edges. In this
    case we first transform the right hand side of the equation in order
    to eliminate one of the virtual loops:
    \begin{align*}
        \ket{\psi_{p\cup p'}
          (f_1,f_2)}
        & =\abs{H}^\frac{5}{2}\!\!
        \sum_{(h_i)(h)}\!\!
        \phi(h_1'''h''')\,
        \phi\bigl(h_2'''\,
        S(h'')\bigr)
        \prod_{j=1}^2
        f_j\bigl(S(h_j'')\bigr)
        \ket{h_1'}\otimes
        \ket{h'}\otimes
        \ket{h_2'} \\
        & =\abs{H}^\frac{3}{2}
        \sum_{(h_i)}
        \phi(h_2'''h_1^{(3)})\,
        f_1\bigl(S(h_1^{(2)})\bigr)\,
        f_2\bigl(S(h_2'')\bigr)
        \ket{h_1^{(1)}}\otimes
        \ket{S(h_1^{(4)})}\otimes
        \ket{h_2'} \\
        & =\abs{H}^\frac{3}{2}
        \sum_{(h_i)}
        \phi(h_2'''h_1^{(4)})\,
        f_1\bigl(S(h_1^{(3)})\bigr)\,
        f_2\bigl(S(h_2'')\bigr)
        \ket{h_1^{(2)}}\otimes
        \ket{S(h_1^{(1)})}\otimes
        \ket{h_2'}
    \end{align*}
    where in the second line we employed Lemma~\ref{lem:Hopf_singlet2}.

    In diagrammatic notation the above amounts to
    \begin{equation*}
        \ket{\psi_{p\cup p'}
          (f_1,f_2)}
        =\abs{H}^\frac{3}{2}
        \sum_{(h_i)}\,
        \vcenter{\hbox{\includegraphics{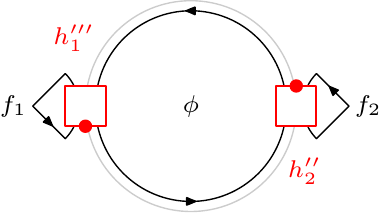}}}\quad
        \vcenter{\hbox{\includegraphics{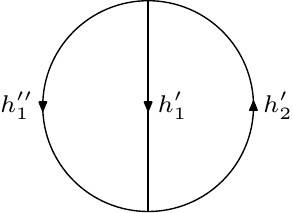}}}
    \end{equation*}
    which proves the claim about the face isometry.

    As for the vertex isometries we can in fact assume $A=H$ without
    loss of generality. By the same token we will never make use of $f_i
    \in B$ in the following argument, hence it suffices to assume $B
    =X$, too. We begin with the state on the original graph:
    \begin{equation*}
        \ket{\psi_s
          (f_1,\dots,f_r)}
        =\norm{h}^{-r}
        \sum_{(h_i)}
        \prod_{j=1}^r
        f_j\bigl(S(h_j'')\,
        h_{j+1}'''\bigr)
        \ket{h_1'}\otimes
        \dots
        \otimes
        \ket{h_r'}.
    \end{equation*}

    \bigskip

    Applying $I_V$ yields:
    \begin{align*}
        & \norm{h}^{r+1}\,
        I_V
        \ket{\psi_s
          (f_1,\dots,f_r)} \\
        & =\!\!\sum_{(h)(h_i)}
        \prod_{j=1}^r
        f_j\bigl(S(h_j'')\,
        h_{j+1}'''\bigr)
        \bigotimes_{k=1}^q
        \ket{h^{(k)}h_k'}\otimes
        \ket{h^{(q+1)}}\otimes
        \ket{h_{q+1}'}\otimes
        \dots\otimes
        \ket{h_r'}
        \displaybreak[0]\\
        & =\!\!\sum_{(h_i)(h)}\!
        f_r\bigl(S(h_r'')\,
        h_1'''\bigr)
        \prod_{j=1}^q
        f_j\bigl(S(h_j'')\,
        h_{j+1}'''\bigr)
        \bigotimes_{k=1}^q
        \ket{h^{(k)}h_k'}\otimes
        \ket{h^{(q+1)}}\otimes\cdots
    \end{align*}
    where in the last line we have abbreviated the remaining tensor
    factor
    \begin{equation*}
        \prod_{j=q+1}^{r-1}\!
        f_j\bigl(S(h_j'')\,
        h_{j+1}'''\bigr)
        \bigotimes_{k=q+1}^r\!
        \ket{h_k'}
    \end{equation*}
    by the trailing ellipsis. In order to simplify notation we will
    continue to do so in the following. Hence:
    \begin{align*}
        & \norm{h}^{r+1}\,
        I_V
        \ket{\psi_s
          (f_1,\dots,f_r)} \\
        & =\!\!\sum_{(h_i)(h)}\!
        f_r\bigl(S(h_r'')\,
        h_1'''\bigr)\,
        f_1\bigl(S(h_1'')\,
        h_2'''\bigr)
        \prod_{j=2}^q
        f_j\bigl(S(h_j'')\,
        h_{j+1}'''\bigr) \\
        & \relphantom{=}
        \ket{h^{(1)}h_1'}\otimes
        \dots\otimes
        \ket{h^{(q)}h_q'}\otimes
        \ket{h^{(q+1)}}\otimes\cdots
        \displaybreak[0]\\
        & =\!\!\sum_{(h_i)(h)}\!
        f_r\bigl(S(h_r'')\,
        S(h^{(1)})\,
        h_1'''\bigr)\,
        f_1\bigl(S(h_1'')\,
        h^{(2)}
        h_2'''\bigr)\,
        f_2\bigl(S(h_2'')\,
        h_3'''\bigr)
        \prod_{j=3}^q
        f_j\bigl(S(h_j'')\,
        h_{j+1}'''\bigr) \\
        & \relphantom{=}
        \ket{h_1'}\otimes
        \ket{h^{(3)}h_2'}\otimes
        \dots\otimes
        \ket{h^{(q+1)}h_q'}\otimes
        \ket{h^{(q+2)}}\otimes\cdots
        \displaybreak[0]\\
        & =\!\!\sum_{(h_i)(h)}\!
        f_r\bigl(S(h_r'')\,
        S(h^{(1)})\,
        h_1'''\bigr)\,
        f_1\bigl(S(h_1'')\,
        h_2'''\bigr)\,
        f_2\bigl(S(h_2'')\,
        h^{(2)}
        h_3'''\bigr)
        \prod_{j=3}^q
        f_j\bigl(S(h_j'')\,
        h_{j+1}'''\bigr) \\
        & \relphantom{=}
        \ket{h_1'}\otimes
        \ket{h_2'}\otimes
        \ket{h^{(3)}h_3'}\otimes
        \dots\otimes
        \ket{h^{(q)}h_q'}\otimes
        \ket{h^{(q+1)}}\otimes\cdots
        \displaybreak[0]\\
        & =\!\!\sum_{(h_i)(h)}\!
        f_r\bigl(S(h_r'')\,
        S(h^{(1)})\,
        h_1'''\bigr)\,
        f_q\bigl(S(h_q'')\,
        h^{(2)}
        h_{q+1}'''\bigr)
        \prod_{j=1}^{q-1}
        f_j\bigl(S(h_j'')\,
        h_{j+1}'''\bigr) \\
        & \relphantom{=}
        \ket{h_1'}\otimes
        \dots\otimes
        \ket{h_q'}\otimes
        \ket{h^{(3)}}\otimes
        \cdots
        \displaybreak[0]\\
        & =\!\!\sum_{(h_i)(h)}\!
        f_q\bigl(S(h_q'')\,
        h'''
        h_{q+1}'''\bigr)
        f_r\bigl(S(h_r'')\,
        S(h'')\,
        h_1'''\bigr)\,
        \prod_{j\neq q,r}\!
        f_j\bigl(S(h_j'')\,
        h_{j+1}'''\bigr) \\
        & \relphantom{=}
        \ket{h_1'}\otimes
        \dots\otimes
        \ket{h_q'}\otimes
        \ket{h'}\otimes
        \ket{h_{q+1}'}\otimes
        \dots\otimes
        \ket{h_r'} \\
        & =\norm{h}^{r+1}
        \ket{\psi_{s\cup s'}
          (f_1,\dots,f_r)}.
    \end{align*}
    Note that from the third line on we repeatedly applied
    Lemma~\ref{lem:Hopf_singlet1}.\qed
\end{proof}

Unfortunately, we do not know how \emph{arbitrary} hierarchy
states~$\ket{\psi_\Gamma^{A,B}}$
are affected by the associated face isometries~$I_F^{A,B}$ for a
general finite-dimensional Hopf $C^*$-algebra. However, for $H=\mathbb{C}G$
we can state the following and leave the proof to the reader.
\begin{lemma}
    Let $K\subset G$ a subgroup and $N\lhd G$ a normal subgroup.
    Furthermore let $f_i\in\mathbb{C}^{G/N}$. Then
    \begin{equation}
        I_F^{K,N}
        \ket{\psi_p^{K,N}
          (f_1,\dots,f_r)}
        =\sqrt{\frac{\abs{KN}}{\abs{G}}}
        \ket{\psi_{p\cup p'}^{K,N}
          (f_1,\dots,f_r)}.
    \end{equation}
\end{lemma}

The above lemma shows that unless one draws the quantum state from the
top of the hierarchy one needs to include an additional factor of
$\sqrt{\abs{KN}/\abs{G}}$ for each face to insure proper \emph{relative}
normalization of the Hopf tensor network states. In order to fix the
normalization \emph{absolutely} we may calculate the norm of a single
(preferably small) Hopf tensor network state on a given surface. Since
all other states on the same surface can be reached from this initial
one via isometries their norm will be determined automatically. It turns
out that the absolute normalization factor is entirely a property of the
surface the Hopf tensor network state is embedded in.
For example, absorbing the additional factor per face into the
definition of the Hopf tensor network states one has for arbitrary
graphs~$\Gamma$ on $S^2$
\begin{equation}
    \norm{\ket{\psi_\Gamma^{K,N}}}
    =\sqrt{\abs{KN/N}}
\end{equation}
while on $T^2$ one has
\begin{equation}
    \norm{\ket{\psi_\Gamma^{K,N}}}
    =\frac{1}{\sqrt{\abs{KN/N}}}\,
     \sqrt{\!\sum_{g\in KN/N}\!
           \abs{C_{KN/N}(g)}}.
\end{equation}
Here $C_G(g)$ denotes the centralizer of the element~$g$ in the
group~$G$. In order to keep the following discussion as general as
possible we will stick to relative normalization unless otherwise noted.

\subsection{Entanglement entropy}

Having defined the isometries~$I_F$ and $I_V$ we finally embark on an exact calculation of the
entanglement entropy for a simply connected region~$R$ on $S^2$. We focus
on the Hopf tensor network state~$\ket{\psi_\Gamma}$.

In the following we show how the inner product between two face
pieces of the Hopf tensor network state~$\ket{\psi_\Gamma}$ depends
on the boundary degrees of freedom.

\begin{proposition}
    \label{prop:overlap_face}
    Let $f_i,g_i\in X$ and $r=\abs{\partial R}$. Then
    \begin{equation}
        \braket{\psi_p(f_1,\dots,f_r)}
        {\psi_p(g_1,\dots,g_r)}
        =\abs{H}^{-r}
        \sum_{(h_i)}
        \phi(h_1'\cdots
        h_r')
        \prod_{j=1}^r
        (g_jf_j^*)(h_j'').
    \end{equation}
\end{proposition}
\begin{proof}
    Let $\lambda_i=h_i=h$ the Haar integral of $H$. Then
    from~\eqref{eq:face2} we get
    \begin{align*}
        \braket{\psi_p(f_1,\dots,f_r)}
        {\psi_p(g_1,\dots,g_r)}
        & =\!\sum_{(h_i)(\lambda_i)}\!\!
        \overline{\phi(h_r'''\cdots
          h_1''')}\,
        \phi(\lambda_r'''\cdots
        \lambda_1''') \\
        & \relphantom{=}
        \prod_{j=1}^r
        \overline{f_j\bigl(S(h_j'')\bigr)}\,
        g_j\bigl(S(\lambda_j'')\bigr)\,
        \braket{h_j'}{\lambda_j'}
        \displaybreak[0]\\
        & =\!\sum_{(h_i)(\lambda_i)}\!\!
        \phi(h_1'''\cdots
        h_r''')\,
        \phi(\lambda_r'''\cdots
        \lambda_1''') \\
        & \relphantom{=}
        \prod_{j=1}^r
        f_j^*(h_j'')\,
        g_j\bigl(S(\lambda_j'')\bigr)\,
        \phi(h_j'\lambda_j')
    \end{align*}
    where we used~\eqref{eq:inner_product}, the properties of $h$ and
    the involution~$*$ as well as~\eqref{eq:dual_star} for the second
    line. Employing Lemma~\ref{lem:Hopf_singlet4} we can simplify this
    to
    \begin{align*}
        \braket{\psi_p(f_1,\dots,f_r)}
        {\psi_p(g_1,\dots,g_r)}
        & =\abs{H}^{-1}\!
        \sum_{(h_i)(\lambda_i)}\!\!
        \phi(h_1'''\cdots
        h_r''')\,
        \phi\bigl(\lambda_r'''\cdots
        \lambda_2'''\,
        S(h_1')\bigr) \\
        & \relphantom{=}
        (g_1f_1^*)(h_1'')
        \prod_{j=2}^r
        f_j^*(h_j'')\,
        g_j\bigl(S(\lambda_j'')\bigr)\,
        \phi(h_j'\lambda_j')
        \displaybreak[0]\\
        & =\abs{H}^{-r}
        \sum_{(h_i)}
        \phi(h_1'''\cdots
        h_r''')\,
        \phi\bigl(S(h_1'\cdots
        h_r')\bigr) \\
        & \relphantom{=}
        \prod_{j=1}^r
        (g_jf_j^*)(h_j'').
    \end{align*}
    Finally, the claim follows from $\phi^2=\phi$.\qed
\end{proof}

Analogously, we are interested in the inner product between vertex
pieces of the same Hopf tensor network state.

\begin{proposition}
    \label{prop:overlap_vertex}
    Let $f_i,g_i\in H^*$ and $r=\abs{\partial R}$. Then
    \begin{equation}
        \braket{\psi_s(f_1,\dots,f_r)}
        {\psi_s(g_1,\dots,g_r)}
        =\abs{H}^{-r}
        \sum_{(h_i)}
        \prod_{j=1}^r
        f_j^*\bigl(h_j^{(3)}
        S(h_{j+1}^{(4)})\bigr)\,
        g_j\bigl(h_j^{(2)}
        S(h_{j+1}^{(1)})\bigr).
    \end{equation}
\end{proposition}
\begin{proof}
    Let again $\lambda_i=h_i=h$ the Haar integral of $H$.
    From~\eqref{eq:vertex2} one has
    \begin{equation*}
        \begin{split}
            \braket{\psi_s(f_1,\dots,f_r)}
            {\psi_s(g_1,\dots,g_r)}
            & =\sum_{(h_i)}
            \sum_{(\lambda_i)}
            \prod_{j=1}^r
            \overline{f_j\bigl(S(h_j'')\,
              h_{j+1}'''\bigr)}\,
            g_j\bigl(S(\lambda_j'')\,
            \lambda_{j+1}'''\bigr)\,
            \braket{h_j'}{\lambda_j'} \\
            & =\sum_{(h_i)}
            \sum_{(\lambda_i)}
            \prod_{j=1}^r
            f_j^*\bigl(h_j''\,
            S(h_{j+1}''')\bigr)\,
            g_j\bigl(S(\lambda_j'')\,
            \lambda_{j+1}'''\bigr)\,
            \phi(h_j'\lambda_j').
        \end{split}
    \end{equation*}
    Additionally, Lemma~\ref{lem:Hopf_singlet4} yields
    \begin{equation*}
        \sum_{(\lambda_i)}
        \prod_{j=1}^r
        \phi(a_j\lambda_j')\,
        g_j\bigl(S(\lambda_j'')\,
                 \lambda_{j+1}'''\bigr)
        =\abs{H}^{-r}
         \sum_{(a_i)}
         \prod_{j=1}^r
         g_j\bigl(a_j''\,
                  S(a_{j+1}')\bigr)
    \end{equation*}
    where $a_j\in H$ is arbitrary. Armed with this identity it is easy
    to verify the claim.\qed
\end{proof}

We now state the main result of this section.

\begin{theorem}
    Let $R\subset M$ a simply connected region and let $\Gamma_R
    \subset\Gamma$ the corresponding subgraph. Then for any $\alpha\geq0$
    the state~$\ket{\psi_\Gamma}$ has the Rényi entanglement entropies
    \begin{equation}
        S_\alpha(\rho_R)
        =\abs{\partial R}
         \log\abs{H}
         -\log\abs{H}.
    \end{equation}
    In particular, the topological entanglement entropy reads
    \begin{equation}
        \gamma
        =\log\abs{H}.
    \end{equation}
\end{theorem}

\begin{proof}
    Without loss of generality we restrict to $M\simeq S^2$. Acting with
    the isometries~$I_F$ and $I_V$ on $\ket{\psi_\Gamma}$ we may wipe
    out the bulk of both $R$ and $\bar{R}$ completely and reduce
    $\Gamma$ to one of the graphs shown in
    Figure~\ref{fig:boundary_graphs} where we indicated several types of
    boundaries.

    \begin{figure}
        \centering
        \includegraphics{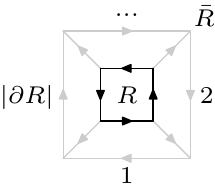}\hspace{2em}
        \includegraphics{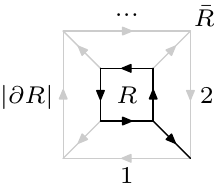}\hspace{2em}\dots\hspace{2em}
        \includegraphics{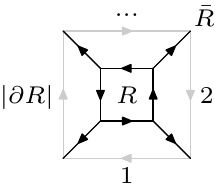}
        \caption{Boundary graphs for a simply connected region $R\subset
            S^2$. Here black edges constitute the corresponding
            subgraph~$\Gamma_R$ while grey edges belong to
            $\Gamma_{\bar{R}}$. These graphs can be reduced further
            (without touching the boundary) to minimal graphs.}
        \label{fig:boundary_graphs}
    \end{figure}

    We further assume the first type of boundary in the following. In
    this case we can continue the isometric reduction to the
    graph~$\Gamma_0$ given by
    \begin{equation*}
        \vcenter{\hbox{\includegraphics{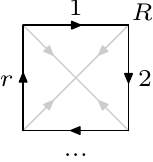}}}
    \end{equation*}
    and for the resulting Hopf tensor network state on this graph
    Proposition~\ref{prop:partitions} yields the natural splitting
    \begin{equation*}
        \label{eq:Schmidt}
        \ket{\psi_{\Gamma_0}}
        =\sum_{(\phi_i)}
         \ket{\psi_R(\phi_1',
                     \dots,
                     \phi_r')}\otimes
         \ket{\psi_{\bar{R}}(\phi_1'',
                             \dots,
                             \phi_r'')}
    \end{equation*}
    where the states
    \begin{align*}
        \ket{\psi_R
             (f_1,\dots,f_r)}
        & =\norm{h}^{-r}\,
           \norm{\phi}^{-1}
           \ket{\psi_p
                (f_1,\dots,f_r)} \\
           \ket{\psi_{\bar{R}}
                (f_1,\dots,f_r)}
        & =\norm{h}^{-r}
           \ket{\psi_s
                (f_r,\dots,f_1)}
    \end{align*}
    are identical to~\eqref{eq:face1} and~\eqref{eq:vertex1}
    respectively up to normalization. Note that $r=\abs{\partial R}$ denotes the length of
    the boundary.

    We determine the reduced density operator $\rho_R
    =N^{-1}\widetilde{\rho}_R$ of the region~$R$ as follows. Let $\phi_i
    =\varphi_i=\phi$ the Haar integral of $X$ and $h_i=\lambda_i
    =\kappa_i=h$ the Haar integral of $H$. Ignoring the normalization
    factor~$N$ for a moment one has
    \begin{align*}
        \widetilde{\rho}_R
        & =\tr_{\bar{R}}\,
           (\ket{\psi_{\Gamma_0}}
           \bra{\psi_{\Gamma_0}}) \\
        & =\!\!\sum_{(\phi_i)(\varphi_i)}\!
           \ket{\psi_R
                (\phi_1',
                 \dots,
                 \phi_r')}
           \bra{\psi_R
                (\varphi_1',
                 \dots,
                 \varphi_r')}\cdot
           \braket{\psi_{\bar{R}}
                   (\varphi_1'',
                    \dots,
                    \varphi_r'')}
                  {\psi_{\bar{R}}
                   (\phi_1'',
                    \dots,
                    \phi_r'')}
           \displaybreak[0]\\
        & =\abs{H}^{r+1}\!\!
           \sum_{(\phi_i)(\varphi_i)}\!
           \ket{\psi_p
                (\phi_1',
                \dots,
                \phi_r')}
           \bra{\psi_p
                (\varphi_1',
                 \dots,
                 \varphi_r')} \\
        & \relphantom{=}
           \sum_{(h_i)}
           \prod_{j=1}^r
           (\varphi_{r+1-j}'')^*
           \bigl(h_j^{(3)}
                 S(h_{j+1}^{(4)})\bigr)\,
           \phi_{r+1-j}''
           \bigl(h_j^{(2)}
                 S(h_{j+1}^{(1)})\bigr)
           \displaybreak[0]\\
        & =\abs{H}^{r+1}\!\!
           \sum_{(\phi_i)(\varphi_i)}
           \biggl(\sum_{(\lambda_i)}
                  \phi(\lambda_r'''\cdots
                       \lambda_1''')
                  \prod_{k=1}^r
                  \phi_k'\bigl(S(\lambda_k'')\bigr)
                  \ket{\lambda_1',
                       \dots,
                       \lambda_r'}\biggr) \\
        & \relphantom{=}
           \biggl(\sum_{(\kappa_i)}
                  \overline{\phi(\kappa_r'''\cdots
                                 \kappa_1''')}
                  \prod_{l=1}^r
                  \overline{\varphi_l'\bigl(S(\kappa_l'')\bigr)}\,
                  \bra{\kappa_1',
                       \dots,
                       \kappa_r'}\biggr) \\
        & \relphantom{=}
           \sum_{(h_i)}
           \prod_{j=1}^r
           (\varphi_j'')^*\bigl(h_j^{(3)}
                                S(h_{j-1}^{(4)})\bigr)\,
           \phi_j''\bigl(h_j^{(2)}
                         S(h_{j-1}^{(1)})\bigr)
           \displaybreak[0]\\
        & =\abs{H}^{r+1}
           \sum_{(h_i)}
           \sum_{(\kappa_i)}
           \sum_{(\lambda_i)}
           \phi(\lambda_r'''\cdots
                \lambda_1''')\,
           \phi(\kappa_1'''\cdots
                \kappa_r''')
           \prod_{j=1}^r
           \phi\bigl(h_j^{(2)}
                     S(h_{j-1}^{(1)})\,
                     S(\lambda_j'')\bigr) \\
        & \relphantom{=}
           \phi\bigl(h_j^{(3)}
                     S(h_{j-1}^{(4)})\,
                     \kappa_j''\bigr)
           \ket{\lambda_1'}\otimes
           \dots\otimes
           \ket{\lambda_r'}\otimes
           \phi(\kappa_1'?)\otimes
           \dots\otimes
           \phi(\kappa_r'?)
    \end{align*}
    where we used Proposition~\ref{prop:overlap_vertex} in the third line.
    
    Furthermore, using Lemma~\ref{lem:Hopf_singlet6}
    we may derive
    \begin{equation*}
        \sum_{(\kappa_i)}
        \phi(\kappa_1'''\cdots
             \kappa_r''')
        \prod_{j=1}^r
        \phi(b_j
             \kappa_j'')\,
        \bigotimes_{k=1}^r
        \phi(\kappa_k'?)
        =\abs{H}^{-r}
         \sum_{(b_i)}
         \phi(b_r''\cdots
              b_1'')\,
         \bigotimes_{j=1}^r
         \phi\bigl(S(b_j')\,
                   ?\bigr).
    \end{equation*}
    With this under our belt we arrive at
    \begin{equation*}
        \begin{split}
            \widetilde{\rho}_R
            & =\abs{H}\!
               \sum_{(h_i)(\lambda_i)}\!
               \phi(\lambda_r'''\cdots
                    \lambda_1''')\,
               \phi\bigl(h_r^{(4)}
                         S(h_{r-1}^{(5)})\cdots
                         h_2^{(4)}
                         S(h_1^{(5)})\,
                         h_1^{(4)}
                         S(h_r^{(5)})\bigr) \\
            & \relphantom{=}
               \prod_{j=1}^r
               \phi\bigl(h_j^{(2)}\,
                         S(h_{j-1}^{(1)})\,
                         S(\lambda_j'')\bigr)
               \ket{\lambda_1',\dots,
                    \lambda_r'}
               \bigotimes_{k=1}^r
               \phi\bigl[S\bigl(h_k^{(3)}
                                S(h_{k-1}^{(6)})\bigr)\,
                         ?\bigr] \\
            & =\abs{H}\!
               \sum_{(h_i)(\lambda_i)}\!
               \phi(\lambda_r'''\cdots
                    \lambda_1''')
               \prod_{j=1}^r
               \phi\bigl(h_j^{(2)}
                         S(h_{j-1}^{(1)})\,
                         S(\lambda_j'')\bigr) \\
            & \relphantom{=}
               \ket{\lambda_1',
                    \dots,
                    \lambda_r'}
               \bigotimes_{k=1}^r
               \phi\bigl(h_{k-1}^{(4)}\,
                         S(h_k^{(3)})\,
                         ?\bigr).
        \end{split}
    \end{equation*}
    Again some more preparation, namely Lemma~\ref{lem:Hopf_singlet7}
    allows us to find
    \begin{equation*}
        \sum_{(\lambda_i)}
        \phi(\lambda_r'''\cdots
             \lambda_1''')
        \prod_{j=1}^r
        \phi(a_j
             \lambda_j'')\,
        \bigotimes_{k=1}^r
        \ket{\lambda_k'}
        =\abs{H}^{-r}
         \sum_{(a_j)}
         \phi(a_1''\cdots
              a_r'')\,
         \bigotimes_{j=1}^r
         \ket{S(a_j')}
    \end{equation*}
    and subsequently we have for the reduced density operator:
    \begin{equation*}
        \widetilde{\rho}_R
        =\abs{H}^{1-r}
         \sum_{(h_i)}
         \bigotimes_{j=1}^r
         \ket{h_j^{(2)}\,
              S(h_{j-1}^{(1)})}\,
         \bigotimes_{k=1}^r
         \phi\bigl(h_{k-1}^{(4)}\,
                   S(h_k^{(3)})\,
                   ?\bigr).
    \end{equation*}
    Finally, it can be shown that the reduced density operator of the
    region~$R$ takes the simple form
    \begin{equation}
        \begin{split}
            \widetilde{\rho}_R
            & =\abs{H}^{1-r}
               \sum_{(h_i)}
               \ket{h_1^{(1)}}\otimes
               \dots\otimes
               \ket{h_{r-1}^{(1)}}\otimes
               \ket{S(h_{r-1}^{(4)}\cdots
                    h_1^{(4)})} \\
            & \relphantom{=}
               {}\otimes
               \phi\bigl(S(h_1^{(2)})\,
                         ?\bigr)\otimes
               \dots\otimes
               \phi\bigl(S(h_{r-1}^{(2)})\,
                         ?\bigr)\otimes
               \phi(h_{r-1}^{(3)}\cdots
                    h_1^{(3)}
                    ?).
        \end{split}
    \end{equation}
    up to normalization. It is easy to see that
    \begin{equation*}
        \begin{split}
            \tr(\widetilde{\rho}_R)
            & =\abs{H}^{1-r}
               \sum_{(h_i)}
               \prod_{j=1}^{r-1}
               \phi\bigl(S(h_j^{(2)})\,
                         h_j^{(1)}\bigr)\,
               \phi\bigl(h_{r-1}^{(3)}\cdots
                         h_1^{(3)}
                         S(h_{r-1}^{(4)}\cdots
                         h_1^{(4)})\bigr) \\
            & =\abs{H}^{1-r}
        \end{split}
    \end{equation*}
    and hence we can fix the normalization by setting $N=\abs{H}^{1-r}$.

    Furthermore it is not difficult to show that $\rho_R$ is
    proportional to a projector. Indeed, consider
    \begin{equation*}
        \begin{split}
            \rho_R^2
            & =\!\!\sum_{(h_i)(\lambda_i)}
               \prod_{j=1}^{r-1}
               \phi\bigl(S(h_j^{(2)})\,
                         \lambda_j^{(1)}\bigr)\,
               \phi\bigl(h_{r-1}^{(3)}\cdots
                         h_1^{(3)}
                         S(\lambda_{r-1}^{(4)}\cdots
                           \lambda_1^{(4)})\bigr) \\
            & \relphantom{=}
               \ket{h_1^{(1)}}\otimes
               \dots\otimes
               \ket{h_{r-1}^{(1)}}\otimes
               \ket{S(h_{r-1}^{(4)}\cdots
                      h_1^{(4)})} \\
            & \relphantom{=}
               {}\otimes
               \phi\bigl(S(\lambda_1^{(2)})\,
                         ?\bigr)\otimes
               \dots\otimes
               \phi\bigl(S(\lambda_{r-1}^{(2)})\,
                         ?\bigr)\otimes
               \phi(\lambda_{r-1}^{(3)}\cdots
                    \lambda_1^{(3)}
                    ?) \\
            & =\abs{H}^{1-r}
               \rho_R.
        \end{split}
    \end{equation*}
    This means that the spectrum of $\rho_R$ is flat.

    Finally one can prove that the rank of the reduced density
    operator~$\rho_R$ is given by $\abs{H}^{r-1}$ hence each non-zero
    eigenvalue equals $\abs{H}^{1-r}$. Then the Rényi entropies read
    \begin{equation}
        S_\alpha(\rho_R)
        =\frac{1}{1-\alpha}
         \log\bigl(\tr(\rho_R^\alpha)\bigr)
        =\abs{\partial R}
         \log\abs{H}
         -\log\abs{H}.
    \end{equation}
    independently of $\alpha$.\qed
\end{proof}

\begin{remark}
    We would like to stress that any reduced density operator~$\rho_R$
    obtained from the states~$\ket{\psi_{K,N}}$ is proportional to a
    projector as has been noted before for the case of group
    algebras at the top of the hierarchy~\cite{Flammia:2009p2617}.
    In other words, the
    entanglement spectrum is flat and all Rényi entropies are equal.
\end{remark}

\bigskip
\noindent
    While this proof yields exactly the entanglement entropy one would
    expect it does leave open one question: what is the mechanism of
    this universal correction to the area law in these Hopf tensor network states?

    It turns out that one can learn the answer by analyzing the boundary
    configurations $(f_1,\dots,f_r)\subset X^r$ of either
    $\ket{\psi_R(f_1,\dots,f_r)}$ or
    $\ket{\psi_{\bar{R}}(f_1,\dots,f_r)}$ in a particular canonical
    basis which is natural to the problem. Namely, by a generalized
    Fourier construction~\cite{Buerschaper:2009p2627,Buerschaper:2010p3003} one can extend the grouplike
    elements~$\mathcal{G}(H)$
    of $H$ to a linear basis the dual of which we denote
    $\mathcal{B}\subset X$. This dual basis is essentially unique and
    therefore serves as our canonical choice. For example, if $H$ is a
    group algebra one trivially has $\mathcal{B}=\{\delta_g \mid g\in
    G\}$.

    Then the dependence of, say, the interior
    states~$\ket{\psi_R(f_1,\dots,f_r)}$ on the boundary configuration
    $(f_1,\dots,f_r)\subset\mathcal{B}^r$ is radically different if $H$
    is a group algebra, the dual of a group algebra or a non-trivial
    Hopf $C^*$-algebra like $H_8$. Indeed, one may show that in the
    first case all non-zero interior states are orthogonal and share the
    same norm. Furthermore an interior state with boundary
    configuration~$(\delta_{g_1},\dots,\delta_{g_r})$ vanishes precisely
    if $g_1\cdots g_r\neq e$. Hence the topological constraint is
    realized by simple group multiplication around the boundary. In the
    case when $H$ equals the dual of a group algebra, \emph{any}
    interior state~$\ket{\psi_R(f_1,\dots,f_r)}$ is different from zero,
    however, altogether they cease to be orthogonal to each other.
    Instead, they come in bunches of exactly $\abs{H}$ identical states
    each. This obviously defines an equivalence relation among different
    boundary configurations. It is only these equivalence classes which
    are orthogonal and need to be summed over if we want to
    turn~\eqref{eq:Schmidt} into a proper Schmidt decomposition. This is
    the second mechanism how the topological constraint is implemented
    in a quantum double model ground state. Finally, if $H$ is a
    nontrivial Hopf $C^*$-algebra like $H_8$ one will observe a
    combination of the two mechanisms just described.

\section{Discussion}
\label{sec:discussion}

First of all, we have developed a tensor network language, based on
the formalism of finite-dimensional Hopf $C^*$-algebras, which is both more flexible
and more natural than conventional descriptions for topological
states. We have given rules to evaluate tensor network diagrams by
means of tensor traces and to construct quantum states on the lattice,
as well as to perform operations related to the spatial lattice
structure, such as cutting and joining along region boundaries.

This tensor network language has been shown to directly lead to the
construction of well-known topologically ordered states, namely ground
states of Kitaev's quantum double models based on groups~\cite{Kitaev:2003}. But in addition, we also obtain the extension of
Kitaev's construction to Hopf algebras. This is physically relevant
because the quantum double models based on Hopf algebras is the
smallest class of models allowing for a non-Abelian electric-magnetic
duality in topological models~\cite{Buerschaper:2010p3003}. All these states are
written, in a basis-independent way, in the form of tensor networks involving the intrinsic Hopf
$C^*$-algebra structure \emph{only}.

Relaxing this latter property leads us to a \emph{hierarchy} of states
defined from different subalgebras of a given Hopf $C^*$-algebra.
We study the classes obtained from group algebras, that is, Kitaev's
original models, and show that the hierarchy arises from the mechanism
of condensation of topological charges~\cite{Bais:2002p1647,Bais:2003p1648}, and so we
conjecture that this mechanism can be described in general in our
language.

Furthermore the hierarchy states can be regarded as ground states
of certain frustration-free Hamiltonians beyond the $\mathrm{D}(H)$-model. These
are obtained from Theorem~\ref{thm:model} by replacing the Haar integrals
$h\in H$ and $\phi\in X$ with $h_A\in A$ and $\phi_B\in B$ respectively
where $A\subset H$ and $B\subset X$ are Hopf $C^*$-subalgebras.

Not least, we have established isometric mappings defining
\emph{entanglement renormalization}~\cite{Vidal:2007p2581,Vidal:2008p2582} extending the
work of \cite{Aguado:2008p518} for the states at the top of the hierarchy.
This is a systematic procedure to thin out degrees of freedom keeping
the topological nature of the states, and hence their nonlocal
properties, intact.  Our computation of the topological entanglement
entropy is an application of this general scheme.

From a more philosophical viewpoint, our graphical language brings
tensor networks closer to algebraic sources; and, in particular, it
highlights the role of Hopf $C^*$-algebras as natural spaces for
transformations and symmetries of quantum many-body systems.  

We also give a broader perspective of Kitaev's original quantum double
construction, giving the mathematical background necessary but
emphasising the physical meaning of the different elements.  Hopf
$C^*$-algebras are shown to be natural and powerful tools in the
practical analysis of two-dimensional quantum problems; we hope that
the applications in this paper constitute a learn-by-doing
introduction to this structure for physicists.

This work opens several areas of research we are currently pursuing.
First of all, it is important to study the precise form in which
charge condensation appears in the Hopf quantum double models.  Next,
it is natural to ask whether all topological models can be cast in
quantum double model form by a suitable generalisation of our tensor
networks.  In~\cite{Buerschaper:2010p3003} we conjecture that the answer is yes, and
that the relevant algebraic structure is the class of weak Hopf
$C^*$-algebras; the details of the mathematical construction and
the corresponding tensor network are work in progress
\cite{Buerschaper:2010b}.

\appendix

\section{Hopf singlets}

We collect here some lemmas, with their proofs, that are used throughout
the text.

\begin{lemma}
    \label{lem:Hopf_singlet1}
    Let $h\in H$ the Haar integral, $a,b,c\in H$ and $f,g\in X$. Then one has
    \begin{equation}
        \sum_{(h)}
        f\bigl(S(h'')\,
               b\bigr)\,
        g(ch''')\,
        ah'
        =\!\sum_{(a)(h)}\!
         f\bigl(S(h'')\,
                a''b\bigr)\,
         g\bigl(c\,
                S(a')\,
                h'''\bigr)\,
         h'.
    \end{equation}
\end{lemma}
\begin{proof}
    It is clear that the statement will follow from
    \begin{equation}
        \label{eq:Hopf_singlet}
        \sum_{(h)}
        ah'\otimes
        S(h'')\otimes
        h'''
        =\!\sum_{(a)(h)}\!
         h'\otimes
         S(h'')\,
         a''\otimes
         S(a')\,
         h'''.
    \end{equation}
    Indeed, we have
    \begin{align*}
        \sum_{(h)}
        ah'\otimes
        S(h'')\otimes
        h'''
        & =\!\sum_{(a)(h)}\!
           a'h'\otimes
           S(h'')\,
           \epsilon(a'')\otimes
           h''' \\
        & =\!\sum_{(a)(h)}\!
           a'h'\otimes
           S(h'')\,
           S(a'')\,
           a'''\otimes
           h'''
           \displaybreak[0]\\
        & =\!\sum_{(a)(h)}\!
           a^{(1)}h'\otimes
           S(a^{(2)}h'')\,
           a^{(4)}\otimes
           \epsilon(a^{(3)})\,
           h'''
           \displaybreak[0]\\
        & =\!\sum_{(a)(h)}\!
           a^{(1)}h'\otimes
           S(a^{(2)}h'')\,
           a^{(5)}\otimes
           S(a^{(4)})\,
           a^{(3)}
           h'''
           \displaybreak[0]\\
        & =\sum_{(a)}
           \sum_{(a'h)}
           (a'h)'\otimes
           S\bigl((a'h)''\bigr)\,
           a'''\otimes
           S(a'')\,
           (a'h)'''
           \displaybreak[0]\\
        & =\!\sum_{(a)(h)}\!
           \epsilon(a')\,
           h'\otimes
           S(h'')\,
           a'''\otimes
           S(a'')\,
           h''' \\
        & =\!\sum_{(a)(h)}\!
           h'\otimes
           S(h'')\,
           a''\otimes
           S(a')\,
           h'''.
    \end{align*}
    \qed
\end{proof}

\begin{lemma}
    \label{lem:Hopf_singlet2}
    Let $a,b\in H$ and $h\in H$ as well as $\phi\in X$ the respective Haar integrals. Then
    \begin{equation}
        \sum_{(h)}
        \phi(ah''')\,
        \phi\bigl(b\,
        S(h'')\bigr)\,
        h'
        =\abs{H}^{-1}
        \sum_{(a)}
        \phi(ba')\,
        S(a'').
    \end{equation}
\end{lemma}
\begin{proof}
    Similarly to the proof of Lemma~\ref{lem:Hopf_singlet1} one can first
    show that
    \begin{equation*}
        \sum_{(h)}
        h'\otimes
        S(h'')\otimes
        ah'''
        =\!\sum_{(a)(h)}\!
         S(a'')\,
         h'\otimes
         S(h'')\,
         a'\otimes
         h'''
    \end{equation*}
    holds for any $a\in H$. Then we use the defining
    property~\eqref{eq:dual_integral} of
    the dual (Haar) integral:
    \begin{align*}
        \sum_{(h)}
        \phi(ah''')\,
        \phi\bigl(b\,S(h'')\bigr)\,
        h'
        & =\!\sum_{(a)(h)}\!
           \phi(h''')\,
           \phi\bigl(b\,
                     S(h'')\,
                     a'\bigr)\,
           S(a'')\,
           h' \\
        & =\!\sum_{(a)(h)}\!
           \phi\bigl[b\,
                     S\bigl(h''\,
                            \phi(h''')\bigr)\,
                     a'\bigr]\,
           S(a'')\,
           h'
           \displaybreak[0]\\
        & =\!\sum_{(a)(h)}\!
           \phi\bigl(b\,
                     S(1_H)\,
                     a'\bigr)\,
           S(a'')\,
           h'\,
           \phi(h'') \\
        & =\phi(h)
           \sum_{(a)}
           \phi(ba')\,
           S(a'').
    \end{align*}
    \qed
\end{proof}

\begin{lemma}
    \label{lem:Hopf_singlet4}
    Let $a,b,c,d\in H$ and $f,g\in X$. Furthermore let $h\in H$ and
    $\phi\in X$ the respective Haar integrals. Then
    \begin{equation}
        \sum_{(h)}
        \phi(ah')\,
        f\bigl(b\,
               S(h'')\,
               c\bigr)\,
        g(dh''')
        =\abs{H}^{-1}
         \sum_{(a)}
         f(ba''c)\,
         g\bigl(d\,
                S(a')\bigr).
    \end{equation}
\end{lemma}
\begin{proof}
    From~\eqref{eq:Hopf_singlet} we get
    \begin{align*}
        \sum_{(h)}
        \phi(ah')\,
        f\bigl(b\,
               S(h'')\,
               c\bigr)\,
        g(dh''')
        & =\!\sum_{(a)(h)}\!
           \phi(h')\,
           f\bigl(b\,
                  S(h'')\,
                  a''
                  c\bigr)\,
           g\bigl(d\,
                  S(a')\,
                  h'''\bigr) \\
        & =\!\sum_{(a)(h)}\!
           f(ba''c)\,
           g\bigl(d\,
                  S(a')\,
                  \phi(h')\,
                  h''\bigr) \\
        & =\phi(h)
           \sum_{(a)}
           f(ba''c)\,
           g\bigl(d\,
                  S(a')\bigr)
    \end{align*}
    which proves the claim.\qed
\end{proof}

\begin{lemma}
    \label{lem:Hopf_singlet6}
    Let $a,b,c,d\in H$ and $h\in H$ and $\phi\in X$ the respective Haar
    integrals. Then
    \begin{equation}
        \sum_{(h)}
        \phi(h'a)\,
        \phi(bh'')\,
        \phi(ch'''d)
        =\abs{H}^{-1}
         \sum_{(b)}
         \phi\bigl(S(b')\,
                   a\bigr)\,
         \phi\bigl(c\,
                   S(b'')\,
                   d\bigr).
    \end{equation}
\end{lemma}
\begin{proof}
    We first show that
    \begin{equation*}
        \sum_{(h)}
        h'\otimes
        bh''\otimes
        h'''
        =\!\sum_{(b)(h)}\!
         S(b')\,
         h'\otimes
         h''\otimes
         S(b'')\,
         h'''.
    \end{equation*}
    Indeed, we have
    \begin{align*}
        \sum_{(h)}
        h'\otimes
        bh''\otimes
        h'''
        & =\!\sum_{(b)(h)}\!
           \epsilon(b')\,
           h'\otimes
           b''h''\otimes
           h''' \\
        & =\!\sum_{(b)(h)}\!
           S(b')\,
           b''h'\otimes
           b'''h''\otimes
           h'''
           \displaybreak[0]\\
        & =\!\sum_{(b)(h)}\!
           S(b^{(1)})\,
           b^{(2)}h'\otimes
           b^{(3)}h''\otimes
           \epsilon(b^{(4)})\,
           h'''
           \displaybreak[0]\\
        & =\!\sum_{(b)(h)}\!
           S(b^{(1)})\,
           b^{(2)}h'\otimes
           b^{(3)}h''\otimes
           S(b^{(5)})\,
           b^{(4)}
           h'''
           \displaybreak[0]\\
        & =\sum_{(b)}
           \sum_{(b''h)}
           S(b')\,
           (b''h)'\otimes
           (b''h)''\otimes
           S(b''')\,
           (b''h)'''
           \displaybreak[0]\\
        & =\!\sum_{(b)(h)}\!
           S(b')\,
           \epsilon(b'')\,
           h'\otimes
           h''\otimes
           S(b''')\,
           h''' \\
        & =\!\sum_{(b)(h)}\!
           S(b')\,
           h'\otimes
           h''\otimes
           S(b'')\,
           h'''.
    \end{align*}
    Again by property~\eqref{eq:dual_integral} of the dual (Haar) integral
    we deduce:
    \begin{align*}
        \sum_{(h)}
        \phi(h'a)\,
        \phi(bh'')\,
        \phi(ch'''d)
        & =\!\sum_{(b)(h)}\!
           \phi\bigl(S(b')\,
                     h'a\bigr)\,
           \phi(h'')\,
           \phi\bigl(c\,
                     S(b'')\,
                     h'''d\bigr) \\
        & =\!\sum_{(b)(h)}\!
           \phi\bigl(S(b')\,
                     a\bigr)\,
           \phi\bigl(c\,
                     S(b'')\,
                     \phi(h')\,
                     h''d\bigr) \\
        & =\phi(h)
           \sum_{(b)}
           \phi\bigl(S(b')\,
                     a\bigr)\,
           \phi\bigl(c\,
                     S(b'')\,
                     d\bigr).
    \end{align*}
    \qed
\end{proof}

\begin{lemma}
    \label{lem:Hopf_singlet7}
    Let $a,b,c\in H$ and $h\in H$ as well as $\phi\in X$ the respective
    Haar integrals. Then
    \begin{equation}
        \sum_{(h)}
        \phi(h''a)\,
        \phi(bh'''c)\,
        h'
        =\abs{H}^{-1}
         \sum_{(a)}
         \phi\bigl(b\,
                   S(a'')\,
                   c\bigr)\,
         S(a').
    \end{equation}
\end{lemma}
\begin{proof}
    Along the same lines as the proof of Lemma~\ref{lem:Hopf_singlet6} one
    can show that
    \begin{equation*}
        \sum_{(h)}
        h'\otimes
        h''a\otimes
        h'''
        =\!\sum_{(a)(h)}\!
         h'\,
         S(a')\otimes
         h''\otimes
         h'''\,
         S(a'').
    \end{equation*}
    As before the rest then follows from the properties of the dual Haar
    integral.\qed
\end{proof}

\begin{acknowledgements}
    We gratefully acknowledge discussions with David Pérez-García,
    J.\,Ignacio Cirac, Norbert Schuch and Sergey Bravyi. For valuable
    comments on an earlier version of this manuscript we would like
    to thank Liang Kong and Joost Slingerland. M.\,A. thanks
    Jürgen Fuchs for an illuminating introduction to Hopf algebras.
\end{acknowledgements}

\bibliographystyle{plain}
\bibliography{refs}

\end{document}